 \documentclass[11pt]{article}
\usepackage{booktabs} % For formal tables
\usepackage[ruled]{algorithm2e} % For algorithms

\usepackage{graphicx}
\usepackage{tikz}
\usetikzlibrary{decorations.pathreplacing}
\usepackage[utf8]{inputenc}
\usepackage{fullpage,amsmath,amsthm,amsfonts,amssymb,enumerate}  %-- For some reason this increases the side margins. Commented out
\usepackage[capitalise]{cleveref}
\usepackage{subcaption}
\usepackage{soul}
\usepackage{physics}
\usepackage{url}
\SetAlFnt{\small}
\SetAlCapFnt{\small}
\SetAlCapNameFnt{\small}
\SetAlCapHSkip{0pt}
\IncMargin{-\parindent}

\newcommand{\hide}[1]{}

\newcommand{\treeroot}{\operatorname{root}}
\newcommand{\na}{\operatorname{na}}
\newcommand{\rh}{\operatorname{rh}}
\newcommand{\rl}{\operatorname{rl}}
\newcommand{\im}{\operatorname{im}}
\newcommand{\wna}{w_{\na}}
\newcommand{\wrh}{w_{\rh}}
\newcommand{\wrl}{w_{\rl}}
\newcommand{\wim}{w_{\im}}

\newcommand{\rem}{\scriptstyle\operatorname{rem}}
\newcommand{\sh}{\scriptscriptstyle\operatorname{H}}
\newcommand{\sd}{\scriptscriptstyle\operatorname{D}}

\newcommand{\sbm}{\scriptstyle\operatorname{bm}}
\newcommand{\sbf}{\scriptstyle\operatorname{bf}}
\newcommand{\ssm}{\scriptstyle\operatorname{sm}}
\newcommand{\ssf}{\scriptstyle\operatorname{sf}}

\newcommand{\event}{\mathcal{B}}

\newcommand{\prob}{\mathbb{P}}
\DeclareMathOperator{\E}{\mathbb{E}}

\newcommand{\taubar}{\overline{\tau}}
\newcommand{\psibar}{{\psi}}
\newcommand{\sw}[2]{s_{w_{#1}}(m_{#2})}
\newcommand{\sm}[2]{s_{m_{#1}}(w_{#2})}
\newcommand{\ratem}[1]{r_{m_{#1}}}
\newcommand{\ratew}[1]{r_{w_{#1}}}

\newcommand{\ptilde}{\tilde{p}}
\newcommand{\tildep}{\tilde{p}}
\newcommand{\qtilde}{\tilde{q}}
\newcommand{\eps}{\epsilon}

\newcommand{\cd}{{\mathcal C}_{\scriptscriptstyle D}}
\newcommand{\ch}{{\mathcal C}_{\scriptscriptstyle H}}

\newcommand{\Bhbm}{{\mathcal B}^{\sh}_{\sbm}}
\newcommand{\Bhbf}{{\mathcal B}^{\sh}_{\sbf}}
\newcommand{\Bhsm}{{\mathcal B}^{\sh}_{\ssm}}
\newcommand{\Bhsf}{{\mathcal B}^{\sh}_{\ssf}}
\newcommand{\Bdbm}{{\mathcal B}^{\sd}_{\sbm}}
\newcommand{\Bdbf}{{\mathcal B}^{\sd}_{\sbf}}
\newcommand{\Bdsm}{{\mathcal B}^{\sd}_{\ssm}}
\newcommand{\Bdsf}{{\mathcal B}^{\sd}_{\ssf}}
\newcommand{\Bcol}{{\mathcal B}_{\operatorname{col}}}

\newcommand{\Btau}{{\mathcal B}_{\tau}}
\newcommand{\Bdrem}{{\mathcal B}^{\sd}_{\rem}}

\newcommand{\notename}[2]{{\textcolor{red}{\footnotesize{\bf (#1:} {#2}{\bf ) }}}}

\newcommand{\rnote}[1]{{\notename{Richard Cole}{#1}}}
\newcommand{\pj}[1]{{\notename{Pranav}{#1}}}
\newtheorem{obs}{Observation} %[section]

\title %[Interview Selection in Large Random Markets]
{Distributed Interview Selection for Stable Matching in Large Random Markets}

\date{}

\newtheorem{observation}{Observation}[section]
\newtheorem{theorem}{Theorem}[section]
\newtheorem{definition}{Definition}[section]
\newtheorem{corollary}{Corollary}[theorem]
\newtheorem{lemma}[theorem]{Lemma}
\newtheorem{claim}[theorem]{Claim}

\hide{
\author{Richard Cole}
\affiliation{
  \institution{New York University}
  \streetaddress{251 Mercer St.}
  \city{New York}
  \state{New York}
  \postcode{10012}
  \country{USA}
}
\email{cole@cs.nyu.edu}

\author{Pranav Jangir}
\affiliation{
  \institution{New York University}
  \streetaddress{251 Mercer St.}
  \city{New York}
  \state{New York}
  \postcode{10012}
  \country{USA}
}
\email{pj2251@nyu.edu}
}

\author{Richard Cole\\
        New York University\\
        cole@cs.nyu.edu
        \and
        Pranav Jangir\\
        New York University\\
        pj2251@nyu.edu}

\begin{document}

\maketitle

\begin{abstract}
In real-world settings of the Deferred Acceptance stable matching algorithm, such as the American medical residency match (NRMP), school choice programs, and various national university entrance systems, candidates need to decide which programs to list. In many of these settings there is an initial phase of interviews or information gathering which affect the preferences on one or both sides.
\emph{We ask: which interviews should candidates seek?}
We study this question in a model, introduced by Lee (2016) and modified by Allman and Ashlagi (2023),
with preferences based on correlated cardinal utilities: Lee's utilities are based on common public ratings of each agent together with individual private adjustments; Allman and Ashlagi added additional adjustments that reflect the effect of the interviews.

We analyze two settings. The first is a residency setting with interviews.
The second is a school-choice setting; what changes here is that the school decisions are based solely on public information such as exam scores or lottery draws.
We describe a distributed, low-communication strategy for the doctors and students, which lead to non-match rates of $e^{(-\widetilde{O}(\sqrt{k}))}$ in the residency setting
and $e^{(-\widetilde{O}(k))}$ in the school-choice setting, where $k$ is the number of interviews per doctor in the first setting, and the number of proposals per student in the second setting; these bounds do not apply to the agents with the
lowest public ratings, the bottommost agents, who may not fare as well. We also obtain bounds on the expected utilities each non-bottommost agent obtains. 
These results are parameterized by the capacity of the hospital programs and schools. Larger capacities improve the outcome for the hospitals and schools, but don't significantly affect the outcomes of the doctors or students.
Finally, in the school choice setting we obtain an $\eps$-Nash type equilibrium for the students apart from the bottommost ones; importantly, the equilibrium holds regardless of the actions of the bottommost students. We also discuss to what extent this result extends to the residency setting.
Our results arise by substantially extending an analysis technique introduced in Agarwal and Cole (2023).
We believe the technique is applicable to a variety of stable matching scenarios beyond those analyzed in this paper. 
We complement our theoretical results with an experimental study that shows the asymptotic results hold for real-world values of $n$. 
\end{abstract}

% Optional
% Can remove later.
%\vspace{1cm}
%\setcounter{tocdepth}{2} % adjust to 2 if desired
%\tableofcontents

%\end{titlepage}

\hide{
Structure : \\
6. Women in cone matching rate. \\
6.5. Men in cone matching rate. \\
6.5 Mention that the women/men hypercone matching is dependent only on the private value of in hypercone men/women. \\
6.5 Corollary --- Women matching rate in men proposing DA. \\
7. Remove all the edge sets, still the same matching rate.. for both men and women. \\
7. Nash Equilibria. \\
8. Matching negative result when loss is bounded. \\ 
9. Appendix --- Capacities on one side. \\
}

\newpage

\section{Introduction}\label{sec::intro}

Consider a doctor applying for residency positions.
Which programs should she apply to? The very top ones?
Or those where she has a reasonable chance of admission (if these are not the same)?
And if the latter, how does she identify them?
We study this issue in settings where Gale and Shapley's deferred acceptance (DA) algorithm~\cite{GS62} is being used to match applicants to positions.
It is well-known that in DA the optimal strategy for the proposing side is to list their choices in order of preference. However, this does not indicate which choices one should list.

Consider the US residency match, which matches prospective residents to positions in hospital programs, which we refer to as hospitals henceforth. The most well-known part of this process is the National Residency Matching Program (NRMP), which runs the doctor-proposing DA to compute the matching (modulo a few details, such as allowing couples to submit joint preferences). Before this happens, both sides determine their preferences as follows. First, doctors apply to hospitals for interviews. Hospitals grant some of the requested interviews. Following the interviews, hospitals and doctors create their preference lists. It is well understood that if there is no interview a hospital will not list the doctor, so the only hospitals worth listing are those at which an interview occurs.

As time is limited, doctors and hospitals can afford only a relatively small number of interviews compared to the size of the pool, as least for larger fields such as pediatrics. Were doctors to all seek interviews at their favorite hospitals, given that preferences are likely to be correlated, this would result in many programs having too few applications, leading to a poor match rate. Clearly, doctors should apply to places where they have some reasonable chance of receiving a position. Indeed, the advice from the NRMP is that doctors should be ``realistic''. But can we make this more precise?

Let's also consider school choice programs in large school districts such as the one in New York City, where the matching is computed using the DA algorithm. This setting differs in that the schools preferences are largely public or random: they depend on a variety of factors, depending on the school, such as gpa, living in a designated zone, and random selection if needed. Similarly, university entrance in many countries is based solely on students performance in a national exam, e.g.\ the gaokao in China or the Joint Entrance Examination in India. In all these settings, students need to decide which schools to list.
% We can think of this as being informed by an information gathering phase, even if in reality this is largely based on word-of-mouth. 
Here too, students are well-advised to be realistic.

To model these settings, we use a model proposed by Allman and Ashlagi~\cite{AA23}, which extends a model introduced by Lee~\cite{Lee16}.

To keep the analysis tractable, we simplify the residency match setting and suppose that each doctor requests exactly $k$ interviews, and the hospitals grant all the requested interviews. Later in the paper, we comment on why we believe it may be feasible to extend our analysis to the more general setting.
Similarly, for the school choice setting, we suppose that each student gathers information to inform their choices (which we will call interviewing) at exactly $k$ schools.

In the residency setting, both sides create preference lists for the agents they interviewed and no others. In the school choice setting, students only list schools at which they interviewed; schools list all students. The doctor-proposing or student-proposing DA algorithm is then run using these preference lists as its input.
Another salient aspect of both these settings is that the hospitals and schools each have multiple positions, which for simplicity we assume is the same size $\kappa$ for all programs.

We obtain the following results.
In our utility-based model of preferences, 
\begin{itemize}
\item  
The probability of failing to match is $e^{-\widetilde{\Theta}(\sqrt{k})}$  in the residency setting and $e^{-\widetilde{\Theta}(k)}$ in the school-choice setting,\footnote{The $\widetilde{O}$ notation suppresses logarithmic terms.} apart possibly for a small bottommost tier of agents. 
\item We obtain bounds on the expected utilities the agents achieve.
\item 
In the school choice setting, this yields an $\epsilon$-Nash style equilibrium: for all but the bottommost $O(\eps)$ students: if the students follow the recommended strategy, deviation yields at most an expected $\epsilon=O(\ln k/k)$ improvement in utility, and this bound is unaffected by the actions of the bottommost students. Note that, by definition, in this setting the schools are non-strategic.
\item 
In the residency setting, if we assume the hospitals are non-strategic, again there is an $\eps$-Nash style equilibrium; here, a bound of $\eps=O(\sqrt{\ln k/k})$ applies to the non-bottommost doctors, and this is unaffected by the actions of the bottommost doctors.
\end{itemize}
Loosely speaking, these results are achieved if each applicant (student or doctor) identifies a personal range from her safe to her stretch choices, and then selects $k$ well-distributed favorite choices from this range; then, following interviews at these $k$ choices, the applicant lists these choices in her post-interview order of preference.

\subsection{Related Work}

Knuth \cite{Kn76,Kn96} initiated the analysis of randomized preference lists, considering lists that were uniform random permutations. Pittel provided definitive analyses of this setting~\cite{Pit89,Pit92}. Immorlica and Madhian~\cite{IM15} introduced a generalization of this model, which they called the \emph{popularity model};  it enables certain correlations among different agents' preferences.
In their model, the first side, men, can have arbitrary preferences; on the second side, women, preferences are selected by weighted random choices, the weights representing the ``popularity'' of the different choices. These results were further extended by Kojima and Parthak
in~\cite{KP09}.

The popularity model does not capture behavior in settings 
where bounds on the number of proposals lead to proposals
being made to plausible partners, i.e.\ partners with
whom one has a realistic chance of matching.
One way to capture such settings is by way of tiers~\cite{ABKS19},
also known as block correlation~\cite{CKN13}:
agents on each side are partitioned into tiers, with all agents in a higher tier preferred to agents in a lower tier, and with uniformly random preferences within a tier. Tiers on the two sides may have different sizes.
If we assign tiers successive intervals of ranks equal to their size,
then, in any stable matching, the only matches will be between agents in tiers whose rank intervals overlap.

A more nuanced way of achieving these types of preferences 
bases agent preferences on cardinal utilities; for each side, these utilities are functions of an underlying common assessment of the other side, together with idiosyncratic individual adjustments for the agents on the other side.
These include the separable utilities defined by Ashlagi, Braverman, Kanoria and Shi~\cite{ABKS19},
and another class of utilities
introduced by Lee~\cite{Lee16}. This last model will be the focus of our study. 

% In Lee's model, utilities are a function of a common or public assessment and idiosynchratic individual assessments. 
To make this more concrete, we review a simple special case of Lee's model, the \emph{linear
separable model}. Suppose that there are $n$ doctors and $n$ hospitals seeking to match with each other.
Each doctor $d$ has a public rating $r(d)$, a uniform random draw from $[0,1)$ \footnote{A natural way of obtaining such a number is by repeated sampling of a 0/1 choice, stopping when the resulting number, viewed as being of the form $0.***\cdots$, is sufficiently precise; this will always be in the range $[0,1)$.}. These ratings can be viewed
as the hospitals' joint common assessment of the doctors. In addition, each doctor $d$ has an individual adjustment $v(d,h)$  for each hospital $h$, which we call a private value;
again, this is a uniform random draw from $[0,1)$. All the draws are independent. Doctor $d$'s utility for hospital $h$ is given
by $r(h) + v(d,h)$; her full preference list has the hospitals in decreasing utility order. The hospital utilities are defined similarly: $h$ has utility $r(d)+v(h,d)$ for doctor $d$, where $v(h,d)$ is her private value for $d$, also a uniform random draw from $[0,1)$.

Lee showed that for $\eps>0$, for large enough $n$, with failure probability at most $\exp(-1/\eps^2)$, all but an $\eps$ fraction of the pairings in a matching would provide the doctor $d$ utility of at least $r(d)+1-\eps$ and hospital $h$ utility of at least $r(h)+1-\eps$. We call $r(d)+1$ doctor $d$'s benchmark and $r(h)+1$ hospital $h$'s benchmark. Informally speaking, most agents achieve utility which is at least their benchmark minus $\eps$. Lee showed $\eps$ could be as small as $O(n^{-1/4})$. The main technical tool in his work was results about the maximum size of bi-cliques in random bipartite graphs.

Later, Agarwal and Cole~\cite{AC23} reanalyzed this setting using a different methodology which analyzed the DA algorithm itself. They showed that, for any constant $c>0$, with failure probability at most $n^{-c}$, the only agents that failed to achieve their benchmark minus $\eps=\widetilde{O}(n^{-1/3})$ were those with public rating at most $\widetilde{O}(n^{-1/3})$. \footnote{The tight log term they obtained was $\ln^{2/3}n$.}
Their main technical tool was to use truncated runs of DA, which they called the double-cut DA, and a specialized stochastic dominance argument. In this work, we generalize these tools, which enables us to analyze the interview setting.

The first analysis of the interview process in Lee's model was by Echenique et al.~\cite{EGWY22}.
They supposed the interview pairs were determined by an initial computation that determined a $k$--$k'$ stable match (meaning that doctors obtained up to $k$ matches and hospitals up to $k'$ matches; setting $k'=\kappa k$ corresponds to each hospital having capacity $\kappa$). They then ran the doctor-proposing DA on these edges using the same utilities; i.e.\ there was no post-interview alteration of the utilities. The goal of the work was to see if one could explain why doctors mostly seem to obtain their declared first or second preference matches in the NRMP matching. By means of extensive simulations, they showed that with $k=k'=5$, match rates and doctor outcomes similar to the NRMP happened. 
% They also proved a theoretical bound.
%, but its conclusions were not as clear-cut as the simulations.

Allman and Ashlagi~\cite{AA23} extended Lee's model by adding a third component to the utilities, the interview value, again a random draw from $[0,1)$. In this model, interviews were chosen based on the utility given by the first two components. Then, following the interview, the interview values were added to the utilities, and these new utilities were used to provide preferences over the interviewed partners. We adopt this model in our paper.
Unfortunately, a key proof in~\cite{AA23} turned out to be incorrect. Their analysis was based on the random graph approach used by Lee. In contrast, our results extend the techniques introduced in \cite{AC23}.

In other recent work, Ashlagi et al.\ \cite{ACRS25} sought to guarantee matches in the presence of interviews;
in tiered markets a bound of $O(\log^3 n)$ interviews is achieved, in contrast to the $O(1)$ interviews that is our focus.
Another recent line of work seeks to evaluate the use of signaling in stable matching settings \cite{ABKS19,AASY25}.
The impact of short lists on the quality of matches has also been studied previously \cite{Arnosti15,KMQ21}.
Finally, the engineering of admissions systems for universities and schools has received considerable attention \cite{APR05,APR09,BCCKP19,CEEetal19,RLPC21}, which partly motivates our analysis of this case.

\paragraph{Roadmap}
In \Cref{sec::model} we define our model and then in \Cref{sec::results} we state our main results.
In \Cref{sec::double-cutDA} we review the double-cut DA, a class of truncated runs of the DA algorithm, which was introduced in \cite{AC23} and is generalized in the current work.
In \Cref{sec::proof-idio-thm} we outline the analysis of the residency setting, and then 
in \Cref{sec::student-school} we briefly explain how the analysis changes for the school choice setting.
In \Cref{sec::extensions}, we explain how our results extend to unequal sized sides (e.g.\ more doctors than hospitals) and discuss how to support generalizations of the utility functions; 
% then, in \Cref{sec::open-prob} 
in addition, we indicate why the analysis might be extended to the full interview setting as mentioned above and mention an unresolved issue. Next, in \Cref{experiments}, we complement our theoretical results with some experimental results, which indicate that the conclusions from our asymptotic bounds hold for real-world values of $n$, the number of agents. Finally, in \Cref{sec::conc}, we offer some concluding remarks. All omitted proofs can be found in the appendix.

\section{Model}\label{sec::model}

\subsection{The Residency Setting}

We suppose there are $n$ doctors applying to $n/\kappa$ hospitals,
each of which has a size $\kappa$ capacity, which is the number of doctors a single hospital can match with.
Each doctor will select $k$ hospitals to apply to. The matches will be made by running the
doctor-proposing DA. The doctor and hospitals will order their preference lists based on utilities, which we specify next.

Each hospital $h$ and doctor $d$ has a public rating, denoted by $r(h)$ and $r(d)$, respectively,
and both are uniform random draws from $[0,1)$. 
In addition, each doctor $d$ has a private (idiosyncratic) value for each hospital $h$, denoted by $v(d,h)$, again a uniform random draw from $[0,1)$.
Each doctor will select $k$ hospitals based on the hospital's public rating and her private values; we specify exactly how a little later. The doctor will then interview at these $k$ hospitals, and obtain a further interview value, denoted $\iota(d,h)$, also a uniform random draw from $[0,1)$.
A doctor $d$'s utility for a particular hospital $h$ is the sum of these three values: $r(h)+v(d,h)+\iota(d,h)$. She then lists the $k$ chosen hospitals in decreasing utility order.

A doctor $d$ selects hospitals according to the following rule.  First, $d$ only considers hospitals with public ratings in the range $[r(d)-a\alpha,r(d)+a\alpha)$, where $a$ and $\alpha$ are parameters we specify later.
We call this range $d$'s \emph{cone}, and denote it by $C(d)$.
Among these hospitals, she selects those with the $k$ highest private values.
Note that because the private values are uniform random draws, we can view this selection as a uniform random selection of $k$ hospitals from her cone.

The hospitals also have interview values for the doctors they interview, denoted by $\iota(h,d)$, again random draws from $[0,1)$. 
A hospital $h$'s utility for a doctor $d$ is the sum of two values: $r(d)+\iota(h,d)$, $d$'s public rating plus $h$'s interview value for $d$. Again, a hospital orders its preference list according to these utility values.
We do not provide the hospitals private values, for 
in computing their overall utility for a doctor $d$, if they had a private value, they would be adding both the private value and the interview value to $r(d)$, which we might as well treat as the interview value, though adding two linear terms yields a different distribution that a single linear term.
\footnote{If instead we had the two linear terms, the analysis becomes a bit more involved, but the bounds remain essentially unchanged.}

We discuss some generalizations of this model later.

\paragraph{The School Choice Setting}
To avoid an excess of notation, we will use the doctor-hospital terminology for this setting too.
Students correspond to doctors and schools to hospitals.
The one change is that every hospital (school) has the same preference list because they all have the same utility $r(d)$ for doctor (student) $d$, for every $d$.
The doctor utilities are unchanged.

\subsection{The Deferred Acceptance Algorithm for Many-to-One Matchings}

Let $D$ be a set of $n$ doctors and $H$ a set of $n/\kappa$ hospitals. Each doctor $d$ has an ordered list of hospitals that represents her preferences, i.e.\ if a hospital $h$ comes before a hospital $h'$ in $d$'s list, then $d$ would prefer matching with $h$ rather than $h'$. 
The position of a hospital $h$ in this list is called $d$'s ranking of $h$. Similarly each hospital $h$ has a ranking of the doctors. These lists need not be complete. \footnote{In general, hospitals could have preferences for sets of doctors rather than for individual doctors, but in the context of the NRPM where hospitals list individual doctors, it is not clear how to represent set-based preferences.} The stable matching task is to match up to $\kappa$ doctors with each hospital in such a way that no two agents prefer each other to their assigned partners. More formally:

\begin{definition}[Matching]
A matching is a mapping $\mu$ of the doctors in $D$ to the hospitals in $H$ or to the ``no-match'' option denoted by $\phi$,
such that at most $\kappa$ doctors are mapped to each hospital: $|\mu^{-1}(h)|\le \kappa$ for every $h\in H$.
\end{definition}

\begin{definition}[Blocking pair]
A matching $\mu$ has a blocking pair $(d,h)$ if and only if: %the following conditions are all met:
\begin{enumerate}
    \item $d$ and $h$ are not matched: $\mu(d)\neq h$.
    \item $d$ prefers $h$ to her current match $\mu(d)$.
    \item $h$ prefers $d$ to one of the doctors to whom she is currently matched, $\mu^{-1}(h)$. 
\end{enumerate}
\end{definition}

\begin{definition}[Stable matching]
A matching $\mu$ is stable if it has no blocking pair.
\end{definition}
\begin{algorithm}[t]
\SetAlgoNoLine
Initially, all the doctors and hospitals are unmatched.\\
\While{some doctor $d$ with a non-empty preference list is unmatched}{ 
   let $h$ be the first hospital on her preference list\;
   \If{$h$ has fewer than $\kappa$ matches currently (``is not fully matched'')}
      {tentatively match $d$ to $h$.}
      {\eIf{$h$ is currently fully matched, and $d'$ is her least favorite current match, and $h$ prefers $d$ to $d'$}
      {make $d'$ unmatched, remove $h$ from $d'$'s preference list, and tentatively match $d$ to $h$.}
      {remove $h$ from $d$'s preference list.}
    }
 }
\caption{\textsf{Doctor Proposing Many-to-One Deferred Acceptance (DA) Algorithm}}
\label{alg:DA}
\end{algorithm}

Gale and Shapley \cite{GS62} proposed the seminal deferred acceptance (DA) algorithm for the stable matching problem. We present the doctor-proposing DA algorithm (Algorithm $1$); the hospital-proposing DA is largely symmetric and is given in the appendix for completeness.
The following facts about the DA algorithm are well known and apply to the many-to-one setting we are considering. We state them here without proof and we shall use them freely in our analysis. 

\begin{obs}
\hspace{0.0001in}
\begin{enumerate}
    \item \label{obs::DA_terminates_in stable_match}
DA terminates and outputs a stable matching.
\item\label{obs::DA_invariant_of_proposal_order}
The stable matching generated by DA is independent of the order in which the unmatched agents on the proposing side are processed.
\item\label{obs::DA_optimality}
Doctor-proposing DA is doctor-optimal, i.e.\ each doctor is matched with the best partner she could be matched with in any stable matching.
\item \label{obs::DA_pessimality}
Doctor-proposing DA is hospital-pessimal, i.e.\ each hospital is matched with the worst set of partners it could be matched with in any stable matching.
\end{enumerate}
\end{obs}

\hide{
\subsection{Strategy}

In this section, we formally describe the constant proposal strategy for the women.

\begin{definition}
    Let $m_i$ be the man alligned with $w_i$ and let $\tau = \tfrac{k}{4a\alpha n}$ be a threshold. Let $\widetilde{M}_i$ denote the set of men $m_j$ with public rating in the range $[\ratem{i} - a\alpha, \ratem{i} + a\alpha]$ and private score for $w_i$, $\sw{i}{j} \geq 1 - \tau$.
\end{definition}

Note that the since the expected number of men in $w_i$'s cone is equal to $2a\alpha n$, the expected number of men in $\widetilde{M}_i$ is $\tfrac{k}{2}$.

\begin{definition}[Constant proposal strategy] \label{def::cone-random-strategy}
    If $\abs{\widetilde{M}_i} \leq k$, then women propose to men in $\widetilde{M}_i$. If $\abs{\widetilde{M}_i} > k$, then women propose to the top $k$ men from $\widetilde{M}_i$ based on their public and private score.
\end{definition}

We say that women with $\abs{\widetilde{M}_i} > k$ \emph{trim} their edges and select the top $k$ edges from $\widetilde{M}_i$. Following lemma upper bounds the number of women who \emph{trim} their edges.

\begin{lemma}
    Let $\event_{\tau}$ be the event that the number of women with $\widetilde{M}_i > k$ is more than $2ne^{-K/3}$. $\event_{\tau}$ occurs with probability at most $\exp(-\tfrac{n}{3}e^{-K/3})$. 
\end{lemma}

We will show that following this strategy leads to an $\epsilon$-Bayes Nash equilibrium.
}

\section{Results}\label{sec::results}

In the school-choice setting, the outcome, when running the student-proposing DA, is an $\eps$-Nash equilibrium for all but the bottommost students (bottommmost w.r.t.\ their public ratings), with this property unaffected by whether or not the bottommost students follow the recommended strategy.

\begin{theorem}\label{thm::eps-nash-no-idio-large-cap}
    Suppose that there are $n$ students and $n/\kappa$ schools each with capacity $\kappa$, for $n$ sufficiently large. Suppose each student uses a cone with public utility range $[r(d)-a\alpha,r(d)+a\alpha)\cap[0,1)$, and lists $k$ schools selected as described in \Cref{sec::model}.
    Suppose further that $a = 5$, $\Theta(1)\le k\le \Theta(\ln n)$, and
    $\alpha =\frac{2(4a+1)\ln k}{k}$.
    Then, the outcome when following the specified application process is an $\eps$-Nash equilibrium for all students with public rating at least $a\alpha$, regardless of the strategy that the remaining students use, where $\eps=O(\ln k/k)$. In addition, for each student and school with public rating at least $a\alpha$, the non-match probabilities are respectively at most $\exp\big(-4k/(4a+1)\big)$ and $\exp\big(-3\kappa \ln k/2\big)$.
\end{theorem}

In the residency setting, the outcome, when running the doctor-proposing DA, is an $\eps$-Nash equilibrium for all but the bottommost doctors (bottommmost w.r.t.\ their public ratings), with this property unaffected by whether or not the bottommost doctors follow the recommended strategy, assuming the hospitals list all the doctors they interview in decreasing utility order.

\begin{theorem}\label{thm::eps-nash-idio-large-cap}
    Suppose that there are $n$ doctors and $n/\kappa$ hospitals each with capacity $\kappa$, for $n$ sufficiently large. Suppose each doctor uses a cone with public utility range $[r(d)-a\alpha,r(d)+a\alpha)\cap[0,1)$, and interviews with $k$ hospitals selected as described in \Cref{sec::model}.
    Suppose further that $a = 5$, $\Theta(1)\le k\le \Theta(\ln n)$, and
    $\alpha =\big(\frac{2(4a+1)\ln k}{k}\big)^{1/2}$.
    Then, the outcome when following the specified application process is an $\eps$-Nash equilibrium for all doctors with public rating at least $a\alpha$, regardless of the strategy that the remaining doctors use, where $\eps=O((\ln k/k)^{1/2})$. In addition, for each doctor and hospital with public rating at least $a\alpha$, the non-match probabilities are respectively at most $\exp\big(-4\big[k\ln k/(4a+1)\big]^{1/2}\big)$ and $\exp\big(-\frac 38\kappa \big[k\ln k/(4a+1)\big]^{1/2}\big)$.
    \end{theorem}

\section{The Analyses}

\subsection{Overview of the Double Cut DA}\label{sec::double-cutDA}

This methodology was introduced in \cite{AC23} to bound the number of proposals needed to obtain a stable matching with high probability, meaning probability $1-n^{-c}$, for any constant $c$ (the bound on the number of proposals increases as $O(\sqrt{c})$).

\begin{figure}[thb]
\begin{center}
\begin{tikzpicture}
\begin{scope}[>=latex]
\draw (-0.2,3) to (-0.2,1.55);
\node at (-0.2,1.485) {$\circ$};
\node at (-2.8,1.485) {hospital $h$};
\draw[dashed] (0.0,3.0) to (2.9,3.0);
\node at (1.5,2.8) {rating $1$};
\draw[->] (-1.9,1.485) to (-0.45,1.485);
% \node at (3,1.485) {$\circ$};
\node at (1.5,1.285) {rating $r(h)$};
\draw[dashed] (0.0,1.485) to (2.9,1.485);
\draw (-0.2,1.46) to (-0.2,0.5);
%\draw[<->] (-0.0,3) to (-0.0,0.5);
%\node at (1.1,2.0) {$[r(h)-\alpha,1)$};
\draw[<->] (-0.4,1.475) to (-0.4,0.5);
\node at (-1.7,1.0) {rating range $\alpha$};
\draw (3,3) to (3,1.55);
\draw (3,1.55) to (3,-1.0);
%\draw (3,-1) to (3,3);
\draw[<->] (3.2,1.475) to (3.2,-1);
% \node at (4.1,0.25) {$\ell$ doctors};
\node at (0, 3.3) {hospitals};
\node at (3, 3.3) {doctors};
\node at (-3.2,0.5) {cutoff $r(h) - \alpha$};
\draw[->] (-1.9,0.5) to (-0.45,0.5);
%\draw[<->] (5.0,3) to (5.0,-1);
%\node at (6.4,1.0) {$=WD[r_{\overline{w}_i},1]$};
%\node at (3,-1) {$\circ$};
%\node at (3.2,-1.3) {doctor $d=d_{i+\ell_i}$};
\node at (4.5,0.2) {rating range $a\alpha$};
\node at (4.5,-1.0) {cutoff $r(h)-a\alpha$};
\end{scope}
\end{tikzpicture}
\end{center}
\caption{The double-cut DA for hospital $h$; doctors with public ratings in $[r(h)-a\alpha,1)$ propose to hospitals with public ratings in $[r(h)-\alpha,1)$. The proposals made are all better for the proposing doctors than any possible proposals to hospitals in $[0,r(h)-\alpha)$.}
\label{fig::key_bounded}
\end{figure}
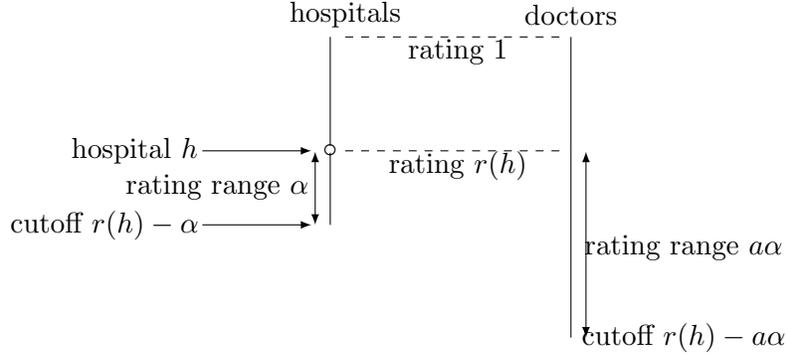

The double-cut DA comprises an initial portion of a run of the DA algorithm.
We will consider both runs in which the doctors propose and runs in which the hospitals propose. There will be a distinct double-cut DA for each hospital $h$ and each doctor $d$, other than those with the smallest public ratings.

Let's focus on the double-cut DA for hospital $h$. It comprises an initial portion of a run of the doctor-proposing DA. This run will have somewhat more doctors than hospitals participating, with the consequence that the hospitals are on the short side, which provides the potential for showing good outcomes for the hospitals, and for $h$ in particular. A key observation is that if we were to complete the run of the full doctor-proposing DA, this would only improve the outcome for $h$. Also, if we were to consider $h$'s outcome in the hospital proposing DA, it could only be better. In conclusion, whatever lower bound on $h$'s match holds in the run of the double-cut DA, also holds in the runs of the doctor and hospital proposing-DAs. These observations also apply to the symmetric double-cut DA's for doctors $d$.

% There are two sets of agents, doctors and hospitals.
%Suppose we want to show that a hospital $h$ receives a high-quality proposal with high probability.
% To keep the overview simple, let's suppose there are $n$ doctors and $n$ hospitals.
To illustrate this technique, we give an overview of the result in
% Our overview uses the setting from 
\cite{AC23}, which analyses the quality of matches in a setting with no interviews. There the doctors and hospitals obtain their utilities for potential partners by adding a private value for that partner to the partner's public rating, where values and ratings are uniform random draws from $[0,1)$. 
% \footnote{A natural way of obtaining such a number is by repeated sampling of a 0/1 choice, stopping when the number is unique; this will always be in the range $[0,1)$.}

To demonstrate the quality of $h$'s match, if any, we analyze a partial run of the doctor-proposing DA,
called the double-cut DA for $h$.
The participants are the hospitals with public rating in the range $[r(h)-\alpha,1)$ and the doctors with public rating in the range $[r(h)-a\alpha,1)$, where $a>1$ and $\alpha$ are suitable parameters (see \Cref{fig::key_bounded}).
Our choice of $\alpha$ will ensure that with high probability there are at least $\alpha n$ more doctors than hospitals participating in the partial run.

The run of the doctor-proposing DA is further restricted as follows. Each doctor in $h$'s cone, i.e.\ with public rating in $[r(h)-a\alpha,r(h)+a\alpha)$, stops making proposals as soon as she reaches a proposal providing her utility less than $r(h)+1-\alpha$ (she does not make this proposal). This means that she will not make any proposals to hospitals with public rating less than $r(h)-\alpha$.
She also stops if she is about to make a proposal to $h$; this will then be her final proposal, so long as it provides utility at least $r(h)+1-\alpha$. The doctors above $h$'s cone are unconstrained. 

In order to obtain useful bounds on the quality of $h$'s match one needs to show the following three bounds on the performance of the double-cut DA.
\begin{enumerate}
    \item There is a sufficiently large surplus $s$ of doctors who could either propose to $h$ or remain unmatched (depending on their random choices of private value). 
    \item If each of the $s$ surplus doctors $d$ were to propose independently to $h$ with probability $\alpha$ (because $d$ has a private value for $h$ providing her utility at least $r(h)+1-\alpha$), then $h$ would receive a high-quality match with the desired probability. %(Note that this is not what actually happens.)
    \item Inconveniently, the actions of the surplus doctors need not be independent. The final step is to show that the case of independent proposals yields a failure probability that is at least as large as the failure probability for the actual run of the double-cut DA. This argument was carried out by a methodology we name \emph{tree stochastic dominance}.\footnote{This name is being coined in this paper.}
\end{enumerate}

% An important property of the double-cut DA for $h$ is that the outcome for $h$ in the full run of the doctor-proposing DA is at least as good as its outcome in the double-cut DA (see \Cref{clm::double-cutDA-provides-bounds} below).
The full analysis uses analogous runs of a double-cut DA for doctors.
An immediate implication of this analysis is that with high probability, for each doctor $d$, her possible matches have public ratings in the cone $C(d)=[r(d)-a\alpha,r(d)+a\alpha)$, and similarly for the hospitals.

In the current work, we need to substantially modify Steps 1 and 3 above to achieve our new results.

In \cite{AC23}, with high probability---meaning polynomially small failure probability---all the doctors more than $a\alpha$ above $h$ in public rating will be matched in the current double-cut DA (this stems from a top-down induction we will explain later) and these matches are all with hospitals above $h$. This leaves roughly $a\alpha n$ hospitals and $2a\alpha n$ doctors unmatched, and furthermore all these doctors are in $h$'s cone. Consequently there is a surplus of $\alpha n$ doctors who either propose to $h$ or remain unmatched, and with high probability, among these doctors some will have high private values for $h$, and $h$ will have high private values for them, resulting in a high private value match on both sides.

\hide{
Step 1 depended on every agent ending up matched with high probability (except for some of the bottommost agents, meaning those with public rating less than $a\alpha$).
Then the surplus is simply the number of participating doctors minus the number of participating hospitals, and with high probability these numbers are close to their expected value, as shown by Chernoff bounds. On setting $a=3$, they obtain a surplus of size at least $\alpha n$.
}

In the current work, as we are limiting each doctor to making $k$ proposals, we no longer have a polynomially small probability of failing to match. In fact, we expect at least an $e^{-\Theta(k)}$ fraction of the doctors to fail to match. In computing the surplus, we will need to subtract the number of unmatched doctors whose public rating lies above $h$'s cone (because they fail to match hospitals above $h$, which therefore could be matched with some of doctors in the previous surplus, thereby reducing the surplus). As it turns out, we will need a high-probability bound on this number, and on the analogous numbers for the not-fully matched hospitals.

The difficulty is that while we can obtain an upper bound on the probability that a single hospital (or doctor) is unmatched, and hence an upper bound on the expected number of unmatched hospitals, we cannot apply a Chernoff bound, as the probabilities for the various hospitals need not be independent nor negatively correlated.
Instead, we modify the double-cut DA to lower bound the number of unmatched hospitals in a suitably sized interval $I$ of hospitals (the interval is given by public ratings). In this modification, we no longer use the full set of doctors in a range $[r(h)-a\alpha,1)$ as before, but a subset of these doctors chosen to ensure that the possible runs of the double-cut DA are close to uniform w.r.t.\ the hospitals in $I$. For then it turns out that the match/non-match probabilities for the hospitals in $I$ become negatively correlated, allowing us to apply a Chernoff bound.

We will also need to analyze an analogous double-cut DA for intervals of doctors. This is not completely symmetric for the hospitals can only make proposals on the edges the doctors selected, and different hospitals in general have differing numbers of available proposals.

The runs of the double-cut DA described in the previous two paragraphs are used to bound the number of unmatched hospitals and doctors. These bounds are used inductively. They are also used in the analysis of runs of the original double-cut DAs; these runs are used to lower bound the quality of the matches achieved by individual doctors and hospitals. The quality bounds in turn are used in the $\eps$-Nash analysis.

A further generalization provides each hospital a capacity $\kappa>1$. For simplicity, we suppose all hospitals have the same capacity, so there are now $n$ doctors and $n/\kappa$ hospitals. As $\kappa$ increases, the probability that a hospital is fully matched increases. Indeed, when $\kappa=\Omega(\ln n)$, all but the bottommost hospitals will be fully matched with failure probability at most $n^{-c}$.

In general, we can describe the double-cut DA as follows. Some collection $\cd$ of proposers (WLOG doctors) all use an initial portion of their preference lists in a run of doctor-proposing DA to propose to a collection $\ch$ of hospitals (these are the destinations of the edges in the above truncated preference lists). Suppose it produces an outcome $\mathcal O$ for the hospitals in $\ch$.

Now suppose we use a larger set of edges in a new run of doctor-proposing DA, which includes all the edges in the previous run. The new edges need not be consecutive in their preference lists and can also come from doctors outside $\cd$. Let ${\mathcal O}'$ be the new outcome.
The following claim is immediate.

\begin{claim}\label{clm::double-cutDA-provides-bounds}
    For every hospital in $\ch$, the outcome in ${\mathcal O}'$ is at least as good as the outcome in ${\mathcal O}$. An analogous claim applies to runs of the hospital-proposing DA. 
\end{claim}
\begin{proof}
    In the doctor-proposing DA, as more proposals are made, the outcome for the receiving hospitals only improves. And because the first run uses initial portions of the preference lists for the doctors in $\cd$, in the second run, these edges, for each preference list, are all proposed before any new edges. Therefore, outcome ${\mathcal O}$ can be achieved before any of the new edges are used, and hence using these new edges can only improve the outcome beyond ${\mathcal O}$.
\end{proof}

This immediately implies that the outcome for the hospitals in the full run of the doctor-proposing DA is at least as good as the outcome in the double-cut run, which is why the bounds in the double-cut run are useful.

\subsection{The Doctor-Hospital Setting: Sketch of the Analysis}
\label{sec::proof-idio-thm}

To carry out the analysis we discretize the setting, and then to obtain a result for the continuous setting
we take the limit as the discretization tends to 0, similarly to \cite{AC23}.

\begin{claim}
\label{clm::discrete-dist-idio}
For any $\delta>0$, there is a discrete utility space in which
the probability of having any differences
in the preference lists in the original continuous setting and in the discrete setting is at most $\delta$.
In addition, the public values of the doctors will all be distinct, as will the public values of the hospitals, and for each doctor, its private values will all be distinct, and this is achieved within the same overall $\delta$ bound on the failure probability.
Furthermore, the probability of each combination of utility choices will be the same.
\end{claim}
\hide{
\begin{proof}
$z\ge n^3/\delta\kappa$ will be an integer in our construction.

We partition the utility range $[0,3]$ into $3z$ intervals each of length $1/z$.
To do this, we round down each utility to the nearest integer multiple of $1/z$.
Note that for each doctor $d$, the probability that it has the same utility for doctors
$d$ and $d'$ is at most $1/z$: to see this, consider selecting the random utility components
one by one, with the interview value for $d'$ being last; this last choice will create the
utility equality with probability at most $1/z$.
Therefore the probability that a doctor has the same utility for two hospitals is at most
$\tfrac 12 \frac{n}{\kappa}\cdot\big(\frac{n}{\kappa}-1\big)/z$, and so the probability that any doctor has the same utility for two hospitals is at most
$n\cdot \frac{n}{\kappa}\big(\frac{n}{\kappa}-1\big)/2z\le n^3/2z\kappa^2$.

Similarly, the probability that any hospital has the same utility for two doctors is at most
$n^3/2z\kappa$.

Therefore, the settings in which a doctor has the same utility for any pair of hospitals or a hospital has the same utility for any pair of doctors
occur with probability at most $n^3/z\kappa\le \delta$.
\end{proof}
}

We will be repeatedly referring to agents, be it a doctor $d$ or a hospital $h$, with public ratings in an interval $I$. For brevity, we will often write $d\in I$ (or $h\in I$) rather than $r(d)\in I$ (or $r(h)\in I$).

The proof of the theorem will use four different double-cut DAs;
their outcomes are specified in the next four lemmas.
In these lemmas we suppose that $\alpha$ and $g-f$ are chosen so that $\alpha$ is an integer multiple of $g-f$. Also, as the proof proceeds, we will identify some low probability events, \emph{bad events}, for short, collectively denoted by ${\mathcal B}$; ${\mathcal B}$ will occur with probability at most
$1/n^c$, for any arbitrary constant $c>0$ (the bounds are a function of $c$ of course).

\begin{lemma}\label{lem::doctors-interval-non-match}
Let $\delta>0$ be an arbitrarily small constant.
Let $I=[f,g)$ be a public rating interval with $g\ge a\alpha$.
Suppose we run the (to be specified) hospital-proposing double-cut DA for interval $I$ using the discretized utilities
from \Cref{clm::discrete-dist-idio}.
Suppose the fraction of unmatched hospitals with public rating in $[g+(a+1)\alpha,1)$ is at most 
$3 \cdot\exp\big(-\alpha k/2(4a+1)\big)$.
Also suppose that event ${\mathcal B}$ does not occur.
Then, with failure probability at most $n^{-(c+2)}$,
the fraction of the doctors in $I$ who fail to match is less than $3\cdot \exp\big(-\frac{\alpha k}{2(4a+1)}\big)$,
if $n\ge2$, $a\ge 5$, $\alpha\cdot\exp\big(\frac{\alpha  k}{2(4a+1)}\big)\ge 3$, and the bounds on $g-f$ from \Cref{clm::hosp-prop-interval-constraints} hold.
\end{lemma}

A similar, but not identical, lemma (\Cref{lem::hospitals-interval-non-match}) for the doctor-proposing double cut DA for intervals $I$ of hospitals is deferred to the appendix.
Its proof is different in some non-trivial ways.
\hide{
\begin{lemma}\label{lem::hospitals-interval-non-match}
Let $\delta>0$ be an arbitrarily small constant.
Let $I=[f,g)$ be a public rating interval with $g\ge a\alpha$.
Suppose we run the (to be specified) doctor-proposing double-cut DA for interval $I$ using the discretized utilities from in \Cref{clm::discrete-dist-idio}.
Suppose the fraction of unmatched doctors with public rating in $[g+(a+1)\alpha,1)$ is at most 
$3\cdot\exp\big(-\alpha k/2(4a+1)\big)$.
Also suppose that event ${\mathcal B}$ does not occur.
Then, with failure probability at most $n^{-(c+2)}$,
the fraction of the hospitals in $I$ that fail to match
is less than $3\cdot \exp\big(-\frac{\alpha\kappa k}{2(4a+1)}\big)$,
if $n\ge2$, $a\ge 5$, $\alpha\cdot\exp\big(\frac{\alpha k}{2(4a+1)}\big)\ge 3$, and the bounds on $g-f$ from \Cref{clm::doc-prop-interval-constraints} hold.
\end{lemma}
}

%The proof of \Cref{lem::hospitals-interval-non-match} is quite similar to the proof of \Cref{lem::doctors-interval-non-match}, with the roles of the doctors and hospitals being switched. It will be deferred to the appendix.

These two lemmas  are used in a mutual induction with the same value of $\delta$ to justify their assumptions on the numbers of unmatched hospitals and doctors, respectively. They also justify the assumption on the numbers of unmatched hospitals and doctors in the next two lemmas. 
% In order to use these bounds we need to know that the double-cut DA for an interval $I$ of doctors, is over doctors and hospitals with public rating in $[g-\alpha,1)$ and a subset of $[g-a\alpha,1]$, respectively, and similarly for the double-cut DA for an interval $I$ of hospitals.

\begin{lemma}\label{lem::doctors-high-value-match-prob}
Let $\delta>0$ be an arbitrarily small constant.
    Let $d$ be a doctor with public rating at least $a\alpha$.
    Suppose we run the hospital-proposing double-cut DA for doctor $d$ using the discretized utilities
from \Cref{clm::discrete-dist-idio}.
    Suppose the fraction of unmatched hospitals with public rating in $[r(d)+(a+1)\alpha,1)$ is at most $3\exp\big(-\alpha k/2(4a+1)\big)$.
    Also, suppose that event ${\mathcal B}$ does not occur.
    Finally, suppose we run the hospital-proposing double-cut DA for $d$.
     Then the probability that $d$ fails to receive a match for which it has interview value $1-\psibar_d$ or larger is at most $\exp\big(-\frac{k\alpha\psibar_d}{4a+1}\big)$, if both $a\ge 5$ and $\alpha\cdot\exp\big(\alpha k/2(4a+1)\big)\ge 3$.
\end{lemma}
We will use two values of $\psibar_d$: we set it to 1 to obtain the non-match probability, and to
$\alpha$ to obtain a bound on the expected utility which we use in the $\eps$-Nash analysis.

A similar lemma (\Cref{lem::hosp-high-value-match-prob}) for the doctor-proposing double cut DA for hospital $h$ is deferred to the appendix. Again, its proof needs non-trivial changes.

\hide{
\begin{lemma}\label{lem::hosp-high-value-match-prob}
Let $\delta>0$ be an arbitrarily small constant.
    Let $h$ be a hospital with public rating at least $a\alpha$.
    Suppose we run the doctor-proposing double-cut DA for hospital $h$ using the discretized utilities as specified in \Cref{clm::discrete-dist-idio}.
    Suppose the fraction of unmatched doctors with public rating in $[r(d)+(a+1)\alpha,1)$ is at most $3\exp\big(-\alpha k/2(4a+1)\big)$.
    Also, suppose that the event ${\mathcal B}$ does not occur.
    Finally, suppose we run the doctor-proposing double-cut DA for $h$.
     Then the probability that $h$ fails to $\kappa$ matches for which it has interview value $1-\psibar_h$ or larger is at most $\exp\big(-\frac{3k\alpha\psibar_h}{}\big)$, if $a\ge 5$, $k\ge 4a+1$, and $\alpha\cdot\exp\big(\alpha k/2(4a+1)\big)\ge 3$.
\end{lemma}
}

% The proof of \Cref{lem::hosp-high-value-match-prob} is similar to the proof of \Cref{lem::doctors-high-value-match-prob}. We defer it to the appendix.

We finish up the argument by showing the bounds in these four lemmas apply
to the runs of the double-cut DA with the actual utilities drawn from a continuous range.

\begin{claim}\label{lem::cont-range-bound}
    The bounds of Lemmas \ref{lem::doctors-interval-non-match} and \ref{lem::doctors-high-value-match-prob} and the corresponding lemmas for the hospitals hold in the continuous setting too.
\end{claim}
\hide{
\begin{proof}
 For any given $\delta>0$, by \Cref{clm::discrete-dist-idio}, the actions in the continuous setting and the corresponding discrete setting can differ only when there is a difference in the preferences lists, which occurs with probability at most $\delta$.
 Therefore, by \Cref{lem::doctors-high-value-match-prob}, the probability that a doctor $d$ does not receive a proposal with interview value at least $1-\psibar_d$ is at most $\delta+ \exp\big(-k(\alpha-\delta)\psibar_d /2(4a+1)\big)$.
 As this holds for any $\delta>0$, the actual probability is at most $\exp\big(-\alpha k\psibar_d /2(4a+1)\big)$.
 Similar arguments apply to the bounds in the remaining three lemmas.
\end{proof}
}

\paragraph{Bad Events ${\mathcal B}$}
These bad events concern the number of agents with ratings in specified public rating intervals (for the sake of brevity, we will say agents in an interval).
We want to limit the analysis to the case where the number of such agents is reasonably close to its expectation. Our analysis will need to assume this for $O(n)$ intervals and unions of intervals.
We defer the details to the appendix.

\subsubsection{Proof of \Cref{lem::doctors-high-value-match-prob}}\label{sec::proof-doc-match-prob}

% We fix the utilities for all edges other than those incident on $d$.

We  use \Cref{lem::hospitals-interval-non-match} along with \Cref{clm::double-cutDA-provides-bounds} to bound the surplus, the number of not-fully matched hospitals in $[r(d)+\alpha,1)$, in the double-cut DA for $d$.
\begin{claim}\label{clm::matching-bounds-extend}
Let $I=[f,g)\subset [r(d)+(a+1)\alpha,1)$ be an interval of hospitals. Then every hospital that is fully matched in the doctor-proposing double-cut DA for $I$ will be fully matched in the hospital-proposing double-cut DA for $d$. 
\end{claim}
\begin{proof}
    The edges that are used in the double-cut DA for $I$ are to hospitals in the range $[g-\alpha,1)$ (as we will see later).
    The full set $F$ of edges from these hospitals all go to doctors in the range $[g-(a+1)\alpha,1)$.
    Therefore so long as $r(d)$ is in the range $[0,g-(a+1)\alpha)$, the double-cut DA for $d$ includes all of $F$, and additional edges forming a set $F'$.
    Now consider running the hospital-proposing DA with edge set $F'$. By \Cref{clm::double-cutDA-provides-bounds}, the outcome for the hospitals in $I$ on the run with $F'$ is at least as good as the run of the double-cut DA for $I$. But the doctor-proposing DA and the hospital-proposing DA with edge set $F'$ yield the same fully matched hospitals (and doctors). Therefore all the hospitals that are fully matched in the double-cut DA for $I$ are also fully matched in the double-cut DA for $d$.
\end{proof}

\begin{claim}\label{lem::surplus-hosp-prop-one-doc}
    Suppose the following good events  occur: there are at most $[1-r(d)+\frac 32\alpha]n$ doctors with public rating at least $r(d)-\alpha$ (the expected number is $[1-r(d)+\alpha ]n$ doctors), and there at least $[1-r(d)+(a-\frac32)\alpha] n/\kappa$ hospitals with public rating in the range $[r(d)-a\alpha,r(d)+a\alpha)\cup[r(d)+(a+1)\alpha,1))$ (the expected number is $[1-r(d)+(a-1)\alpha] n/\kappa$ hospitals).
    % $\Bdbm([r(d)-\alpha,1))$ and $\Bhbf([r(d)-a\alpha,r(d)+a\alpha)\cup[r(d)+(a+1)\alpha,1))$, and
    Also suppose that both $a\ge 5$ and $\alpha\cdot\exp\big(\alpha k/2(4a+1)\big)\ge 3$.
    Then the surplus $s\ge\alpha n/\kappa$.
\end{claim}
\begin{proof}
To bound the surplus, we start with the assumed lower bound on the number of available hospitals:
$[1-r(d)+(a-\frac32)\alpha] \frac{n}{\kappa}$.
We then subtract an upper bound given by \Cref{lem::doctors-high-value-match-prob} on the number of hospitals that did not fully match in the range $[r(d)+(a+1)\alpha,1)$, for by \Cref{clm::matching-bounds-extend},
all hospitals in this range that matched for intervals $I$ in this range, also match in the double-cut DA for $d$. We take a disjoint collection of intervals $I$ that cover this range; they contain at most $n/\kappa$ hospitals. Therefore, the number that are not fully matched is at most $\frac{n}{\kappa} \cdot 3\exp\big(-\alpha k/2(4a+1)\big) \le  \alpha \cdot \frac n{\kappa}$, if $\exp\big(- \alpha k/2(4a+1)\big)\le \frac 13\alpha$.
We then determine how many hospitals could be fully matched by the doctors at hand, which is at most $1/\kappa$ times the assumed upper bound on the number of doctors: $\frac1{\kappa}[1-r(d)+\frac 32\alpha ]n$. This yields a surplus of at least
 \hide{
    % As $\Bdbm([r(d)-\alpha,1))$ does not occur, 
    By assumption, the number of doctors with public rating at least $r(d)-\alpha$ is at most $[1-r(d)+\frac 32\alpha ]n$. 
    By the assumption in \Cref{lem::doctors-high-value-match-prob} and by \Cref{clm::matching-bounds-extend}, the number of not-fully matched hospitals in $[r(d)+(a+1)\alpha,1)$ is at most
    $3\frac n{\kappa}\cdot\exp\big(-\alpha\kappa k/2(4a+1)\big)\le  \alpha \cdot \frac n{\kappa}$, if $\exp\big(- \alpha\kappa k/2(4a+1)\big)\le \frac 13\alpha$.
    In addition, some or all of the hospitals in $[r(d)+a\alpha,r(d)+(a+1)\alpha)$ may not be fully matched.
    % As $\Bhbf([r(d)-a\alpha,r(d)+a\alpha)\cup[r(d)+(a+1)\alpha,1))$ does not occur,
    By assumption, the number of hospitals with public rating in $[r(d)-a\alpha,r(d)+a\alpha)\cup[r(d)+(a+1)\alpha,1))$ is at least $[1-r(d)+(a-\frac32)\alpha] n/\kappa$.
    Thus the surplus, the number of not-fully matched available hospitals in $d$'s cone, i.e., hospitals that could propose to $d$, is at least
    }
    \begin{align*}
        \big[1-r(d)+(a-\tfrac32)\alpha\big] \cdot\frac{n}{\kappa} - \big[1-r(d)+\tfrac 32\alpha\big]\cdot\frac{n}{\kappa} -  \alpha \cdot\frac{n}{\kappa}
        \ge \big(a - 4\big)\alpha\cdot \frac{n}{\kappa}\le \alpha \cdot \frac{n}{\kappa},
        ~~~~\text{if $a\ge 5$}.
    \end{align*}
\end{proof}

We put one more (harmless) stop on the hospital proposing double-cut DA in the statement of the lemma: if $d$ receives a proposal that provides
$d$ with interview value $1-\psibar_d$ or larger, we stop the run as the desired outcome has been achieved.
We fix the public ratings of the doctors and hospitals. We also fix all private values and interview values of the doctors other than $d$ and the interview values of all hospitals outside $d$'s cone.
    
We model the possible runs of the double-cut DA using a tree $T$.
As we descend the tree, the maximum utility achievable by each hospital will be decreasing.
At each node $v$, each hospital $h$ has a current maximum possible utility $u_c(h,v)$, either due to its being tentatively matched with this utility, or having the possibility of achieving this utility, but no larger, due to the remaining random choices; this is called $h$'s \emph{current utility} for short.
In making these random choices, we will allow the hospitals to partially or fully determine some interview values before we determine whether there is an interview.

The nodes and the edges to $v$'s children represent the possible actions by one currently unmatched hospital (if several hospitals could make a proposal, to avoid ambiguity, a fixed tie-breaking ordering is used; in a subsequent lemma we will impose some mild, but feasible, requirements on the tie-breaking). 
$T$ has two types of nodes. A node is of the first type if it corresponds to a proposal by some hospital $h$ outside $d$'s cone, where this proposal provides $h$ utility equal to its current possible utility $u_c(h,v)$.
Due to the tie-breaking, and the fixing of utilities for $h$,
there can be only one such proposal at node $v$. 
Reflecting this, there is a single edge descending from $v$ corresponding to making this proposal.
A node $v$ is of the second type if it corresponds to a possible proposal by a hospital $h$ in $d$'s cone providing $h$ utility equal to its current utility $u_c(h,v)$; below, we write $u=u_c(h,v)$ for brevity. The following are the possible actions at such a node $v$:
    \begin{itemize}
        \item The choice that $h$ have interview value $u-r(d')$ for $d'$ and there is a proposal to $d'$ with doctor interview value at least $1-\psibar_{d'}$. We call this a \emph{high-real proposal}. If $d'=d$, we terminate the DA run at this point, as this achieves the desired outcome. 
        \item The choice that $h$ have interview value $u-r(d')$ and there is a proposal to $d'$ with doctor interview value less than $1-\psibar_{d'}$. We call this a \emph{low-real proposal}. 
        If $d'=d$, we don't allow $h$ to make any further proposals.
        \item The choice that $h$ have interview value $u-r(d')$ but there is no proposal to $d'$, because there was no interview (this is a ``no action'' step, but we need to handle it separately). We call this an \emph{imaginary proposal}. Again, if $d'=d$, we don't allow $h$ to make any further proposals.
        (For the purposes of the analysis, we want to choose the interview value before we determine whether there was an interview.)
        \item ``No action'' (due to the determination that $h$'s remaining interview values are smaller than needed to make a proposal currently; the next node corresponds to $h$ checking if it has a proposal at its next utility value in decreasing order of utilities).
    \end{itemize}
    Let $q(v)$ be the probability that at node $v$ the continuation of the run of the double-cut DA results in zero high-real proposals to doctor $d$. 
    We wish to analyze a corresponding setting with independent proposals, which will have probability
    $\qtilde(v)\ge q(v)$ of zero real proposals at node $v$; $\qtilde$ is designed so that it will be straightforward to upper bound $\qtilde(\treeroot(T))$.
    
    To define $\qtilde(v)$ we consider the following setup.
    Let $h$'s \emph{residual utility} at node $v$, $u(h,v)$, equal the current utility $u$ minus the minimum utility allowed in the run of the double-cut DA, namely $r(d)+1-\alpha$.
    % , minus the length of the excluded utilities from fixing the utilities w.r.t.\ doctors $d'\ne d$, which is at most $n/z\le \delta$.
    Let $s$ be the minimum surplus over all the possible paths in the computation, assuming no bad events happen.
    Suppose that above node $v$ there have been $x$ low-real proposals and $y$ imaginary proposals to $d$.
    Let $S(v)$ be the set of the $s-x-y$ hospitals with the smallest residual utilities
    among those that have not made a real or imaginary proposal to $d$ already.
    Suppose that among the $|C(d)|-x-y$ hospitals that have not made a real or imaginary proposal to $d$, $k-x$ are chosen uniformly at random. Suppose we then have each of the selected hospitals in $S(v)$ choose its interview value and propose to $d$ if the interview value is at least $1-\alpha$.
Then $\qtilde(v)$ is defined to be the probability that this process results in no high-real proposal to $d$.
In \Cref{clm::fixed-k-stoc-dom-idio},
we will show that $\qtilde(v)\ge q(v)$.
Before proving this claim, we deduce a bound in the $\qtilde$ setting, and therefore in the $q$ setting also, on the probability that $d$ receives no proposal for which $d$ has interview value at least $1-\psibar_d$.

\begin{claim}\label{lem::fail-prob-hospital-prop}
    Let $d$ be a doctor with public rating at least $a\alpha$.
    Suppose $s\ge \frac{\alpha n}{\kappa}$ and that
    there are at most $\frac12(4a+1)\frac{\alpha n}{\kappa}$ hospitals in $C(d)$ (the expected number is $\frac{2a\alpha n}{\kappa}$, except at the ends of the public rating range where it can be as low as $\frac{a\alpha n}{\kappa}$). Then, in the $\qtilde$ setting, the probability that $d$ receives no proposal for which $d$ has interview value at least $1-\psibar_d$ is at most $\exp\big(-2k\alpha\psibar_d/(4a+1)\big)$.
\end{claim}
\begin{proof}
    For each hospital $h$, its residual utility at the root of $T$ is at least $\alpha$.
    Consider a hospital $h$ in the surplus.
    The probability that $h$ is selected by $d$ is at least $\frac{k}{|C(d)|}\ge \frac{2k\kappa}{(4a+1)\alpha n}$. The probability that $h$ has interview value at least $1-\alpha$, as required by the double-cut DA for $h$ to propose this edge, is $\alpha$,
    and the probability that $d$ has interview value at least $1-\psibar_d$ for $h$ is $\psibar_d$.
    So the probability that there is an edge from $h$ to $d$ that is chosen by the double cut DA and
    that provides $d$ interview value at least $1-\psibar_d$ is at least $\frac{2k\kappa\psibar_d}{(4a+1)n}$.
    Therefore the probability that $d$ receives no such edge is at most
    $\big(1- \frac{2k\kappa\psibar_d}{(4a+1)n}\big)^s\le \exp\big(-\frac{2k\alpha\psibar_d}{4a+1}\big)$,
    as $s\ge \frac{\alpha n}{\kappa}$.
    \hide{
    The expected number of hospitals in $S(\treeroot(T))$ that are selected by $d$ is
    $ks/|C(d)|\ge 2k\frac{\alpha  n}{\kappa}/(4a+1)\frac{\alpha n}{\kappa}= 2k/(4a+1)$.
    If in addition, the hospital has interview value at least $1-\alpha$, as required by the double-cut DA, and $d$ has interview value at least $1-\psibar_d$, the expected number becomes
    $2k\alpha\psibar_d/(4a+1)$.
    By a Chernoff bound, the probability that there are no such hospitals is at most
    $\exp\big(-k\alpha\psibar_d/(4a+1)\big)$. 
    }
\end{proof}

\begin{claim}\label{clm::fixed-k-stoc-dom-idio}
    For every node $v$ in $T$, $\qtilde(v)\ge q(v)$.
\end{claim}
\begin{proof}
We prove the result by an induction from the leaves to the root of tree $T$.

Recall that $u(h,v)$ denotes the residual utility of $h$ at node $v$,
and let $u_1\le u_2\le \ldots\le u_{s-x-y}$ be the $s-x-y$ smallest residual utilities being used to calculate $\tilde{q}(v)$; the hospitals that have already proposed to $h$ are excluded.

\smallskip\noindent
\emph{Base case: a leaf node}\\
At a leaf node $v$ all the random decisions have been made, so either $d$ has received the desired proposal or not. Therefore $q(v)=\qtilde(v)$.

\smallskip\noindent
\emph{Inductive step at an internal node $v$ with one child $w$}\\
Here $q(v)=q(w)$ and $\qtilde(v)=\qtilde(w)$.
So the inductive claim holds trivially.

\smallskip\noindent
\emph{Inductive step at an internal node $v$ with multiple children}\\
Let $h$ be the hospital making a random choice at node $v$.
Let its children be named $\wrh(d')$, corresponding to a real high proposal to doctor $d'$, $\wrl(d')$, corresponding to a real low proposal to doctor $d'$, 
$\wim(d')$, corresponding to an imaginary proposal to $d'$,
and $\wna$, corresponding to ``no action''.
Recall that $d$ has received $x$ low-real and $y$ imaginary proposals prior to reaching node $v$.
Then $\qtilde(v)$ is calculated using $s-x-y$ residual utilities,
$\qtilde(\wna)$ is calculated using $s-x-y$ residual utilities, 
as are $\qtilde(\wrl(d')),\qtilde(\wrh(d')),\qtilde(\wim(d'))$ for $d'\ne d$,
while $\wrl(d)$ and $\wim(d)$ are calculated using $s-x-y-1$ residual utilities,
as in these cases, $h$ makes no further proposals in the run of the double-cut DA;
$\wrh(d)$ is a terminal node as it achieves the desired outcome of a proposal to $d$ with a ``high'' interview value (at least $1-\psibar_h$).
Note that $q(\wrh(d))=\qtilde(\wrh(d))=0$.
Finally, in the calculation of $\qtilde(w)$ for $w\ne \wrh(d),\wrl(d),\wim(d)$, all the residual utilities except at most one are unchanged; the only one that may change is due to the fact that $u(h,w)<u(h,v)$. 

Let $p_w$ be the probability that the edge from $v$ to its child $w$ is chosen when at node $v$.
Then $q(v) = \sum_{\text{$w$ a child of $v$}} p_w \cdot q(w)$.
To prove our inductive bound it will suffice to show that
\begin{align}\label{eqn::qtilde-bound-idio}
    \qtilde(v) \ge \sum_{\text{$w$ a child of $v$}} p_w \cdot \qtilde(w).
\end{align}

\smallskip\noindent
\emph{Case 1.} $u(h,v)$ is among the smallest $s-x-y$ residual utilities at node $v$\\
This immediately implies that $u(h,w)$ is among the smallest $s-x-y$ residual utilities at $w$ for $w\ne \wrh(d),\wrl(d),\wim(d)$, as $u(h,w)<u(h,v)$. The residual utilities for the other hospitals are unchanged at every child of $v$.

In this case, we show that \eqref{eqn::qtilde-bound-idio} holds with equality.
To do this we consider the random choices that are made in evaluating $\qtilde(v)$,
by instantiating exactly the randomness needed to select a child $w$ of $v$.

One way to make these random choices is to determine if $h$'s utility for
$d$ equals the utility at node $v$; if so, a further random choice determines if the
proposal is real or imaginary.
If it is real, yet another random choice determines if it is high or low.
At each child, we then make the remaining random choices to determine the probability
of no high-real proposal to $d$. But at each node $w$, this probability
is exactly $\qtilde(w)$, for we are using the same random variables apart from $h$'s
utility which is either the value $u(h,v)$ or a value determined at $w$, for $w\ne \wrh(d),\wrl(d),\wim(d)$; furthermore, this utility is not used at $v$'s remaining children.
Thus $\qtilde(v)=\sum_{\text{$w$ a child of $v$}} p_w\cdot \qtilde(w)$ in this case.

\smallskip\noindent
\emph{Case 2}.  $u(h,v)$ is not among the smallest $s-x-y$ residual utilities at node $v$.\\
Because of the rounding of the utilities, this implies that $u(h,w)\ge u_{s-x-y}$, for $w\ne \wrh(d),\wrl(d),\wim(d)$. All the other
residual utilities are unchanged, and therefore $\qtilde(v)=\qtilde(w)$ for these $w$.
Consider the effect of increasing $u_{s-x-y}$ to $u(h,w)$ (note that the latter value is the same for all these $w$. At each child $w'= \wrh(d),\wrl(d),\wim(d)$ of $v$, $\qtilde(w')$ is unchanged, and at each child $w\ne \wrh(d),\wrl(d),\wim(d)$, $\qtilde(v)$ and $\qtilde(w)$ are reduced equally,
as the increase to $u_{s-x-y}$ increases the probability that a proposal to $d$ occurs.
Therefore if \eqref{eqn::qtilde-bound-idio} held after this increase, it would have held before the increase.
Now increase $u_{s-x-y}$ and any other intermediate $u$ values further to $u(h,v)$. $\qtilde(v)$ continues decreasing, but all
the $\qtilde(w)$ are unchanged, for $w$ a child of $v$.
When $u_{s-x-y}=u(h,v)$, this becomes an instance of Case 1, at which point \eqref{eqn::qtilde-bound-idio} holds.
We conclude that \eqref{eqn::qtilde-bound-idio} held prior to the increase of $u_{s-x-y}$.
\end{proof}

\subsubsection{Overview of the Proof of \Cref{lem::doctors-interval-non-match}}
\label{sec::doc-int-idio-no-match}

The participating doctors are those with public rating at least $g-\alpha$.
The participating hospitals are a subset of those with public rating at least $g-a\alpha$, which we specify more fully below. Our goal is to handle the set of doctors in $I$ in the same way as we handled the single doctor $d$ in \Cref{lem::doctors-high-value-match-prob}. The challenge we face is that the stochastic tree dominance argument is going to show, for a suitable bound $r$, the probability that at most $r$ doctors in $I$ receive proposals is upper bounded by the probability of this event for $s$ independent proposals to $I$. However, we need to ensure that each of these proposals is roughly uniform over $I$,
and this need not be true for all the proposals we consider.

% This requires the removal of some hospitals with rating at least $g-a\alpha$.

To achieve this near uniformity, we start by excluding all hospitals with public rating in the range $[f+a\alpha,g+a\alpha)$, for these are the hospitals that are in cone for some but not all the doctors in $I$. We note that
the hospitals in the range $L\triangleq [g-a\alpha,f+a\alpha)$ are in-cone for all doctors in $I$.

Next, we remove any hospital in $L$ that has two or more edges to doctors in $I$ (typically, most hospitals in this range have no such edges), for the first proposal to $I$ from such hospitals would not be uniform over $I$.
In addition, the doctors $d$ incident on these edges are removed, as are the
hospitals which are neighbors of removed doctors.
Let $I'$ and $L'$ be the resulting sets of doctors and hospitals, respectively.
We note that the second set of removed hospitals each had only one edge to doctors in $I$,
and therefore the doctors in $I'$ each have all their $k$ edges incident on hospitals in $L'$.
In \Cref{clm::hosp-removal-lots-edges} we show that $|I'|\ge |I|\big[1-\exp\big(-\alpha k/2(4a+1)\big)\big]$ with high probability.

We partially exclude some additional hospitals so as to ensure that they have a roughly equal probability of proposing to any doctor in $I'$ over the run of the double-cut DA.
In particular, for each doctor $d'\in I'$, if the hospital obtains a utility in the range
$[r(d')+1-\alpha,r(d')+1)$, we want the utility to be in the smaller range
$\cap_{d\in I'}[r(d)+1-\alpha,r(d)+1)\supseteq [g+1-\alpha,f+1)$.
To achieve this, we terminate the run of any of these hospitals which obtains a utility in the range
$\big([f+1-\alpha,g+1-\alpha)\cup [f+1,g+1)\big)\cap([r(d)+1-\alpha,r(d)+1)$ for any doctor $d$ in $I'$ (note that $f\le r(d)<g$), right before they make such a proposal.

The rest of the argument will be similar to the one for \Cref{lem::doctors-high-value-match-prob}.
       For this result we are only seeking to ensure matches, and not that the doctors achieve high interview values. 
       %Still, we want the hospitals to make lots of ``real'' proposals to the doctors in $I'$ so as to ensure most of these doctors are matched; this makes the structure of the analysis somewhat similar to that in the prior lemma.  %We let $r+1$ be the desired number of real proposals; we specify $r$ later.
%
To this end, we fix $I'$ and the utilities for all doctors outside of $I'$.
We will be seeking to analyze the double-cut DA as if $I'$ was a single doctor that makes $k|I'|$ proposals.
The reason this does not cause difficulties is that because of the above exclusions 
% assuming $\Bcol$ does not occur, then 
every edge from a doctor in $I'$ to a hospital in $\cap_{d\in I'} C(d)$ is to a distinct hospital.

We defer to the appendix the claims that the surplus $s\ge \alpha n/\kappa$, the probability that the number of unmatched doctors is more than
    $3|I|\cdot \exp\big(-\frac{\alpha k}{2(4a+1)}\big)$ is at most $n^{-(c+2)}$,
    and that a suitably defined $\qtilde$ and $T$ satisfy $\qtilde(v)\ge q(v)$ for all
    nodes $v\in T$.

\subsubsection{The $\eps$-Nash Equilibrium for all but the Bottommost Doctors}
\label{sec::idio-eps-Nash}

We begin by identifying sufficient conditions to satisfy the constraints in Claims~\ref{clm::conditions-hosp-prop-one-doctor}, \ref{clm::conditions-doctor-prop-one-hosp},
\ref{clm::hosp-prop-interval-constraints}, and \ref{clm::doc-prop-interval-constraints},
as subsumed in \Cref{clm::final-cstrnts-idio}.

\begin{claim}\label{clm::final-cstrnts-idio}
    With $\delta=0$ and $\alpha=\big(\frac{2(4a+1)\ln k}{k}\big)^{1/2}$, $a=5$, and $b=3$ the following constraints suffice: 
\begin{align*}
                \Omega(1)&\le k \le \frac 12(c+2)\ln n \\
         \frac{n}{\ln n} &=\Omega\Big(\frac{(c+2)k^{2.5}\kappa}{(\ln k)^{1/2}}\exp\big(\Omega\big( (k\ln k)^{1/2}\big)\big)\Big).
    \end{align*}
\end{claim}

\begin{claim}\label{clm::prob-B}
Let ${\mathcal B}$ denote the bad events excluded in \Cref{clm::final-cstrnts-idio}.
% Claims~\ref{clm::conditions-hosp-prop-one-doctor}, \ref{clm::conditions-doctor-prop-one-hosp}, \ref{clm::hosp-prop-interval-constraints}, \ref{clm::doc-prop-interval-constraints}.
The probability that ${\mathcal B}$ occurs is at most $12n^{-(c+1)}$.
\end{claim}
\hide{\begin{proof}
    Summing the probabilities of the bad events in these four claims gives a bound of
    $\big(7n+\frac{15}{g-f}+\frac 3n\kappa\big)n^{(-c+2)}\le 11n^{(-c+1)}$, if $g-f\le n/15$.
    \rnote{There may be double counting here. I have not checked.}
\end{proof}
}

Next, we show that for a non-bottommost student, a student with public rating at least $a\alpha$,
changing strategies will not significantly improve her outcome in expectation.

\begin{claim}\label{lem::typical-student-eps-nash}
Suppose the constraints in \Cref{clm::final-cstrnts-idio} hold.
    Let $d$ be a doctor with public rating at least $a\alpha$, and let $\alpha=\big(\frac{2(4a+1)\ln k}{k}\big)^{1/2}$.
        In expectation, ex-ante, $d$ could improve her utility by at most 
        $O(\alpha)=O((\ln k/k)^{1/2})$, if $k\ln k\ge \big(\frac83\big)^2\cdot\frac{4a+1}{\kappa^2}$.
\end{claim}
\begin{proof} (Sketch.)~
Basically, there are two ways in which a doctor could gain expected $O(\alpha)$ utility in principle (but it is not clear such changes in strategy would lead to actual gains).
First, by changing a bid in cone, this might improve the public rating component of the match utility
by up to $2a\alpha$. Second, by changing bids so as to avoid matches with low interview values,
she might improve the match utility when this happens by $O(1)$; but as \Cref{lem::hosp-high-value-match-prob} shows, this is a somewhat uncommon event, with probability at most $\alpha$, yielding an expected gain of $O(\alpha)$. The other gains, from better results when ${\mathcal B}$ occurs and improvements in the private value, are small in comparison. The full proof can be found in the appendix.
\end{proof}

This concludes the proof of \Cref{thm::eps-nash-idio-large-cap}.

\subsection{The Student-School Setting}\label{sec::student-school}

The analysis for this result is quite similar to that for the residency setting.
We give an overview of what changes here, and provide a full explanation in the appendix.

Lemmas \ref{lem::hospitals-interval-non-match} and \ref{lem::hosp-high-value-match-prob} are unchanged, except that there is no hospital interview value to consider in \Cref{lem::hosp-high-value-match-prob}
(which more or less amounts to setting $\psibar_h=0$).
The statement in \Cref{lem::doctors-high-value-match-prob}:
\begin{quote}
    Then the probability that $d$ fails to receive a match for which it has interview value $1-\psibar_d$ or larger is at most $\exp\big(-\frac{k\alpha\psibar_d}{4a+1}\big)$
\end{quote}
changes to
\begin{quote}
    Then the probability that $d$ fails to receive a match for which it has interview value $1-\psibar_d$ or larger is at most $\exp\big(-\frac{4k\psibar_d}{4a+1}\big)$.
\end{quote}
This is because we can improve the corresponding bound in \Cref{lem::fail-prob-hospital-prop}, the reason being again
that there is no hospital interview value to consider, so instead of being limited
to interview values in the range $[1-\alpha,1)$, which happens with probability $\alpha$,
we have an event with probability 1.
This is the mismatch probability reported in \Cref{thm::eps-nash-no-idio-large-cap}.
% We could also improve the bound in \Cref{lem::doctors-interval-non-match} similarly, but this does not materially change the overall result, so we chose to leave it unchanged.

Now, in the analog of \Cref{lem::typical-student-eps-nash}, namely \Cref{clm::final-cstrnts-noid},
    we set $\alpha=\frac{4(4a+1)\ln k}{k}$ instead of $\Big(\frac{4(4a+1)\ln k}{k}\Big)^{1/2}$, which yields the $\eps$ value stated in \Cref{thm::eps-nash-no-idio-large-cap}.

We note that in this setting the stable matching is unique (see \Cref{clm::unique-sm}).

\section{Extensions/Open Problems}\label{sec::extensions}

\paragraph{Unequal Numbers of Doctors and Hospital Places}

In practice, it is likely that these two numbers will differ somewhat. Indeed, in the NRMP data
for 2024 \cite{NRMP24}, there were roughly 40,000 positions and 50,000 applicants (including those from foreign medical schools). Our results are essentially unchanged. All that happens is that if we align
the doctors and hospitals by public rating order, then the results for  the non-bottommost agents
continue to hold as before, but it is possible for this group to include all the agents on the short side.
To make this precise, suppose there are $n+m$ doctors and $n/\kappa$ hospitals, with public rating
 ranges $[0,1+\frac mn)$ and $[\frac nm,1+\frac nm)$ (we can always shift the range; it is defined this way for convenience).
 Then, for the doctors, the bottommost agents are those with public rating in the interval
 $[0,a\alpha+\frac nm)$, and for the hospitals, it is $[\frac nm, \max\{a\alpha,\frac nm\})$;
 note that if $\frac nm\ge a\alpha$, there are no bottommost hospitals.
 If there are more hospitals than doctors, the definition is reversed.
 This contrasts with the outcome in random markets \cite{AKL17} where all agents on the short side benefit substantially.

\paragraph{Non Uniform Hospital Capacities}

 We believe the current analysis can accommodate this with some modest changes. The most important is that doctors would need to select hospitals in proportion to their capacities. One way to do this would be for each doctor to order her private values as follows: suppose she has private value $1-\delta_i$ for hospital $h_i$ with capacity $\kappa_i$. Then she would choose the $k$ hospitals that maximize the values $1 - \frac{\delta_i}{\kappa_i}$.
 Now, a hospital's probability of being fully matched will be a function of its capacity.
 We believe the results stated in Theorems \ref{thm::eps-nash-no-idio-large-cap} and \ref{thm::eps-nash-idio-large-cap} continue to hold, where for hospital $i$ we replace the term $\kappa$ by its capacity $\kappa_i$.

 \paragraph{More General Utility Functions}
 In general, the utilities could be functions of the public ratings, private and interview values. We believe the results would be similar. % with the exponent of the match rate depending on the degree of non-uniformity. 
  One part of our analysis needs the expected number and the actual number of agents in a public rating interval to be similar. With bounded rates of change of the utilities w.r.t.\ the public ratings, this will continue to hold. And for the tree stochastic dominance, we believe it suffices to have rates of change for the doctors' utilities w.r.t.\ each of the private and interview values being within constant factors of each other, and similarly for the hospitals.
  We note that this was the generalization considered in \cite{AC23}, who showed that their results continued to hold up to constant factors. We believe the same holds here.

In the appendix, we consider one instance of this generalization, in which the weight of the interview values is reduced. The reason to consider this case is that 
in our main analysis, the interview values largely determine the preference orderings,
though not completely; there is still a modest bias in favor of higher public ratings.
So it seems natural to consider what happens if we reduce the weight of the interview values.

\paragraph{The Request Interview Setting}

It would be interesting to extend this analysis to a setting in which the doctors request interviews but the hospitals grant only some of them.
For this setting, we provide the hospitals with private values, which they use to decide which interviews to grant.
We conjecture that when the hospitals each receive $\Omega(\ln n)$ proposals with failure probability
$n^{-(c+2)}$, then the double-cut method can be extended to this setting, yielding results similar to those for the setting we analyzed in \Cref{sec::proof-idio-thm}.
Based on the NRMP data, this seems a reasonable assumption.

We suggest the following parameters for our protocol.
Each doctor sends $k^2$ requests for interviews to the hospitals, choosing the in-cone hospitals with the $k^2$ highest private values.
From these requests, each hospital then grants $\kappa k^{3/2}$ interviews, again choosing the doctors with the highest private values. As each hospital receives an expected $\kappa k^2$ requests,
if $\kappa k^2=\Omega(\ln n)$, then with high probability each hospital receives at least
$\frac 12 \kappa k^2$ requests, and so is able to grant $\kappa k^{3/2}$ interviews, if $k\ge 4$.
Furthermore, the hospital's private value for each of these interviews will be at least
$1-2/k^{1/2}$ in expectation, and so have a range smaller than the cone size, as we will be setting
$\alpha = \Theta\big(\sqrt{\ln k/k}\big)$ as before.
Finally, when listing preferences, the hospitals will limit themselves to the doctors with the top $\kappa k$ interview values; again, in expectation, these interview values are at least $1-2/k^{1/2}$.
It is not critical that the hospitals reduce their proposals in this way: it just seems natural given the form of their utility functions.

The challenging part in the analysis is to redefine $\qtilde$ for this setting.
Let us consider the definition of $\qtilde$ in \Cref{sec::proof-idio-thm}.
The difficulty is that when a hospital $h$ determines that it has an interview 
value enabling it to propose to $d$, we need to decide whether $d$ requested an interview (which we can handle as before), and also whether $h$ selected $d$ for an interview.
This depends on the number of proposals $h$ received.
But with the assumption that this is $\Omega(\ln n)$ proposals, with high probability
we know that $h$ received between $\frac12 \kappa k^2$ and $2\kappa k^2$ proposals,
and therefore, if we underestimate the selection probability as being $1/2k^{1/2}$,
we are at worst underestimating it by a factor of 4, which changes the constants in our
analysis, but allows it to go forward with the same structure.

The changes to the proofs of the other lemmas are similar. As we obtain the same bounds up to constant factors in the exponent, the $\eps$-Nash analysis remains the same but with $\eps$ increased by constant factors (not in the exponent).

\hide{
\subsection{Multiple Public Ratings}

In many settings, there is more than one item of public information and different agents weigh them differently.
We propose a simple stylized extension of the current residency setting to address this.
Suppose each doctor has two public ratings, $r(d)$ and $s(d)$, each a uniform random draw from $[0,1)$.
Each hospital has a preference $\lambda(h)$, which it uses to obtain its utility: $\lambda(h) \cdot r(d) + (1-\lambda(h))\cdot s(d)$. Furthermore, $\lambda(h)$ is also public information.
As a first step, we propose considering the setting without interviews.
Suppose the doctor utilities are unchanged, namely $d$ has utility $r(h)+v(d,h)$ for $h$, as before.
The advantage on the doctor's side is that they know how each hospital will assess them.
The question to resolve is how the doctors should target their applications given the relative values of $r(d)$ and $s(d)$.

\rnote{If I can recall it, it may be worth indicating our idea.}

If we can solve the problem for two types are information, it seems plausible that the solution would extend to more types. Furthermore, we believe that adding interviews would lead to results similar to those in the present paper.
}

\paragraph{The Outcome for the Bottommost Agents} One important unresolved issue is what is the outcome for the bottommost doctors and what strategy they should follow.
We suspect they should bias their proposals to be directed more toward the bottom of their cones, but even then it seems we will need different methods to analyze this.

\paragraph{The Range for Utility Values}
Our analysis implicitly assumes that all matches are preferred over non-matching. This seems unlikely in practice. To address this, one could shift the range of the utilities so that the likely outcome for the bottommost agents on one or both sides is 0 or even negative. How this might affect the strategies of the bottommost agents merits further investigation.

 \section{Experiments}\label{experiments}

We simulated the residency settings described in the analysis, namely the setting in which the doctors select the interviews and the setting in which the doctors request interviews.
We carried out simulations for $n$ doctors with $n=\text{2,000}$, for $\kappa =1$ and 5,
and for $k=5$ and 12, and a cone size of 0.3.  In all the graphs, the gray whiskers are showing the 10/90\% boundaries. The $x$-axis in our graphs is indexed by the ranks based on public ratings; in addition, we have grouped the agents in groups of 10 consecutive agents to make the trends clearer.
By loss we mean the average reduction in utility compared to the benchmark: for doctor $d$ this is $r(d)+2$ and for hospital $h$ it is $r(h)+1$.

We start by showing the matching rates, averaged over 100 runs, for $n=\text{2,000}$, $k=5$ and $\kappa=5$, for both settings; for the hospitals, this is the non-full matching rate, so we expect it to be higher than the doctor rate.  When the doctors select interviews, their non-match rate increases from 0\% to 4\% as their public rating drops; when they request the interviews it varies between 2\% and 5\%, for reasons that are not clear to us.  In contrast, the hospitals have a very high match rate at first, consistent with our analysis; the unavoidable consequence is that they have a lower match rate at low public ratings, which is consistent with the reported observation that low-rank hospitals have significant difficulty finding matches.

\hide{
We see that increasing $\kappa$ from 1 to 5 brings improved matching rates particularly for the hospitals but also for the doctors
(our analysis predicted this for the hospitals). However, the losses for the topmost doctors and hospitals increase with the larger capacity. We believe this is due to the fact that the topmost hospitals receive a constant factor fewer proposals. They also increase for the bottommost hospitals; this is an unavoidable consequence of the greater improvement in the match rate for the non-bottommost hospitals, which has the consequence of there being too few doctors at hand to match with the bottommost hospitals.
}

\begin{figure}[h]
    \begin{minipage}{0.5\textwidth}
        \includegraphics[width=\textwidth]{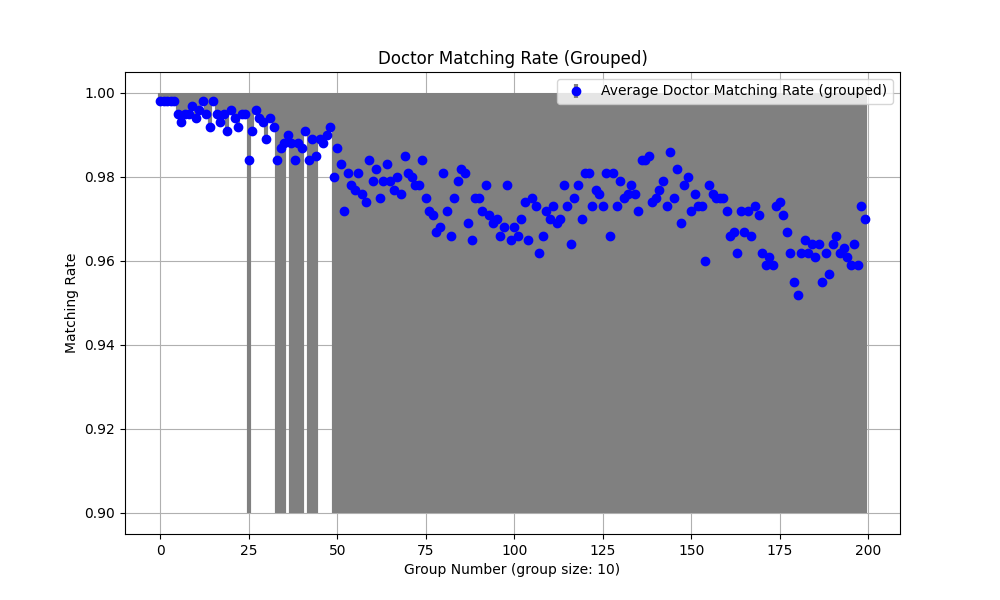}
    \end{minipage}%
    \begin{minipage}{0.5\textwidth}
        \includegraphics[width=\textwidth]{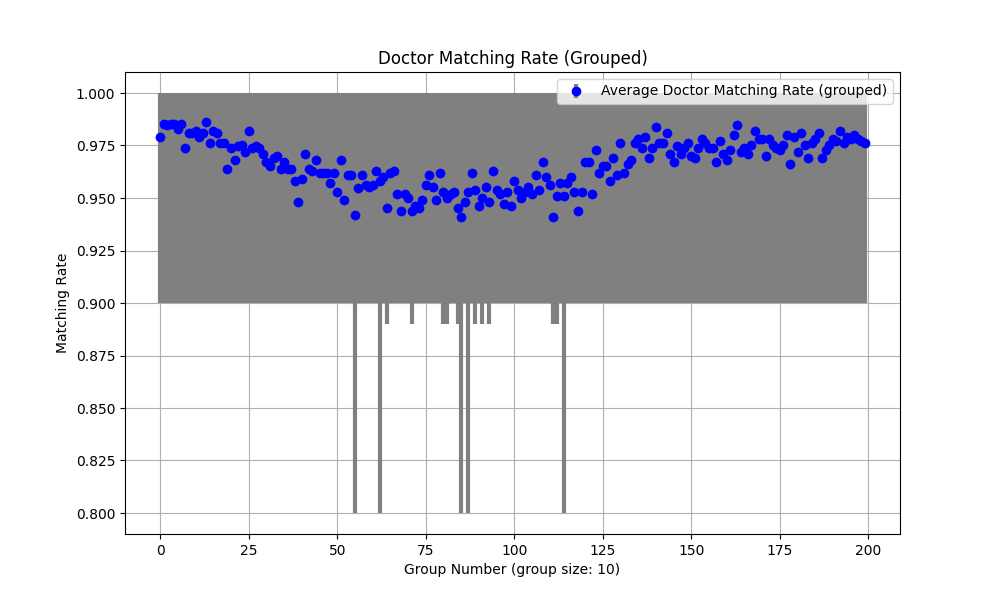}
    \end{minipage}
    \caption{Doctor matching rate, $n=\text{2,000}$, $k=5$ and $\kappa=5$, doctor selects interviews (on the left), doctor requests  (on the right). }
\end{figure}

\begin{figure}[h]
    \begin{minipage}{0.5\textwidth}
        \includegraphics[width=\textwidth]{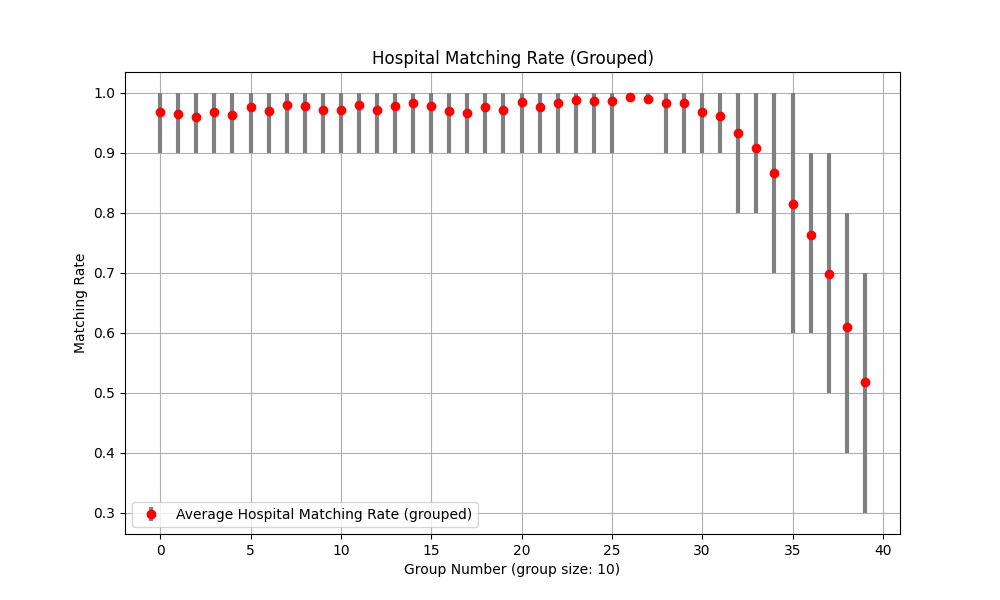}
    \end{minipage}%
    \begin{minipage}{0.5\textwidth}
        \includegraphics[width=\textwidth]{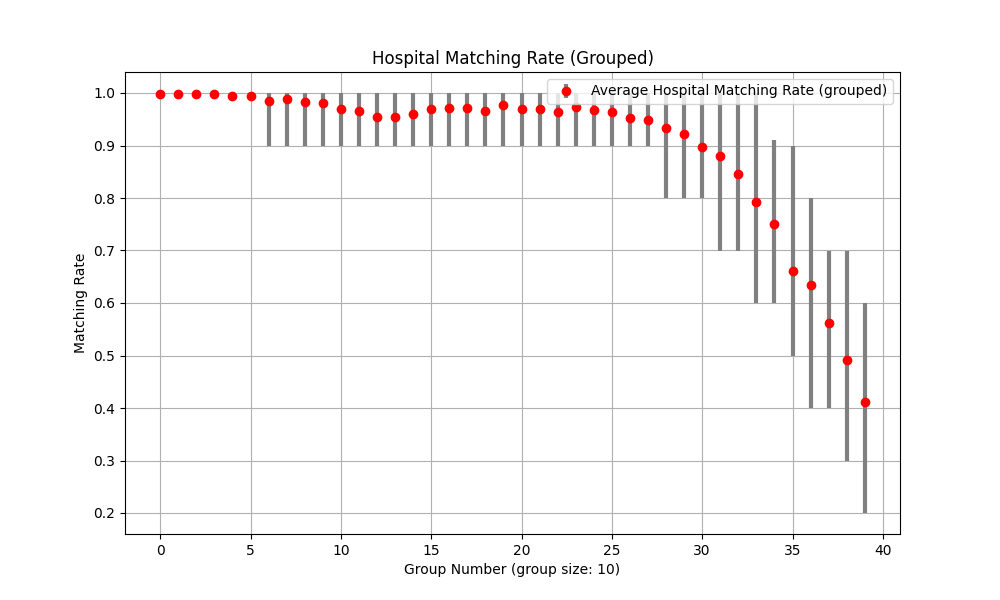}
    \end{minipage}
    \caption{Hospital matching rate, $n=\text{2,000}$, $k=5$ and $\kappa=5$, doctor selects interviews (on the left), doctor requests (on the right).}
\end{figure}

\begin{figure}[h]
    \begin{minipage}{0.5\textwidth}
        \includegraphics[width=\textwidth]{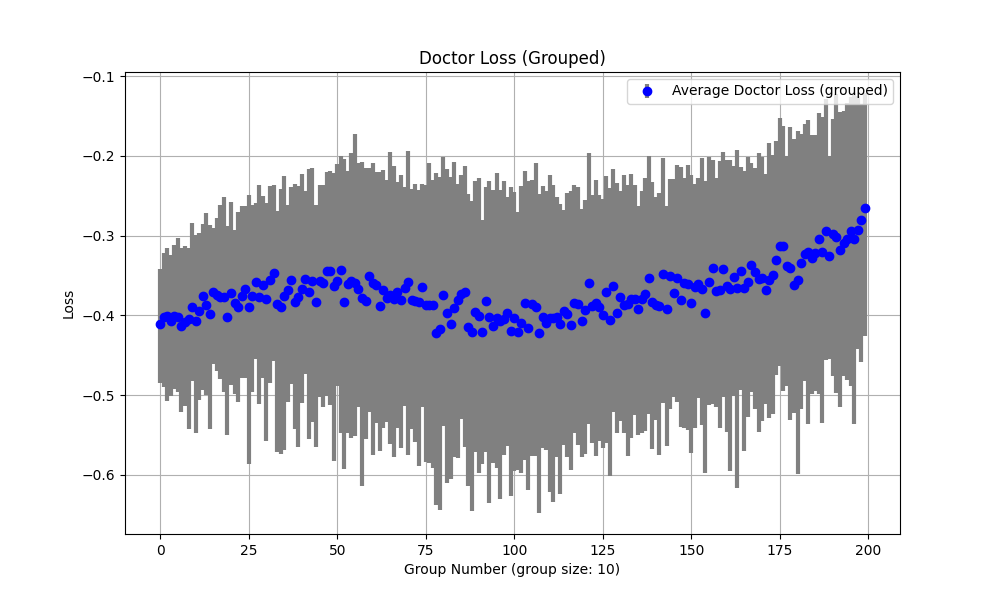}
    \end{minipage}%
    \begin{minipage}{0.5\textwidth}
        \includegraphics[width=\textwidth]{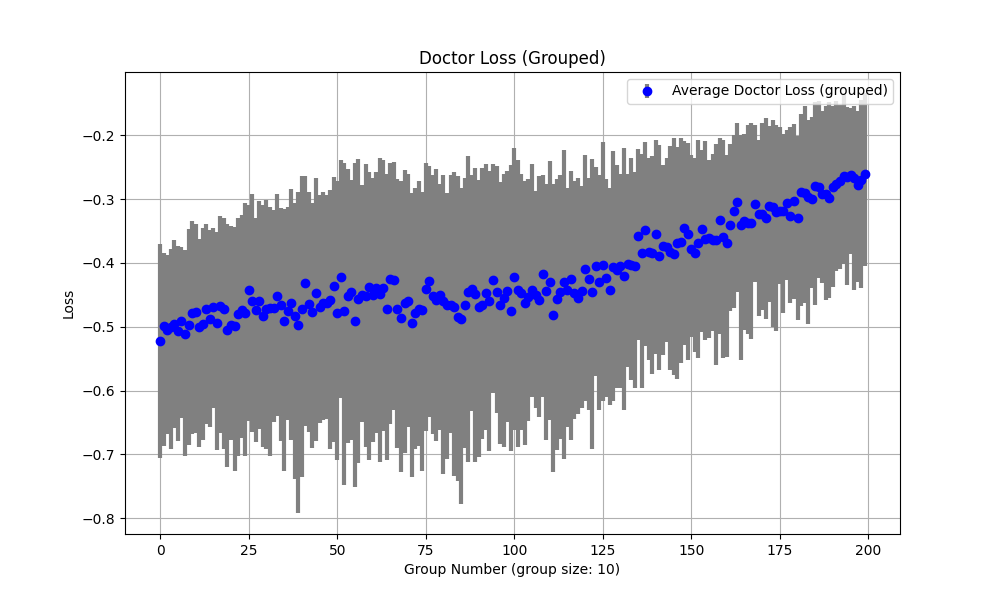}
    \end{minipage}
    \caption{Doctor loss, $n=\text{2,000}$, $k=5$ and $\kappa=5$, doctor selects interviews (on the left), doctor requests (on the right).}
\end{figure}

We then show the expected losses relative to the benchmark.
The trend lines are potentially misleading, for at the top end, agents are necessarily going to lose w.r.t.\ their public rating, and gain at the bottom end. The hospitals receive greater losses at the bottom end due to their reduced match rate. In addition, in the setting with doctors requesting interviews, the hospitals will have losses due to both private values and interview values being less than 1, while in the doctors select interviews setting only interview value losses are present.
So in fact, the improvement in the first setting is greater than suggested by the graph.

\begin{figure}[h]
    \begin{minipage}{0.5\textwidth}
        \includegraphics[width=\textwidth]{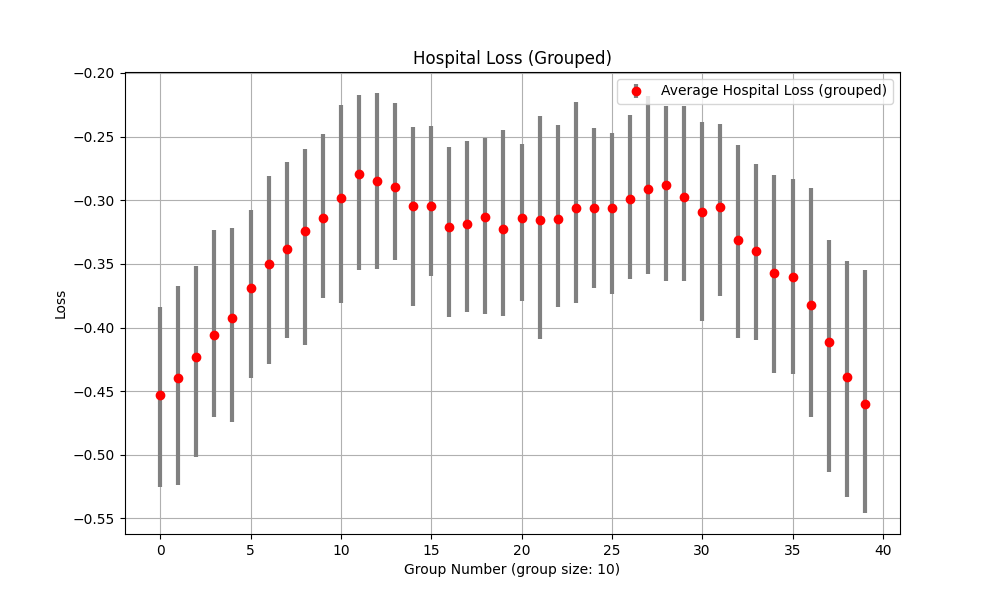}
    \end{minipage}%
    \begin{minipage}{0.5\textwidth}
        \includegraphics[width=\textwidth]{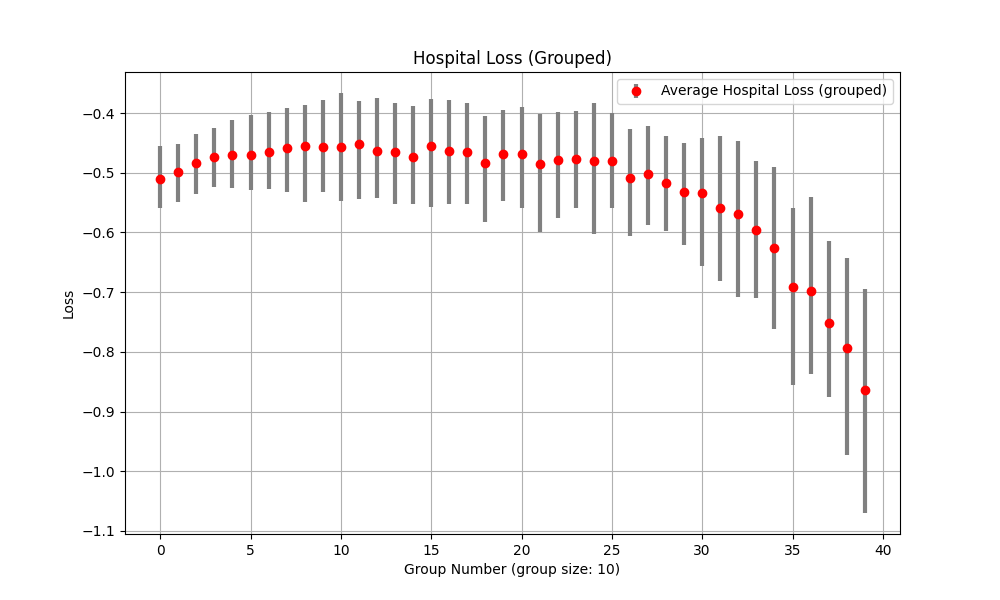}
    \end{minipage}
    \caption{Hospital loss, $n=\text{2,000}$, $k=5$ and $\kappa=5$, doctor selects interviews (on the left), doctor requests (on the right).}
\end{figure}

More results and discussion can be found in the appendix.

\section{Conclusion}\label{sec::conc}

We have considerably extended the scope of the tools introduced by Agarwal and Cole \cite{AC23}, namely the double-cut DA and the tree stochastic dominance.
We believe we have shown these tools are quite versatile and enable the analysis of stable matching scenarios that were previously out of reach.
In particular, we have analyzed a distributed interview matching scenario, and indicated how this analysis might be extended to a setting in which doctors request interviews only some of which are granted, which seems a plausible modeling of the real-world setting.
Our experimental results also indicate that the theoretical asymptotic bounds hold for realistic values of $n$.

%\printbibliography

\newpage
\appendix
\section{The Appendix}

\begin{algorithm}[h]
\SetAlgoNoLine
Initially, all the doctors and hospitals are unmatched.\\
\While{some hospital $h$ with a non-empty preference list is not fully matched}{ 
   let $d$ be the first doctor on her preference list\;
   \If{$d$ is not matched}
      {tentatively match $d$ to $h$.}
      {\eIf{$d$ is currently matched to $h'$, and she prefers $h$ to $h'$}
      {remove $d$ from $h'$'s matches, remove $d$ from $h'$'s preference list, and tentatively match $d$ to $h$.}
      {remove $d$ from $h$'s preference list.}
    }
 }
\caption{\textsf{Hospital Proposing Many-to-One Deferred Acceptance (DA) Algorithm}}
\label{alg:DA-hosp}
\end{algorithm}

For the reader's convenience, we repeat the statements of all claims we prove here.

\subsection{Proofs and other Deferred
Material from \Cref{sec::proof-idio-thm}}

\subsubsection{The Discretization}

\begin{lemma}\label{lem::hospitals-interval-non-match}
Let $\delta>0$ be an arbitrarily small constant.
Let $I=[f,g)$ be a public rating interval with $g\ge a\alpha$.
Suppose we run the (to be specified) doctor-proposing double-cut DA for interval $I$ using the discretized utilities from \Cref{clm::discrete-dist-idio}.
Suppose the fraction of unmatched doctors with public rating in $[g+(a+1)\alpha,1)$ is at most 
$3\cdot\exp\big(-\alpha k/2(4a+1)\big)$.
Also suppose that event ${\mathcal B}$ does not occur.
Then, with failure probability at most $n^{-(c+2)}$,
the fraction of the hospitals in $I$ that fail to match
is less than $3\cdot \exp\big(-\frac{\alpha k}{2(4a+1)}\big)$,
if $a\ge 5$, $\alpha\cdot\exp\big(\frac{\alpha k}{2(4a+1)}\big)\ge 3$, and the constraints from \Cref{clm::doc-prop-interval-constraints} hold.
\end{lemma}

\begin{lemma}\label{lem::hosp-high-value-match-prob}
Let $\delta>0$ be an arbitrarily small constant.
    Let $h$ be a hospital with public rating at least $a\alpha$.
    Suppose we run the doctor-proposing double-cut DA for hospital $h$ using the discretized utilities as specified in \Cref{clm::discrete-dist-idio}.
    Suppose the fraction of unmatched doctors with public rating in $[r(d)+(a+1)\alpha,1)$ is at most $3\exp\big(-\alpha k/2(4a+1)\big)$.
    Also, suppose that the event ${\mathcal B}$ does not occur.
    Finally, suppose we run the doctor-proposing double-cut DA for $h$.
     Then the probability that $h$ fails to have $\kappa$ matches for which it has interview value $1-\psibar_h$ or larger is at most $\exp\big(-\frac{3\kappa k\alpha\psibar_h}{8(4a+1)}\big)$, if $a\ge 5$, $\alpha\cdot\exp\big(\alpha k/2(4a+1)\big)\ge 3$, and $O(1) \le k\le O((c+2)\ln n)$
     (see \Cref{clm::conditions-doctor-prop-one-hosp} for a more precise statement of these bounds).
\end{lemma}

{\sc \Cref{clm::discrete-dist-idio}}~
\emph{For any $\delta>0$, there is a discrete utility space in which
the probability of having any differences
in the preference lists in the original continuous setting and in the discrete setting is at most $\delta$.
In addition, the public values of the doctors will all be distinct, as will the public values of the hospitals, and for each doctor, its private values will all be distinct, and this is achieved within the same overall $\delta$ bound on the failure probability.
Furthermore, the probability of each combination of utility choices will be the same.
}
\begin{proof}
$z\ge 5n^3/2\delta\kappa$ will be an integer in our construction.
In order to obtain the equal probabilities in the statement of the lemma, for the purposes of this proof we will calculate all utilities mod 1.

We partition the utility range $[0,1)$ into $z$ intervals each of length $1/z$.
To do this, we round down each component of the utility to the nearest integer multiple of $1/z$.
Note that for each doctor $d$, the probability that it has the same utility for hospitals
$h$ and $h'$ is at most $1/z$: to see this, consider selecting the random utility components
one by one, with the interview value for $h'$ being last; this last choice will create the
utility equality with probability at most $1/z$.
Therefore the probability that a doctor has the same utility for two hospitals is at most
$\tfrac 12 \cdot\frac{n}{\kappa}\cdot\big(\frac{n}{\kappa}-1\big)/z$, and so the probability that any doctor has the same utility for two hospitals is at most
$n\cdot \frac{n}{\kappa}\big(\frac{n}{\kappa}-1\big)/2z\le n^3/2z\kappa^2$.

Similarly, the probability that any hospital has the same utility for two doctors is at most $n^3/2z\kappa$.

Also similarly, the probability that two doctors have the same public value is $1/z$, and hence the probability that any two doctors have the same public value is at most $\frac {n^2}{2z}$. Likewise, the probability that any two hospitals have the same public value is at most $\frac {n^2}{2z\kappa^2}$. Finally, the probability that a doctor has two equal private values is at most $\frac {n^2}{2z\kappa^2}$, so the probability that this happens for any doctor is at most $\frac {n^3}{2z\kappa^2}$.

Therefore, the overall failure probability is at most $5n^3/2z\kappa\le \delta$.

Finally, to see that each combination of utilities has the same probability it suffices to note that because of the circularity due to doing the above calculations mod 1, we are treating each continuous length $1/z$ interval for each component of the utility in exactly the same way.
\end{proof}

\noindent
{\sc \Cref{lem::cont-range-bound}}
    \emph{The bounds of Lemmas \ref{lem::doctors-interval-non-match}--\ref{lem::fail-prob-hospital-prop} with $\delta=0$ hold in the continuous setting too.
    }
\begin{proof}
 For any given $\delta>0$, by \Cref{clm::discrete-dist-idio}, the actions in the continuous setting and the corresponding discrete setting can differ only when there is a difference in the preferences lists, which occurs with probability at most $\delta$.
 Therefore, by \Cref{lem::doctors-high-value-match-prob}, the probability that a doctor $d$ does not receive a proposal with interview value at least $1-\psibar_d$ is at most $\delta+ \exp\big(-k\alpha\psibar_d /(4a+1)\big)$.
 As this holds for any $\delta>0$, the actual probability is at most $\exp\big(-k\alpha \psibar_d /(4a+1)\big)$.
 Similar arguments apply to the bounds in the remaining three lemmas.
\end{proof}

\subsubsection{Bad events}
Let $I$ be an interval of public ratings or a union of such intervals of total length $l$.
We will identify two different pairs of bad events to handle the cases $l\ge \alpha$ and $l<\alpha$, respectively.
For the case $l\ge \alpha$, let $\Bhbf(I)$ and $\Bhbm(I)$ be the events that the number of hospitals in $I$ is smaller than $\big(\ell-\frac12 \alpha\big)\cdot \frac{n}{\kappa}$ and more than $\big(\ell+\frac12 \alpha\big)\cdot \frac{n}{\kappa}$, respectively.
For the case $l< \alpha$, let $\Bhsf(I)$ and $\Bhsm(I)$ be the events that the number of hospitals in $I$ is smaller than $\frac 12\ell\cdot \frac{n}{\kappa}$ and more than $2\ell \cdot \frac{n}{\kappa}$, respectively.
(H is for hospital, b for big and s for small $l$, f for too few and m for too many agents).
We define $\Bdbf(I)$, $\Bdbm(I)$, $\Bdsf(I)$ and $\Bdsm(I)$ analogously for doctors (D for doctor), with $n$ replacing $n/\kappa$ in the previous bounds.

We also want to define analogous events for the number of hospitals in $I$ other than a hospital $h$ known to be in $I$. For brevity, we write $\Bhbf(I\setminus h)$ instead of $\Bhbf(I\setminus\{r(h)\})$.
So, for example, if $h\in I$, we define $\Bhbf(I\setminus h)$ to be the event that the number of hospitals in $I$ is less than $\big(\ell-\frac12 \alpha\big)\cdot \big(\frac{n}{\kappa}-1\big)$, and similarly for $\Bhbm(I\setminus h)$, $\Bhsf(I\setminus h)$, $\Bhsm(I\setminus h)$, and likewise for the analogous events concerning doctors.

\begin{claim}\label{lem::number-agents-in-intervals}
    Suppose $\frac {n}{\ln n}\ge \frac{12(c+2)\kappa \ell}{\alpha^2}$. If $l\ge \alpha$, then  $\Bhbm(I)$ and $\Bhbf(I)$ each occur with probability at most $n^{-(c+2)}$,
    while if $\ell\ge 8(c+2)\kappa\ln n/n$, then $\Bhsm(I)$ and $\Bhsm(I)$ each occur with probability at most $n^{-(c+2)}$.
    
    Similarly, if $\frac {n}{\ln n}\ge \frac{12(c+2)\ell}{\alpha^2}$, and $l\ge \alpha$ then $\Bdbm(I)$ and $\Bdbf(I)$ each occur with probability at most $n^{-(c+2)}$;
    and if $\ell\ge 8(c+2)\ln n/n$, then $\Bdsm(I)$ and $\Bdsf(I)$ each occur with probability at most $n^{-(c+2)}$.

    Analogous bounds hold with $I$ replaced by $I\setminus h$; the one change is that the $n/\kappa$ or $n$ in the relevant bound is replaced by $(n/\kappa)-1$ or $n-1$, respectively.
\end{claim}
\begin{proof}
    The expected number of hospitals in $I$ is $\ell\cdot \frac{n}{\kappa}$.
    By a Chernoff bound, there are more than $\big(\ell+\frac12 \alpha\big)\cdot \frac{n}{\kappa}$ hospitals in $I$ with probability at most $\exp\big(-\frac13\cdot \frac{\alpha^2}{4\ell}\cdot \frac{n}{\kappa}\big)\le n^{-(c+2)}$, if $\frac {n}{\ln n}\ge \frac{12(c+2)\kappa\ell}{\alpha^2}$.
     And by another Chernoff bound,
    there are fewer than $\big[\ell-\frac12\alpha\big]\cdot \frac{n}{\kappa}$ hospitals in $I$ with probability at most $\exp\big(-\frac12\big(1 - \frac {\alpha}{2\ell}\big)^2\ell\cdot \frac{n}{\kappa}\big)\le n^{-(c+2)}$, if $\frac {n}{\ln n}\ge 8(c+2)\kappa l$ (using $\ell\ge \alpha$).
    Similarly, there are more than $2 \ell \cdot\frac n{\kappa}$ hospitals in $I$ with probability at most $\exp\big(-\frac13\ell\cdot \frac n\kappa \big)\le n^{-(c+2)}$, and fewer than $\frac12\ell \cdot\frac n{\kappa}$ hospitals in $I$ with probability at most $\exp\big(-\frac18\ell\cdot\frac n{\kappa}\big)\le n^{-(c+2)}$, if $l\ge 8(c+2)\kappa \ln n/n$.
    \end{proof}

The next corollary shows that, for each doctor and hospital, its cone has a population close to its expectation with high probability.
To make this precise, we introduce the following notation.
Recall that $C(d)$ denotes the cone for doctor $d$; analogously, we let $C(h)$ be the cone for hospital $h$, the
set of doctors with public rating in the range $[r(h)-a\alpha,r(h)+a\alpha)$.
\begin{corollary}\label{cor::cone-population-correct}
    For each doctor $d$ and hospital $h$, the events
    $\Bdbm(C(h))$, $\Bdbf(C(h))$, $\Bhbm(C(d))$, $\Bhbf(C(d))$ each occur with probability at most $n^{-(c+2)}$.
    (Note that the cone for $h$ is an interval of doctors and vice-versa.)
\end{corollary}
\begin{proof}
    It suffices to note that the length of these intervals is typically $2a\alpha$, and at least $a\alpha$ (it can be shorter for the very topmost and bottommost agents), and then apply \Cref{lem::number-agents-in-intervals}.
\end{proof}

In the proofs of Lemmas~\ref{lem::doctors-interval-non-match}--\ref{lem::doctors-high-value-match-prob} we assume that none of these events occur, i.e.\ that all the cones have size fairly close to their expectation.

\subsubsection{Omitted Claim in \Cref{sec::proof-doc-match-prob}}

\begin{claim}\label{clm::conditions-hosp-prop-one-doctor}
 The conditions in Claims~\ref{lem::surplus-hosp-prop-one-doc} and~\ref{lem::fail-prob-hospital-prop} amount to the following constraints: $a\ge 5$,  and $\alpha\cdot\exp\big(\alpha k/\big)\ge 3$.
 They also exclude events with total probability at most $3n^{-(c+1)}$.
 \end{claim}

\subsubsection{Proof of \Cref{lem::hosp-high-value-match-prob} (Doctors Match Probability)}\label{sec::proof-hosp-match-prob}
 
 This is quite similar to the proof of \Cref{sec::proof-doc-match-prob}, except that the roles of the doctors and hospitals are switched.

We will use the following high-probability lower bound $1-\taubar$ on the private values of the edges selected by each doctor $d$.

\begin{claim}\label{clm::value-of-taubar}
    Let $b>2$ and $c>0$ be constants. Suppose that $\Bhbf(C(d))$ does not occur for any doctor $d$.
    Suppose $\taubar\ge \frac{b(c+2)\ln n} {(a-\frac 12)\frac{\alpha n}{\kappa}}$. Then, for each doctor $d$, with failure probability at most $n^{-(c+2)}$, if $k \le \big(1-\big(\frac 2b\big)^{1/2}\big) \cdot b(c+2)\ln n$,
    $d$ has at least $k$ in-cone edges with private value $1-\taubar$ or larger.
\end{claim}
\begin{proof}
As $\Bhbf(C(d))$ does not occur, the number of hospitals in $d$'s cone is at least $(a-\frac 12)\cdot\frac{\alpha n}{\kappa}$.
    The expected number of $d$'s in-cone edges with private value $1-\taubar$ or larger
    is $\taubar$ times the number of hospitals in her cone, i.e.\ at least $\taubar (a-\tfrac12)\cdot \frac{\alpha n}{\kappa}\ge b(c+2)\ln n$,
    as $\taubar \ge \frac{b\kappa(c+2)\ln n} {(a-\frac12)\alpha n}$.
    Then, by a Chernoff bound, the probability that she
    has fewer than $k$ such edges is at most $\exp\big[-\tfrac b2 (c+2)\ln n \big(1 - \frac k{b(c+2)\ln n}\big)^2\big]
    \le n^{-(c+2)}$, if $k \le \big(1-\big(\frac 2b\big)^{1/2}\big) \cdot b(c+2)\ln n$.
\end{proof}

Going forward, we suppose the event $\Btau$, that some doctor fails to have $k$ edges with private value $1-\taubar$ or larger, does not occur.

\begin{claim}\label{lem::surplus-hosp-prop-one-hosp}
    Suppose the following bad events do not occur:
    $\Bhbm([r(h)-\alpha,1))$ and $\Bdbf([r(h)-a\alpha,r(h)+a\alpha)\cup[r(h)+(a+1)\alpha,1))$, and also suppose that both $a\ge 5$ and $\alpha\cdot\exp\big(\alpha k/2(4a+1)\big)\ge 3$.
    Then the surplus $s\ge\alpha n$.
\end{claim}
\begin{proof}
    As $\Bhbm([r(h)-\alpha,1))$ does not occur, the number of hospitals with public rating at least $r(h)-\alpha$ is at most $[1-r(h)+\frac 32\alpha ]\frac{n}{\kappa}$. As $\Bdbf([r(h)-a\alpha,r(h)+a\alpha)\cup[r(h)+(a+1)\alpha,1))$ does not occur,
    the number of doctors with public rating in this range is at least $[1-r(h)+(a-\frac32)\alpha] n$.
    By \Cref{lem::doctors-interval-non-match}, the number of unmatched doctors in $[r(h)+(a+1)\alpha,1)$ is at most
    $3n\cdot\exp\big(-\alpha k/2(4a+1)\big)\le  \alpha n$, if $\exp\big(- \alpha k/2(4a+1)\big)\le \frac 13\alpha$.
    Thus the surplus, the number of unmatched available doctors in $h$'s cone, i.e., doctors that could propose to $h$, is at least
    \begin{align*}
        \big[1-r(h)+(a-\tfrac32)\alpha\big] n - \big[1-r(h)+\tfrac 32\alpha\big]n -  \alpha n
        \le \big(a - 4\big)\alpha n\le \alpha n,
    \end{align*}
    if $a\ge 5$.
\end{proof}

As in \Cref{sec::proof-doc-match-prob}, the tree $T$ of possible runs of the double-cut DA has
both real and imaginary proposals. For each edge $(d',h')$, for any doctor $d'$ outside $C(h)$, we fix the utilities for both $d'$ and $h'$.
Also, for each doctor $d$ in $C(h)$, we fix its top $k$ private values, but without specifying which edges these values are attached to.
Now, we are seeking the probability that $h$ receives fewer than $\kappa$ proposals from the doctors in $C(h)$. Again, we obtain the bound by showing that $\qtilde(v) \ge q(v)$ for all nodes $v$ in $T$, and for $\treeroot(T)$ in particular, but the meanings of $q$ and $\qtilde$ are slightly different here.

A node corresponding to an action by a doctor outside $C(h)$ has just one outgoing edge corresponding to its single possible action. For a doctor $d$ in $C(h)$, the set-up is more elaborate. Let $1-\tau_1\ge 1-\tau_2\ge \ldots\ge 1-\tau_k$ be the private values for its $k$ edges; let's name the corresponding proposals $e_1,e_2,\ldots,e_k$. $T$ will have a series of nodes corresponding to deciding in turn, whether each $e_i$ provides $d$ utility equal to its current utility (this is the constraint on the tie-breaking mentioned earlier; we use the private values to order edges that produce equal utility). For each potential proposal, the decision making proceeds in two steps:
first, whether the proposal has a utility equal to $d$'s current utility value in tree $T$;
note that we have not yet chosen the hospital that receives this proposal. 
%This is a decision based on making a random choice for its interview value. 
For each hospital $h$ for which the current utility is possible for $d$, which we call $d$'s current cone, the probability that $h$ receives this proposal is the same.
In the $q$ setting, we then choose the edge uniformly from the hospitals in $d$'s current cone that she has not already proposed to; in this setting, there are no imaginary proposals.
A doctor stops making proposals either when she has gone through all $k$ of her proposals or she makes a proposal to $h$.
Let $x$ be the number of proposals $h$ has received prior to reaching node $v$.
Then $q(v)$ is the probability that $h$ receives at most $\kappa -x -1$ further proposals in the
computation starting at $v$.

In the $\qtilde$ setting, we keep track of how many proposals a doctor has made.
Again, as soon as a doctor makes a proposal, real or imaginary, to $h$, she stops.
Let $S(v)$ be the set of doctors in $h$'s cone that have not yet proposed to $h$,
let $x$ be the number of real proposals to $h$, and $y$ the number of imaginary proposals to $h$.
Let $k(d,v)$ be the number of proposals remaining for doctor $d$ at node $v$. Also, recall that $u(d,v)$ is her residual utility at node $v$.
To evaluate $\qtilde(v)$, the doctors with the $s-x-y$ smallest products $u(d,v)\cdot k(d,v)$ will attempt to propose to $h$; it will receive a real proposal from doctor $d$ with probability $u(d,v)\cdot k(d,v)/[(2a+\frac 12)\alpha n/\kappa]$, which will be no larger than the probability of a proposal in the $q$ setting; the difference in probability will be handled by making proposals imaginary with probability equal to the difference. Now, we make this precise.

We use the same value for the probability that $d$ has a proposal with the current utility $u$ as in the $q$ setting. 
If $d$ is found to have such a proposal, 
then, uniformly at random, we choose one of the hospitals in $d$'s current cone that has not already received a proposal from $d$ to receive the proposal.
If the proposal is to a hospital $h'\ne h$, then it is real.
If it is to $h$,
we choose whether it is real with probability $\tau=\frac{k(d,v)}{(2a+\frac12)\alpha n/\kappa}$ (the reason for this probability is to avoid having to account for the possibly varying cones sizes when selecting $s-x-y$ doctors that attempt to propose); if the proposal is not real, we make it imaginary.
 
We note that  $\ptilde(\wrl) \le p(\wrl)$: for the edge is real in the $\qtilde$ setting with probability
$\tau$; this value is computed assuming the cone sizes for all $d\in C(h)$ are as large as possible;
this only reduces the probability of the edge being real.
Then we set $\ptilde(\wim)=p(\wrl)-\ptilde(\wrl)$.
% As in \Cref{sec::doc-int-idio-no-match}, $\tildep(\wrl(h))\le p(\wrl(h))$; again, we choose $\tildep(\wim(h))=p(\wim(h))+[\wrl(h))-\tildep(\wrl(h)]$.

%Again, let $x$ be the number of hospitals in $I'$ that have received real proposals.
%At node $v$, we will be seeking the probability that at most an additional $r-x$ hospitals in $I'$ receive real proposals.

\begin{claim}\label{clm::qtilde-larger-for-intervals-in-doc-prop}
    $\qtilde(v)\ge q(v)$ for all $v$ in $T$.
\end{claim}
\begin{proof}
    The base case argument is unchanged from \Cref{clm::fixed-k-stoc-dom-idio}.

    We now make some observations for the case that $v$ is a non-leaf node. $v$ will have one child $\wna$ corresponding to no-action. For each hospital $h'\ne h$ it will have one child, $w(h')$, corresponding to a proposal to $h'$, and for $h$ it will have two children $\wrl(h)$ and $\wim(h)$, $\wrl$ and $\wim$ for short, corresponding to real and imaginary proposals to $h$.
    We let $p(w)$ denote the probability of taking the edge to node $w$ in the $q$ setting, and $\ptilde$ the probability in the $\qtilde$ setting. We note that by construction, $\ptilde(w)=p(w)$ for all nodes other than $\wrl,\wim$.
    For these nodes, we have $\ptilde(\wrl) \le p(\wrl)$ and $\ptilde(\wim)=p(\wrl)-\ptilde(\wrl)$. Also, $\qtilde(wrl)\le \qtilde(\wim)$
    for in both settings the doctors have the same $k$ and $u$ values, but in the subtree rooted at $\wim$ we need to obtain one more real proposal than in the subtree rooted at $\wrl$, and so the failure probability, $\qtilde(\wim)$, must be at least as large as $\qtilde(\wrl)$.

\hide{
    \smallskip\noindent
    \emph{Case A}. The proposal is to some $h'\ne h$:
    
    Then $v$ has two children $\wrl(h')$, $\wim(h')$, $\wrl$ and $\wim$ for short.
    Here $k(d,w)=k(d,v)-1$ and $u(d,w)<u(d,v)$, for $w=\wrl,\wim$; the variables associated with the other doctors are unchanged. At both children, we seek the probability that $h$ receives at most $\kappa-x-1$ proposals. 

    As already noted, $\ptilde(\wrl) \le p(\wrl)$ and $\ptilde(\wim)=p(\wrl)-\ptilde(\wrl)$.
       In addition, $\qtilde(\wrl)\le \qtilde(\wim)$, for in both settings the doctors have the same $k$ and $u$ values, but in the subtree rooted at $\wim$ we need to obtain one more real proposal than in the subtree rooted at $\wrl$, and so the failure probability, $\qtilde(\wim)$, must be at least as large as $\qtilde(\wrl)$. (One may wonder at this seemingly unnecessarily complex setup; but it is used later in the proof of \Cref{lem::fail-prob-doctor-prop}).

Now, we deduce our result with the help of a little algebra.
\begin{align*}
    \ptilde(\wim)\cdot \qtilde(\wim) + \ptilde(\wrl)\cdot \qtilde(\wrl) 
    & \ge   [\ptilde(\wim) + \ptilde(\wrl)]\cdot \qtilde(\wrl)\\
    &=p(\wrl)\cdot \qtilde(\wrl)\\
    &\ge p(\wrl)\cdot q(\wrl).
\end{align*}
}

    \smallskip\noindent
        At $\wrl$ we are seeking the probability of at most $\kappa-x-1$ real proposals to $h$, and at the remaining nodes  the probability of at most $\kappa-x$ real proposals to $h$.
        At $\wrl$ and $\wim$ the proposals come from the $s-x-y-1$ doctors with the smallest $u$--$k$ products, and at the remaining nodes from the $s-x-y$ such doctors.

     \smallskip\noindent
     \emph{Case 1}. $u(d,v)\cdot k(d,v)$, the proposal probability for $h$ at node $v$, is among the $s-x-y$ smallest proposal probabilities, $u_1k_1\le u_2k_2\le\ldots\le u_{s-x-y}\cdot k_{s-x-y}$, at node $v$. \\
     Since only $u(d,w)$ and $k(d,w)$ change (by becoming smaller) at any of $v$'s children $w$, the proposal probability for $d$ continues to be among the $s-x-y$ smallest at $w\ne \wrl,\wim$, the children where it is still in use. 
 %    $\ptilde(\wna)=p(\wna)$, and as in case A, $\ptilde(\wrl) \le p(\wrl)$,  $\ptilde(\wim)=p(\wrl)-\ptilde(\wrl)$, and $\qtilde(\wrl)\le \qtilde(\wim)$.
     Therefore instantiating the randomness needed to choose an edge to a child of $v$ ensures that 
     \begin{align*}
     \qtilde(v)&=\ptilde(\wrl)\cdot\qtilde(\wrl)+\ptilde(\wim)\cdot \qtilde(\wim) + \sum_{w\ne \wrl,\wim}\ptilde(w)\cdot \qtilde(w)\\
     &\ge [\ptilde(\wrl)+\ptilde(\wim)]\cdot \qtilde(\wrl)+ \sum_{w\ne \wrl,\wim}\ptilde(w)\cdot \qtilde(w)\\
     &=p(\wrl)\cdot \qtilde(\wrl) +  \sum_{w\ne \wrl,\wim} p(w)\cdot \qtilde(w) \\
     &\ge p(\wrl)\cdot q(\wrl) + \sum_{w\ne \wrl,\wim}p(w)\cdot q(w)=q(v).
     \end{align*}

     \smallskip\noindent
     \emph{Case 2}. $u(d,v)\cdot k(d,v)$, the proposal probability for $h$ at node $v$, is not among the $s-x-y$ smallest proposal probabilities at node $v$. \\
     The argument is very similar to the Case 2 argument in \Cref{clm::fixed-k-stoc-dom-idio}; the only change is to the breakpoints.
     First, we increase $u_{s-x-y}$ to $u(d,\wna)$ as before (unless it is already larger).
     Then, we continue to increase $u_{s-x-y}$ until 
     $u_{s-x-y}\cdot k(d',\wna)= u(d,\wna)\cdot k(d,\wna)$, where $d'$ is the doctor with utility $u_{s-x-y}$.
\end{proof}

\begin{claim}\label{lem::fail-prob-doctor-prop}
    Let $h$ be a hospital with public rating at least $a\alpha$.
    Suppose that $s\ge \alpha n$ and $\Bdbm(C(h))$ does not occur. Then the probability that in the $\qtilde$ setting $h$ receives fewer than $\kappa$ proposals for which $h$ has interview value at least $1-\psibar_h$ is at most $\exp\big(-3\kappa k\alpha\psibar_h/8(4a+1)\big)$, if $k\ge \frac{(2+\sqrt{2})\cdot(4a+1)}{2(\alpha-\taubar)\psibar_h}$ and $\taubar \le \frac {\alpha}{4}$.
\end{claim}
\begin{proof}
    For each doctor $d$, her residual utility at the root of $T$ is at least $\alpha-\taubar$.
    The probability that $d$ proposes to $h$ in the $\qtilde$ setting is 
    $\frac{(\alpha-\taubar) k}{(2a+\frac12)\alpha n/\kappa}$.
     Therefore, the expected number of doctors in $S(\treeroot(T))$ that selected $h$ is at least $2k\alpha n(\alpha-\taubar)/(4a+1)\alpha n= 2\kappa k(\alpha-\taubar)/(4a+1)$.
    If in addition, we restrict the edges to those for which $h$ has interview value at least $1-\psibar_h$, the expected number becomes
    $2\kappa k(\alpha-\taubar)\psibar_h/(4a+1)$.
    Due to the discretization, these edges are negatively correlated, and so
    by a Chernoff bound, the probability that fewer than $\kappa$ proposals to $h$ are made is at most
        $\exp\big(-\frac 12\big[1 -\frac{4a+1}{2k(\alpha-\taubar)\psibar_h}\big]^2\cdot 2\kappa k(\alpha-\taubar)\psibar_h/(4a+1)\big)\le \exp\big(-\kappa k(\alpha-\taubar)\psibar_h/2(4a+1)\big)$, if $k\ge \frac{(2+\sqrt{2})\cdot(4a+1)}{2(\alpha-\taubar)\psibar_h}$. Substituting $\taubar \le \frac{\alpha}{4}$ yields the claimed bound.
\end{proof}

\begin{claim}\label{clm::conditions-doctor-prop-one-hosp}
 The conditions in Claims~\ref{clm::value-of-taubar}---\ref{lem::fail-prob-doctor-prop} amount to the following constraints: $a\ge 5$, $b>2$, $\taubar\ge \frac{b(c+2)\ln n}{(a-\frac12)\frac{\alpha n}{\kappa}}$, $\frac{(2+\sqrt{2})\cdot(4a+1)}{2(\alpha-\taubar)\psibar_h} \le k\le\big[1-\big(\frac 2b\big)^{1/2}\big]\cdot b(c+2)\ln n$, and $\alpha\cdot\exp\big(\alpha k/2(4a+1)\big)\ge 3$.
 They also exclude events with total probability at most $\big(2n+\frac{3n}{\kappa} \big)\cdot n^{-(c+2)}$.
 \end{claim}

\subsubsection{Omitted proofs for \Cref{lem::doctors-interval-non-match} (Match Probabilities for an Interval of Doctors)}

\begin{claim}\label{clm::hosp-removal-lots-edges}
Suppose that the following events do not occur: $\Bdsf(I)$ and $\Bhbf(C(d))$ for all $d\in I$.
Then at most $|I|\exp\big(-\alpha\kappa k/2(4a+1)\big)$ hospitals are removed because they have 2 or more edges to doctors in $I$, 
at most $|I|\exp\big(-\alpha k/2(4a+1)\big)$ doctors are incident on the removed edges, and at most $|I|k\exp\big(-\alpha k/2(4a+1)\big)$ hospitals and edges incident on $I$ are removed, with failure probability at most $n^{-(c+2)}$, if\\ $4(c+2)\cdot\exp\big(\alpha k/2(4a+1)\big)\cdot \frac{\ln n}{n} \le g-f \le 2(a-\frac12)\alpha\cdot\exp\big(-\alpha k/2(4a+1)\big)/(e^2k^2\kappa)$.
\end{claim}
\begin{proof}
  We begin by giving a bound the number of edges from $I$ that collide, i.e., that have the same hospital as an endpoint.
  As $\Bhbf(C(d))$ does not occur for any $d\in I$, $|C(d)| \ge (a-\frac12)\alpha n/\kappa$ for each
  $d\in I$. For brevity we let $m$ denote $(a-\frac12)\alpha n/\kappa$.
  The probability that an edge from $I$ collides with any of the other edges out of $I$ is at most
  $\frac{(|I|-1)k}{m}\le \frac{|I|k}{m}=\frac{|I|k\kappa}{(a-\frac12)\alpha n}$. To give a high probability bound on the number of colliding edges
  we use stochastic dominance: we consider $|I|k$ 0/1 random variables which are 1 with respective probabilities
  $\frac{i\kappa}{(a-\frac12)\alpha n}$, for $0\le i \le |I|k-1$. 
  (We can think of this as arising by adding the edges one by one; the discretization only reduces the collision probability.)~
  Their sum has expectation bounded by $b=\frac{|I|^2k^2\kappa}{2(a-\frac12)\alpha n}$. Hence, by a Chernoff bound, the probability that the sum is greater than $s=\frac12|I|\exp\big(-\alpha k/2(4a+1)\big)$ is at most
  \begin{align*}
      \Big(\frac {e^{(s/b) -1}}{(s/b)^{s/b}}\Big)^{b}
      \le \Big(\frac{eb}{s}\Big)^{s} 
      = \Big(\frac{e|I|k^2\kappa\cdot\exp\big(\alpha k/2(4a+1)\big)}{(a-\frac 12) \alpha n}\Big)^{\frac12|I|\exp(-\alpha k/2(4a+1))}.
  \end{align*}
  This probability is maximized when $|I|$ is as small as possible, and as $\Bdsf(I)$ does not occur,
  this is when $|I|=\frac 12(g-f)n$.
  Then the probability is bounded by $n^{-(c+2)}$
  if $g-f\ge 4(c+2)\exp(\alpha k/2(4a+1))\cdot \frac{\ln n}{n}$ (which causes the exponent in the above expression to be at least $(c+2)\ln n)$, and $g-f \le 2(a-\frac12)\alpha\exp\big(-\alpha k/2(4a+1)\big)/(e^2k^2\kappa)$ (which causes the base term to be at most $1/e$).
  
  By stochastic dominance, $\frac12|I|\exp\big(-\alpha k/2(4a+1)\big)$ is also a bound on the number of collisions, and hence $|I|\exp\big(-\alpha k/2(4a+1)\big)$ bounds the number of edges that collide, the 
  number of hospitals with two or more edges into $I$ and the number of doctors incident on colliding edges. Therefore, in total, at most $|I|k\exp\big(-\alpha k/2(4a+1)\big)$ edges and hospitals are removed.
\end{proof}

Let $\Bcol$ denote the event that more than $|I|k\exp\big(-\alpha k/2(4a+1)\big)$ hospitals or edges incident on $I$ are removed. Going foward, we suppose that $\Bcol$ does not occur.

The nodes and edges in the tree $T$ of possible computations in the run of the double cut DA need to be redefined. Now we partition proposals according to whether they from hospitals that are outside the cones for doctors in $I'$ and the hospitals that can propose to the doctors in $I'$. 
We fix the utilities for the first category of proposals and consequently
a node corresponding to such a proposal will have a single child.
Once more, proposals to doctors in $I'$ can be real or imaginary. 
% Again, the need for imaginary proposals is due to selecting $h$'s utility before determining whether there is an edge from $I'$ to $h$.

We will need to redefine $q(v)$ and $\qtilde(v)$, as follows.
We let $x$ denote the number of doctors in $I'$ who receive actual or real proposals before the computation reaches node $v$.
Each time a real proposal is made to a doctor $d$ in $I'$, we might as well remove the remaining $k-1$ edges to $d$ as they will not cause any new doctors to receive a proposal.
This results in the removal of $k$ hospitals: the one making the real proposal, and the other $k-1$ with edges to $d$.
Next, we let $y$ denote the number of imaginary proposals to doctors in $I'$ who have not received a real proposal when node $v$ is reached. Again, we will stop the computation of the proposing hospital following any imaginary proposal to a doctor in $I'$. So at node $v$, $kx+y$ hospitals will have had their runs stopped. (The reason for the slightly convoluted definition of $y$ is that the $kx$ term may include some imaginary proposals.)
        
        Then $q(v)$ is defined to be the probability that at node $v$ the continuation of the run of the double-cut DA results in at most $r-x$ additional real proposals to distinct doctors in $I'$. As before, let $h$'s \emph{residual utility} at node $v$, $u(h,v)$, equal the current utility $u$ minus the minimum allowed utility in this run of the double cut utility, namely $u-(g+1-\alpha)$. Let $s$ be the minimum surplus over all the possible paths in the computation assuming no bad events occur.
 %   Let $S(v)$ be the set of the $s-kx-y$ hospitals with the smallest residual utilities     among those which have not made a real or imaginary proposal to a doctor in $I$ before reaching $v$.
    
    To define $\qtilde(v)$ we consider the following setup. 
    Suppose that the $s-kx-y$ hospitals with the smallest residual utilities independently attempt to make a proposal, real or imaginary, with probability corresponding to their residual utility to a hospital in $I'$ that has not yet received a real proposal. Suppose that, for each proposal, we then select whether it is real with probability $\tau =\frac{k(|I'|-x)}{(2a+\frac12)\alpha n/\kappa}$, where $|I'|$ is the number of doctors in $I'$ and $(2a+\frac 12)\alpha n/\kappa$ is an upper bound on the number of hospitals in $\cap_{d\in I} C(d)$, the intersection of the cones for $d\in I$, assuming $\Bhbm(C(d))$ does not occur for any $d\in I$. 
Then $\qtilde(v)$ is defined to be the probability that this process results in proposals to at most $r-x$ more doctors in $I'$.
\begin{claim}\label{clm::qtilde-bound-for-hosp-doc-int}
    $\qtilde(v)\ge q(v)$ for all $v$ in $T$.
\end{claim}
\begin{proof}
The proof that $\qtilde(v)\ge q(v)$ for all nodes $v$ is similar to the proof in \Cref{lem::doctors-high-value-match-prob}. There are three significant changes.
We no longer distinguish low-real and high-real proposals; rather than determining the probability of zero high-real proposals, we want to determine the probability of at most $r$ real proposals to distinct doctors in $I'$.
The second change is related: rather than a node $v$ having the four children $\wrh$,
$\wrl$, $\wim$, corresponding to high-real, low-real, and imaginary proposals to $d$, 
it will have children $\wrl(d)$, $\wim(d)$ for each $d\in I'$, corresponding to real and imaginary proposals to $d\in I'$; at these children, the goal will be to bound the probability that the additional number of doctors in $I'$ receiving real proposals is bounded by, respectively, $r-x-1$ and $r-x$.
It will continue to have children $w(d')$, for $d'\ne d$, corresponding to a proposal to $d'$, and $\wna$ corresponding to no action; for these children,
the goal is to bound the probability of at most an additional $r-x$ doctors receiving real proposals.

The base case, the case of a node with a single child, and case 2 are unchanged.
We give a new proof for the remaining case, namely case 1.

\smallskip\noindent
\emph{Case 1.} $u(d,v)$ is among the smallest $s-kx-y$ residual utilities at node $v$\\
This immediately implies that $u(d,w)$ is among the smallest $s-kx-y$ residual utilities at $w\ne \wrl(d),\wim(d)$ for $d\in I'$, as $u(h,w)<u(h,v)$, while the residual utilities for the other doctors are unchanged; we also note that they are unchanged at $\wrl(d)$ and $\wim(d)$ for $d\in I'$.

Again, we consider the random choices that are made in evaluating $\qtilde(v)$,
by instantiating exactly the randomness needed to select a child $w$ of $v$.
As before, $\qtilde(v)=\sum_{\text{$w$ a child of $v$}} \ptilde(w)\cdot \qtilde(w)$,
where $\ptilde(w)$ is the probability that the edge to $w$ is taken.

Now, $\ptilde(w)=p(w)$ for $w\ne \wrl(d),\wim(d)$, for $d\in I'$.
Also, for all $d\in I'$, $\ptilde(\wrl(d)) \le p(\wrl(d))$: for the edge is real in the $\qtilde$ setting with probability
$\tau$; this value is computed assuming the cone sizes for all $d\in I'$ are as large as possible;
this only reduces the probability of the edge being real.
Therefore, we can define $\ptilde(\wim(d))=p(\wrl(d))-\ptilde(\wrl(d))$ (for on all the paths to nodes $\wim(d)$, $h$'s computation is terminated, which implies that the subtrees at each node
$\wim(d)$ are isomorphic,  meaning that all that matters is the total probability weight assigned to these trees, and not how the probability weight is distributed among them, or equivalently, among
the edges to the nodes $\wim(d)$).

In addition, $\qtilde(\wrl(d))\le \qtilde(\wim(d))$. To see this, note that the bad event captured by
$\qtilde(\wrl(d))$ is having proposals to $r-(x+1)$ or fewer doctors in $I'$ and the bad event captured by
$\qtilde(\wim(d))$ is having proposals to $r-x$ or fewer doctors in $I'$. But $\wim(d)$ has all the proposers
present in $\wrl(d)$ plus up to another $k-1$ in addition, who have edges to  doctor $d$. So whenever the random choices result in at most $r-(x+1)$ doctors receiving proposals at $\wrl(d)$, the same random choices at $\wim(d)$ result in at most $r-(x+1)$ doctors receiving proposals plus, in addition, $d$ possibly receiving one or more proposals, for a total of at most $r-x$ doctors receiving proposals.

Now, we deduce our result with the help of a little algebra.
\begin{align*}
   &\sum_{d\in I'}\big[\ptilde(\wim(d))\cdot \qtilde(\wim(d)) + \ptilde(\wrl(d))\cdot \qtilde(\wrl(d))\big] +\sum_{d\ne \wrl(d),\wim(d)}\ptilde(w)\cdot \qtilde(w)\\
    &\hspace*{0.3in} \ge   \sum_{d\in I'}\big[\ptilde(\wim(d))+\ptilde(\wrl(d)) \big]\cdot \qtilde(\wrl(d)) +\sum_{d\ne \wrl(d),\wim(d)}p(w)\cdot \qtilde(w)\\ 
  %  &\hspace*{1in}+ \big[\ptilde(\wim(d))-p(\wim(d))\big]\cdot \big[\qtilde(\wim(d))-\qtilde(\wrl(d))\big]\\
    &\hspace*{0.3in} \ge \sum_{d\in I'} p(\wrl(d))\cdot q(\wrl(d))  +\sum_{d\ne \wrl(d),\wim(d)}p(w)\cdot q(w) \\
    &\hspace*{0.3in}=q(v).
\end{align*}
\end{proof}

\begin{claim}\label{qtilde-bound-at-root-hosp-to-doc-int}
Let $I=[f,g)$, and suppose that $s\ge \alpha n/\kappa$, 
$6\exp\big(\frac{\alpha k}{2(4a+1)}\big) \cdot (c+2)\frac{\ln n}{n} \le g-f \le  \frac {\alpha}{2}$.
Also, suppose that the following events do not occur: $\Bcol$ and $\Bdsm(I)$.
Then the probability that the number of unmatched doctors is more than
    $3|I|\cdot \exp\big(-\frac{\alpha k}{2(4a+1)}\big)$ is at most $n^{-(c+2)}$.
\end{claim}
\begin{proof}
We seek to identify as small an $r$ as possible such that $\qtilde(\treeroot(T))\le n^{-(c+2)}$,
for we can then deduce that $q(\treeroot(T))\le n^{-(c+2)}$ for this value of $r$.
The remainder of the analysis concerns $\qtilde(\treeroot(T))$.

In the $\qtilde(\treeroot(T))$ setting, recall that the utility of the hospital for a doctor in $I'$
is restricted to be in the range $[g+1-\alpha,f+1)$, a range of size $\alpha-(g-f)$.
Therefore, at the root of $T$, $s\ge \alpha n/\kappa$ hospitals each seek to propose to some unspecified doctor in $I'$ with probability at least
\begin{align*}
[\alpha-(g-f)]\cdot \frac {2|I'| k}{(4a+1)\alpha n/\kappa}
 \ge \frac {\alpha|I'| k}{(4a+1)\alpha n/\kappa},
\end{align*}
using $g-f\le \frac {\alpha}{2}$.

Thus the expected number of proposals to a single doctor $d$ in $I'$ is at least
$\frac{\alpha k}{(4a+1)}$.
By a Chernoff bound, $d$ receives no proposals with probability at most
$\exp\big(-\frac12 \cdot\frac{\alpha k}{(4a+1)} \big)$.
Let $X_d$ be the 0/1 random variable that equals 1 if $d$ receives no proposals.
The $X_d$ are negatively correlated and therefore we can apply Chernoff bounds to their sum.
In expectation, the sum of the $X_d$ is at most
$|I'|\cdot \exp\big(-\frac{\alpha k}{2(4a+1)}\big)$, and by a Chernoff bound, it is more than
$2|I'|\cdot \exp\big(-\frac{\alpha k}{2(4a+1)}\big)$
with probability at most $\exp\big(-\frac 23 |I'|\cdot \exp\big(-\frac{\alpha k}{2(4a+1)}\big)\big)\le n^{-(c+2)}$, if $|I'| \ge \frac 32 \exp\big(\frac{\alpha k}{2(4a+1)}\big) \cdot (c+2)\ln n$.
By \Cref{clm::hosp-removal-lots-edges}, if $\Bcol$ and $\Bdsf(I)$ do not occur, and if $\exp\big(-\alpha k/2(4a+1)\big)\le \frac12$, then $|I'|\ge |I|\big[1-\exp\big(-\alpha k/2(4a+1)\big)\big]\ge \frac 12|I|\ge \frac 14 (g-f)n$;
then $g-f \ge 6\exp\big(\frac{\alpha k}{2(4a+1)}\big) \cdot (c+2)\frac{\ln n}{n}$ suffices.  

On setting $r=2|I'|\cdot \exp\big(-\frac{\alpha k}{2(4a+1)}\big)$, we deduce that at most
$2|I'|\exp\big(-\frac{\alpha k}{2(4a+1)}\big)$ doctors in $|I'|$ are unmatched, with failure probability
at most $n^{-(c+2)}$. As $\Bcol$ does not occur, we know that $|I|-|I'| \le |I|\exp\big(-\frac{\alpha k}{2(4a+1)}\big)$, from which our result follows.
\end{proof}

For the next claim, the following additional notation will be helpful.
$J=[g+(a+1)\alpha,1)\cup [g-a\alpha,f+a\alpha)$ is the set of participating hospitals excluding those in $[g+a\alpha,g+(a+1)\alpha)$,
and $K=[g-\alpha,1)\setminus I$ is the set of participating doctors excluding those in $I$.

\begin{claim}\label{clm::surplus-interval-hosp-prop}
    \emph{Suppose that none of the following events occur: 
    $\Bdsm(I)$, $\Bhbf(J)$, $\Bhbm(K)$.
    Then the surplus $s\ge \alpha n/\kappa$, if $\alpha\cdot\exp\big(\alpha k/2(4a+1)\big)\ge 3$,
    $g-f\le \frac32 (a+1)\alpha/\kappa k$, 
    and $a\ge 5$.
    }
\end{claim}
\begin{proof}
We need to consider two cases, depending on whether the span of the cone for some doctor in $I$  overlaps the range $[1-\alpha,1)$, or in other words whether or not $g\le 1-(a+1)\alpha$.
\\
    \emph{Case 1}. $g\le 1-(a+1)\alpha$.\\
        As $\Bhbf(J)$ does not occur, $J$ contains at least $\big[1-g-(a+1)\alpha)+(2a\alpha- [g-f]) -\frac 12\alpha\big]\cdot \frac{n}{\kappa}= \big[(1-g) +(a-\frac32)\alpha -(g-f)\big]\cdot \frac{n}{\kappa}$ hospitals.
        By \Cref{lem::hospitals-interval-non-match}, at most a fraction $3\cdot \exp\big(-\frac{\alpha k}{2(4a+1)}\big)$ of the hospitals in $[g+(a+1)\alpha,1)$ are not fully matched.
        % (strictly, the fraction in a minimal set of intervals covering $[g+a\alpha,1)$).
        By \Cref{clm::hosp-removal-lots-edges}, the number of hospitals that are eliminated is at most $|I|k\exp\big(-\alpha k/2(4a+1)\big)\le 2(g-f)kn\cdot\exp\big(-\alpha k/2(4a+1)\big)$, as $\Bdsm(I)$ does not occur; we want this to be bounded by $3(a+1)\alpha \frac{n}{\kappa}\cdot\exp\big(-\alpha k/2(4a+1)\big)$, which holds
        if $g-f\le \frac32 (a+1)\alpha/\kappa k$.
        Then, the number of not-fully matched and eliminated hospitals in the intervals
        $[g-a\alpha,f+a\alpha)\cup[g+(a+1)\alpha,1)$ is at most $[1-(g+(a+1)\alpha)+(a+1)\alpha)]\frac{n}{\kappa}\cdot 3\exp\big(-\alpha k/2(4a+1)\big)\le \alpha \frac{n}{\kappa}$, if $3\le \alpha \exp\big(\alpha k/2(4a+1)\big)$. This leaves at least $\Big(\big[1-g)+(a-\frac 32)\alpha- (g-f)\big]- \alpha\Big)\frac n{\kappa}$ that are either fully matched or in the interval $[g-a\alpha,f+a\alpha)$.
        
        As $\Bdbm(K)$ does not occur, there are at most $[1-g+\frac32\alpha-(g-f)]n$ participating doctors outside of $I$.
        Therefore, the number of not-fully matched hospitals in $[g-a\alpha,f+a\alpha)$ is at least 
    $\Big(\big[1-g)+(a-\frac 32)\alpha- (g-f)\big]- \alpha -\big[1-g+\frac32\alpha-(g-f)\big]\Big)\cdot \frac{n}{\kappa}\ge (a-4)\frac{\alpha n}{\kappa}\ge \frac{\alpha n}{\kappa}$.
            \\
    \emph{Case 2}. $g> 1-(a+1)\alpha$.\\
    As in Case 1, $K$ contains at most $\big[1-g+\frac32\alpha-(g-f)\big]n\le \big[(a+\frac52)\alpha -(g-f)\big]n$ doctors.
     Now, $J=[g-a\alpha,f+a\alpha)$.
    Similarly to Case 1, the number of hospitals in $[g-a\alpha,f+a\alpha)$ that are remain after the elimination
    is at least $\big(\big[(2a-\frac12)\alpha-(g-f)\big] -\alpha\big)\cdot\frac {n}{\kappa}$.
    So the surplus is at least
    $\big(\big[(2a-\frac32)\alpha-(g-f)\big] - \big[(a+\frac52)\alpha -(g-f)\big] \big)\cdot\frac {n}{\kappa}
    \ge (a-4)\frac {\alpha n}{\kappa}\ge \frac {\alpha n}{\kappa}$.
    \end{proof}

\begin{claim}\label{clm::hosp-prop-interval-constraints}
    The conditions in Claims~\ref{clm::hosp-removal-lots-edges}--\Cref{clm::surplus-interval-hosp-prop} amount to the following constraints:
    \begin{align*}
        \alpha\cdot \exp\Big(\frac {\alpha k}{2(4a+1)}\Big)&\ge 3~~~~(\Cref{clm::surplus-interval-hosp-prop})\\
        a&\ge 5~~~~(\Cref{clm::surplus-interval-hosp-prop})
   %     n&\ge \frac {e(|I|k\kappa)^2}{2\tildec[(a-\frac12)\alpha]^2}~~~~(\Cref{clm::hosp-removal-lots-edges})
    \end{align*}
    and several constraints on $g-f$, namely:
    \begin{align*}
        6(c+2)\exp\Big(\frac{\alpha k}{2(4a+1)}\Big)\cdot\frac{\ln n}{n}
        \le g-f
        &\le \min\Big\{\frac{(2a-1)\alpha\cdot\exp\big(-\frac{\alpha k}{2(4a+1)}\big)}{(ek)^2\kappa},~\frac{\alpha}{2},~\frac{3(a+1)\alpha}{2\kappa k}\Big\}\\
       & (\text{resp.\ Claims \ref{clm::hosp-removal-lots-edges} and \ref{qtilde-bound-at-root-hosp-to-doc-int}, \ref{clm::hosp-removal-lots-edges}, \ref{qtilde-bound-at-root-hosp-to-doc-int}, \ref{clm::surplus-interval-hosp-prop}})\\
       &\le \frac{(2a-1)\alpha\cdot\exp\big(-\frac{\alpha k}{2(4a+1)}\big)}{(ek)^2\kappa},~~~~~~~~~~~~\text{if $(ek)^2 \ge 2(2a-1)$}.
    \end{align*}
    In addition, over the whole induction, they exclude events with total probability at most $\big( \frac{6}{g-f} + \frac{n}{\kappa}\big)\cdot n^{-(c+2)}$.
\end{claim}
\begin{proof}
By inspection, 
the excluded events in \Cref{clm::hosp-removal-lots-edges} have probability at most $\big(\frac{n}{g-f}+\frac{n}{\kappa}\big)\cdot n^{-(c+2)}$ and it has failure probability at most $\frac{1}{g-f}\cdot n^{-(c+2)}$;
the additional excluded events in \Cref{qtilde-bound-at-root-hosp-to-doc-int} have probability at most
$\frac{1}{g-f}\cdot n^{-(c+2)}$ and this claim has failure probability at most $\frac {1}{g-f}\cdot n^{-(c+2)}$;  
the additional excluded events in \Cref{clm::surplus-interval-hosp-prop} have probability at most $\frac{2}{g-f} \cdot n^{-(c+2)}$. The bound on the probability of excluded events follows readily.
\end{proof}

\subsubsection{Proof of \Cref{lem::hospitals-interval-non-match} (Match Probabilities for Intervals of Hospitals)}

$I=[f,g)$ is an interval of hospitals here. We define $L=[g-a\alpha,f+a\alpha)$ to be the interval of doctors that are in-cone for all hospitals in $I$. We want to ensure that there are a sufficient number
of potential proposals from $L$ that are uniform over $I$ w.r.t.\ the interview values. Suppose we guarantee, with failure probability at most $n^{-(c+2)}$, that the doctors' proposals all have private value at least $1-\taubar$ (bounded using \Cref{clm::value-of-taubar}, above) and utility at least $f+2-\alpha$. This yields a set $P$ of $[\alpha-\taubar-(g-f)]\cdot|L|$ such proposals in expectation, and at least
$\frac12 [\alpha-\taubar-(g-f)]\cdot|L|$ with high probability. Then we will seek to count how many hospitals in $I$ receive at least $\kappa$ proposals from the proposals in $P$ and show that this count is sufficiently large with
high probability. Similarly to the hospital removal in \Cref{sec::doc-int-idio-no-match}, we will remove all doctors that have two or more edges to hospitals in $I$.
The next lemma bounds the number of removed doctors.

\begin{claim}\label{clm::doc-removal-lots-edges}
Suppose that the following events do not occur: $\Bhsf(I)$, $\Bhsm(I)$, $\Bdbm(L)$, $\Bdbf(L)$, and $\Bhbf(C(d))$ for all $d$ in $L$.
Then at most $|L|\exp\big(-\alpha k/2(4a+1)\big)$ doctors are removed because they have two or more edges to hospitals in $I$, 
with failure probability at most $n^{-(c+2)}$, if 
$2(c+2)\exp\big(\alpha k/2(4a+1)\big)\cdot \frac{\ln n}{n}\le (g-f)\le \frac{2(2a-1)^2\alpha }{(4a+1)\cdot(e\kappa k)^2}\cdot\exp\big(-\alpha k/2(4a+1)\big)$.
% and $(2a-1)\alpha \exp\big(-\alpha k/\big)\ge (c+2)\cdot \frac{\ln n}{n}$.
\end{claim}
\begin{proof}
Let $d$ be a doctor in $L$, the intersection of the cones for hospitals in $I$.
The probability that an edge from $d$ is to a hospital in $I$ is at most $|I|/[(a-\frac12)\cdot(\alpha n/\kappa)]$, as $\Bhbf(C(d))$ does not occur for any doctor in $I$'s cone. Therefore the probability that $d$ has two or more edges to $I$ is at most
$\frac12 k^2 \cdot \big[\frac{2|I|}{(2a-1)\cdot(\alpha n/\kappa)}\big]^2$.
% $\le 2\Big[\frac{k(g-f)\kappa}{(a-\frac12)\alpha}\Big]^2$, as $\Bdsm(I)$ does not occur.

Therefore the expected number of doctors with two or more edges to $I$ is at most 
$\frac12|L| \cdot \big[\frac{2k|I|\kappa}{(2a-1)\alpha n}\big]^2$.
Note that $|L|\le (2a+\frac12)\alpha n$, as $\Bdbm(L)$ does not occur.
By a Chernoff bound, there are more than $|I|\exp\big(-\alpha k/\big)$ such doctors with probability at most 
\begin{align*}
    \Big[\frac{(4a+1)e|I|(\kappa k)^2\exp\big(\alpha k/2(4a+1)\big)}{(2a-1)^2\alpha n}\Big]^{|I|\exp\big(-\alpha k/2(4a+1)\big)}.
    \end{align*}
    This is maximized when $|I|$ is as small as possible, i.e.\ $\frac12 (g-f)\frac{n}{\kappa}$, as $\Bhsf(I)$ does not occur; so it is at most
\begin{align*}
    \Big[\frac{(4a+1)e[(g-f)\kappa k^2\exp\big(\alpha k/2(4a+1)\big)}{2(2a-1)^2\alpha}\Big]^{\frac12 (g-f) \frac n{\kappa}\exp\big(-\alpha k/2(4a+1)\big)}.
\end{align*}
This is bounded by $n^{-(c+2)}$ if
$(g-f)\le \frac{2(2a-1)^2\alpha }{(4a+1)\cdot(e\kappa k)^2}\cdot\exp\big(-\alpha k/2(4a+1)\big)$ and $g-f\ge 2\kappa(c+2)\exp\big(\alpha k/2(4a+1)\big)\cdot \frac{\ln n}{n}$.

Finally, we need $|I|\exp\big(-\alpha k/2(4a+1)\big)\le |L|\exp\big(-\alpha k/2(4a+1)\big)$.
Assuming $\Bhsm(I)$ and $\Bdbf(L)$ do not occur, it suffices that
$2(g-f)n\le (a-\frac12)\alpha \frac n{\kappa}$, i.e.\ $g-f\le \frac14(2a-1)\frac{\alpha}{\kappa}$, and this is subsumed by the first upper bound on $g-f$.
\end{proof}
Let $L'$ be the set of doctors remaining after the removals from $L$, and left
$\Bdrem$ be the event that more than $|L|\exp\big(-\alpha k/2(4a+1)\big)/(e^2 \kappa)$ doctors are removed for any of the intervals $L$ that occur in the analysis.
Going forward, we suppose that $\Bdrem$ does not occur.

We suppose the event $\Btau$, that some doctor fails to have $k$ edges with private value $1-\taubar$ or larger does not occur, as specified in \Cref{clm::value-of-taubar}.

The proof that $\qtilde(v)\ge q(v)$ needs some modifications.
Again the possible runs of the double-cut DA will use some additional stops.
Doctors with two or more edges to $I$ are eliminated.
Of the remaining doctors with a single edge to $I$, if it provides utility in the range
$[f+2-\alpha,g+2-\alpha)\cup [f+2-\taubar,g+2)$ we do not include the proposal to $I$ in our count of proposals. This ensures that for the remaining doctors with a proposal to $I$, the receiving
hospital is uniformly distributed over $I$.

Once again, the tree $T$ of possible runs of the double-cut DA has
both real and imaginary proposals. For each doctor $d$ outside $L'$ we fix the utilities.
Also, for each doctor $d$ in $L'$, we fix its top $k$ private values, but without specifying which edges these values are attached to.
Now, we are seeking the probability that the doctors in $L'$ propose to at most $r$ distinct hospitals in $I$ for an $r$ to be specified later; in addition, we want to show that $\qtilde(v) \ge q(v)$ for all nodes $v$ in $T$, and for $\treeroot(T)$ in particular.

A node $v$ in $T$ corresponding to an action by a doctor outside $L'$ has just one outgoing edge corresponding to its single possible action. For a doctor $d$ in $L'$, we proceed as in \Cref{sec::proof-hosp-match-prob}.
We consider the edges out of $d$ one at a time.
The decision making corresponding to one edge proceeds in two steps:
first, whether it has a utility equal to $d$'s current utility value in tree $T$. This is a decision based on making a random choice for its interview value. For each possible destination of this edge, this occurs with the same probability. 
In the $q$ setting, we then choose the edge uniformly from the hospitals in $d$'s current cone that she has not already proposed to; in this setting, there are no imaginary proposals.
A doctor stops making proposals either when she has gone through all $k$ of her proposals or she makes a proposal to a hospital in $I$.

We turn to the $\qtilde$ setting. 
Again, as soon as a doctor makes a proposal, real or imaginary, to a hospital in $I$, she stops.
We keep track of how many proposals a doctor has made; 
let $k(d,v)$ be the number of remaining proposals for doctor $d$ at node $v$. 
Also, recall that $u(d,v)$ is her residual utility at node $v$.
Again, let $x$ be the number of hospitals in $I'$ that have received real proposals, and $y$ the number that have received imaginary proposals.
At node $v$, we will be seeking the probability that at most an additional $r-x$ hospitals in $I'$ receive real proposals.
The doctors with the $s-x-y$ smallest products $u(d,v)\cdot k(d,v)$ will attempt to propose; each real proposal will occur with the same probability, which will be no larger than the probability of a proposal in the $q$ setting; the difference in probability will be handled by making proposals imaginary with probability equal to the difference. Now, we make this precise.

We use the same value for the probability that $d$ has a proposal with the current utility $u$ as in the $q$ setting. 
If $d$ is determined to have such a proposal, uniformly at random, we choose one of the hospitals in $d$'s current cone that has not already received a proposal from $d$ to receive the proposal; then we choose whether it is real with probability $\tau=\frac{k(d,v)}{(2a+\frac12)\alpha n}$.
As in \Cref{sec::proof-hosp-match-prob}, $\tildep(\wrl(h))\le p(\wrl(h))$; again, we choose
$\tildep(\wim(h))=\wrl(h))-\tildep(\wrl(h))$.

\begin{claim}\label{clm::qtilde-larger-for-intervals-in-hosp-prop}
    $\qtilde(v)\ge q(v)$ for all $v$ in $T$.
\end{claim}
\begin{proof}
The argument is very similar to that in \Cref{clm::qtilde-bound-for-hosp-doc-int}.
    The base case argument and case 2 are unchanged. Case 1 is the same except that $v$ will have two children $\wrl(h)$ and $\wim(h)$ for each $h\in I$.

\end{proof}

For the next claim, which bounds the surplus, the following additional notation will be helpful.
$J=[g+(a+1)\alpha,1)\cup [g-a\alpha,f+a\alpha)$ is the set of participating doctors minus those in $[g+a\alpha,g+(a+1)\alpha)$,
and $K=[g-\alpha,1)\setminus I$ is the set of participating hospitals excluding those in $I$.
   
\begin{claim}\label{clm::surplus-interval-doc-prop}
Suppose that neither of the following events occurs: $\Bdbf(J)$, $\Bhbm(K)$.
    Then the surplus $s\ge \alpha n$, if $\alpha\cdot\exp\big(\alpha k/2(4a+1)\big)\ge 3$, and $a\ge 5$.
\end{claim}
\begin{proof}
We need to consider two cases, depending on whether the span of the cone for some doctor in $I$ has  overlaps the range $[1-\alpha,1)$.
\\
    \emph{Case 1}. $g\le 1-(a+1)\alpha$.\\
        As $\Bdbf(J)$ does not occur, $J$ contains at least $\big[1-g-(a+1)\alpha)+(2a\alpha- [g-f]) -\frac 12\alpha\big]n= \big[(1-g) +(a-\frac32)\alpha -(g-f)\big]n$ doctors.
        By assumption, at most a fraction $3\cdot \exp\big(-\frac{\alpha k}{2(4a+1)}\big)$ of the doctors in $[g+(a+1)\alpha,1)$ are not fully matched.
        % (strictly, the fraction in a minimal set of intervals covering $[g+a\alpha,1)$).
        By \Cref{clm::doc-removal-lots-edges}, the number of doctors in $L$ that are eliminated is at most $|L|\exp\big(-\alpha k/2(4a+1)\big)$.

        Thus the number of doctors remaining is at least $\big[(1-g) +(a-\frac32)\alpha -(g-f)\big]n\cdot \big[1 -3\cdot \exp\big(-\frac{\alpha k}{2(4a+1)}\big)\big] \ge \big[(1-g) +(a-\frac52)\alpha -(g-f)\big]n$, as $\exp\big(-\alpha k/2(4a+1)\big)\le \frac 13 \alpha$.
              
        As $\Bhbm(K)$ does not occur, there are at most $[1-g+\frac32\alpha-(g-f)]\frac{n}{\kappa}$ participating hospitals outside of $I$.
        Therefore, the surplus, the number of doctors in $[g-a\alpha,f+a\alpha)$ not matched to $K$,
        is at least 
    $\big(\big[1-g)+(a-\frac52)\alpha-(g-f)\big] -\big[1-g+\frac32\alpha-(g-f)\big]\big)n\ge (a-4)\alpha n\ge \alpha n$.
            \\
    \emph{Case 2}. $g> 1-a\alpha$.\\
    As in Case 1, $K$ contains at most $\big[1-g+\frac32\alpha-(g-f)\big]n\le \big[(a+\frac32)\alpha -(g-f)\big]\cdot \frac{n}{\kappa}$ hospitals.
     Now, $J=[g-a\alpha,f+a\alpha)$.
    Similarly to Case 1, the number of doctors in $[g-a\alpha,f+a\alpha)$ that are available for matching
    is at least $\big[(2a-\frac32)\alpha-(g-f)\big]n$.
    So the surplus is at least
    $\big(\big[(2a-\frac32)\alpha-(g-f)\big] - \big[(a+\frac32)\alpha -(g-f)\big] \big)\cdot\frac {n}{\kappa}
    \ge (a-3)\frac \alpha n\ge \alpha n$.
\end{proof}

\begin{claim}\label{qtilde-bound-at-root-doc-to-hosp-int}
Let $I=[f,g)$, and suppose that $s\ge \alpha n$, 
$g-f \ge 3\exp\big(\frac{\alpha k}{2(4a+1)}\big) \cdot (c+2)\frac{\kappa\ln n}{n}$,
$g-f+\taubar\le \frac {\alpha}{2}$, and $\alpha k\ge 2(4a+1)$.
Also, suppose that $\Bhsf(I)$ not occur.
Then the probability that the number of not fully matched hospitals is more than
    $3|I|\cdot \exp\big(-\frac{\alpha k}{2(4a+1)}\big)$ is at most $2n^{-(c+2)}$.
\end{claim}
\begin{proof}
We seek to identify as small an $r$ as possible such that $\qtilde(\treeroot(T))\le n^{-(c+2)}$,
for we can then deduce that $q(\treeroot(T))\le n^{-(c+2)}$ for this value of $r$.

In the $\qtilde(\treeroot(T))$ setting, $s\ge \alpha n$ doctors each seek to propose to an arbitrary  hospital $h$ in $I$ with probability at least
\begin{align*}
[\alpha-(g-f)-\taubar]\cdot \frac {2\kappa k}{(4a+1)\alpha n}
 \ge \frac{\alpha}{2}\cdot \frac {2\kappa k}{(4a+1)\alpha n},
\end{align*}
using $g-f+\taubar\le \frac {\alpha}{2}$.

In expectation, $h$ receives at least $\frac{\alpha \kappa k}{4a+1}$ proposals.
By a Chernoff bound, the probability that $h$ receives fewer than $\kappa$ proposals is at most
$\exp\big(-\alpha k/[2(4a+1)]\big)$ if $\kappa=1$, and if $\kappa\ge 2$ it is at most
$\exp\big(-\frac12\big[1 -\frac{4a+1}{\alpha k}]^2\cdot \frac{\alpha \kappa k}{4a+1}\big)$;
If $\frac{4a+1}{\alpha k}\le 1-\frac{1}{\sqrt{2}}$, then the latter probability is also bounded
by $\exp\big(-\frac{\alpha  k}{2(4a+1)}\big)$, and the condition amounts to
$\alpha k\ge (4a+1)\cdot(2+\sqrt{2})$.

\hide{
In expectation, the hospitals in $I$ receives at least $|I|\cdot\frac{\alpha\kappa k}{2(4a+1)}$ proposals. 
By a Chernoff bound, $I$ receives fewer than $\frac 12 |I|\cdot\frac{\alpha\kappa k}{2(4a+1)}$ proposals with probability at most
$\exp\big(- \frac{|I|\alpha\kappa k}{16(4a+1)} \big)\le n^{-(c+2)}\big)$,
if $|I|\alpha\kappa k \ge 16(4a+1)\cdot(c+2)\ln n$. 
As $\Bhsf(I)$ does not occur, $|I|\ge \frac 12 (g-f)\frac{n}{\kappa}$,
and so $g-f\ge \frac{32(4a+1)}{\alpha k}\cdot(c+2)\frac {\ln n}{n}$ suffices.
We set $r+1=  |I|\cdot\frac{\alpha\kappa k}{}$.
So we know there are this many proposals to $I$ in run of the double-cut DA with failure probability at most $n^{-(c+2)}$,
and furthermore each proposal is uniform over $I$.
We now analyze how many hospitals in $I$ receive fewer than $\kappa$ proposals.
}

Let $X_h$ be the 0/1 random variable that equals 1 if $h$ receives fewer than $\kappa$ proposals.
The $X_h$ are negatively correlated and therefore we can apply Chernoff bounds to their sum.
In expectation, the sum of the $X_h$ is at most
$|I|\cdot \exp\big(-\frac{\alpha k}{2(4a+1)}\big)$, and by a Chernoff bound, it is more than
$2|I|\cdot \exp\big(-\frac{\alpha k}{2(4a+1)}\big)$
with probability at most $\exp\big(-\frac 23 |I|\cdot \exp\big(-\frac{\alpha k}{2(4a+1)}\big)\le n^{-(c+2)}$, if $|I| \ge \frac 32 \exp\big(\frac{\alpha k}{2(4a+1)}\big) \cdot (c+2)\ln n$.
As $\Bhsf(I)$ does not occur, $|I|\ge \frac 12 (g-f)\frac{n}{\kappa}$,
so $g-f \ge 3\exp\big(\frac{\alpha k}{2(4a+1)}\big) \cdot (c+2)\frac{\kappa\ln n}{n}$ suffices.
%
%Finally, looking at the two lower bounds on $g-f$, the second one dominates so long as $\alpha k\ge 8(4a+1)$.
\end{proof}

\begin{claim}\label{clm::doc-prop-interval-constraints}
    The conditions in Claims~\ref{clm::doc-removal-lots-edges}--\Cref{qtilde-bound-at-root-doc-to-hosp-int} amount to the following constraints:
    \begin{align*}
        \alpha\cdot \exp\Big(\frac {\alpha k}{2(4a+1)}\Big)&\ge 3~~~~(\text{Claims 
         \ref{clm::surplus-interval-hosp-prop} and \ref{clm::surplus-interval-doc-prop}}))\\
        \taubar &\ge \frac{b\kappa}{(a-\frac12)\alpha}\cdot(c+2)\cdot\frac {\ln n}{n}~~~~(\Cref{clm::value-of-taubar})\\
        b& >2~~~~(\Cref{clm::value-of-taubar})\\
        k &\le \Big[1 -\Big(\frac 2b\Big)^{1/2}\Big]\cdot b(c+2)\ln n~~~~(\Cref{clm::value-of-taubar})\\
        \alpha k &\ge ~~~~(\Cref{qtilde-bound-at-root-doc-to-hosp-int})\\
        a & \ge 5~~~~(\ref{clm::surplus-interval-doc-prop})
    \end{align*}
    and several constraints on $g-f$, namely:
    \begin{align*}
       & 
        3\exp\Big(\frac{\alpha k}{2(4a+1)}\Big)\cdot \kappa\cdot(c+2)\cdot\frac{\ln n}{n}
        \le g-f ~~~~(\text{Claims \ref{clm::doc-removal-lots-edges} and \ref{qtilde-bound-at-root-doc-to-hosp-int}})\\
        &\hspace*{1in}\le \min\Big\{\frac{\alpha}{2}-\taubar,
        ~\frac{(2a-1)^2\alpha\cdot\exp\big(-\frac{\alpha k}{2(4a+1)}\big)}{(2a+\frac12)\cdot(e\kappa k)^2} \Big\}~~~~(\text{resp.\ Claims \ref{qtilde-bound-at-root-doc-to-hosp-int}, \ref{clm::doc-removal-lots-edges}})\\
       &\hspace*{1in}\le \frac{(2a-1)^2\alpha\cdot\exp\big(-\frac{\alpha k}{2(4a+1)}\big)}{(2a+\frac12)\cdot(e\kappa k)^2},~~~~~~~~~~~~\text{if $~\taubar\le \frac{\alpha}{4}$ and $k\ge \frac{\sqrt{2a-1}}{e}$}.
    \end{align*}
    In addition, over the whole induction, they exclude events with total probability at most $\big(2n+ \frac{9}{g-f}\big)\cdot n^{-(c+2)}$.
\end{claim}
\begin{proof}
    By inspection, 
the excluded events in \Cref{clm::doc-removal-lots-edges} have probability at most $\big(n+\frac{4}{g-f}\big)\cdot n^{-(c+2)}$ and it has failure probability at most $\frac{1}{g-f}\cdot n^{-(c+2)}$;
there are no additional excluded events in \Cref{clm::value-of-taubar} and it has failure probability at most $n\cdot n^{-(c+2)}$;
the additional excluded events in \Cref{clm::surplus-interval-doc-prop} have probability at most
$\frac{2}{g-f}\cdot n^{-(c+2)}$; 
there are no additional excluded events in \Cref{qtilde-bound-at-root-doc-to-hosp-int} and this claim has failure probability at most $\frac {2}{g-f}\cdot n^{-(c+2)}$;  
 The bound on the probability of excluded events follows readily.
\end{proof}

\subsubsection{Proofs for \Cref{sec::idio-eps-Nash}, the $\eps$-Nash Equilibrium}

\begin{claim}\label{clm::final-cstrnts-idio}
 With $\delta=0$, $\alpha=\big(\frac{\ln k}{k}\big)^{1/2}$, $a=5$, and $b=3$, the constraints in the above claims are satisfied if 
\begin{align*}
     %   1&\le \Big(\frac{\ln k}{k}\Big)^{1/2}\cdot\exp\Big(\frac{(2k\ln k)^{1/2}}{18}\cdot\min\{1,\kappa/2\}\Big)\\
        O(1)&\le k \le \frac 12(c+2)\ln n \notag\\
         \frac{n}{\ln n}& \ge\frac{(c+2)k^{2.5}\kappa^2}{3(\ln k)^{1/2}}\exp\big(\frac{(k\ln k)^{1/2}}{4}\big).
    \end{align*}
\end{claim}
\begin{proof}
We begin by showing the following constraints suffice and then substitute the given values for $\alpha$, $a$ and $b$.
    \begin{align}
         3&\le \alpha\cdot\exp\Big(\frac{\alpha k}{2(4a+1)}\Big)\notag\\
         \frac{8(4a+1)}{\alpha}&\le k \le \frac 12(c+2)\ln n \label{eqn::alpha-cstrnt}\\
         \alpha&\ge \Big[\frac{8b(c+2)\kappa}{(2a-1)}\cdot\frac{\ln n}{n}\Big]^{1/2} \notag\\
         \frac{n}{\ln n}& \ge\frac{3(c+2)\cdot(4a+1)e^2k^2\kappa^2}{2(2a-1)^2\alpha}\cdot\exp\Big(\frac{\alpha k}{2(4a+1)}\Big). \notag
    \end{align}

We set $\taubar = \frac{2b(c+2)\kappa \ln n}{(2a-1)\alpha n}$,
and $\frac14\alpha \ge \taubar$,
which yields the constraint $\alpha^2\ge \frac{8b(c+2)\kappa}{(2a-1)}\cdot\frac{\ln n}{n}$.
The lower and upper bounds on $g-f$ in \Cref{clm::hosp-prop-interval-constraints},
yields the constraint $\frac{n}{\ln n} \ge \frac{6(c+2)e^2k^2\kappa}{2(2a-1)\alpha}\cdot\exp\big(\frac{\alpha k}{2(4a+1)}\big)$.
Similarly, the bounds on $g-f$ in \Cref{clm::doc-prop-interval-constraints} yields the constraint
$\frac{n}{\ln n}\ge \frac{3(c+2)\cdot(4a+1)e^2k^2\kappa^2}{2(2a-1)^2\alpha}\cdot\exp\big(\frac{\alpha k}{2(4a+1)}\big)$.
The second bound is the larger.
Finally, the condition $n \ge \frac {2\kappa k}{(4a+1)\alpha}$ is subsumed by the condition on $n/\ln n$.
\end{proof}
{\sc Comment.} A larger bound on $k$ would be possible if we allowed $b$ to be larger. Also, we point out that for sufficiently large $n$, our constants are unduly pessimistic, as the Chernoff bounds we use are far from asymptotically tight. \footnote{We chose to use the bounds in the paper for simplicity and because for the values of $n$ that occur in the NRMP, the values of $n\le 10,000$, and in this range the bounds we use are more reasonable.}

{\sc \Cref{clm::prob-B}}~
\emph{Let ${\mathcal B}$ denote the bad events excluded in Claims~\ref{clm::conditions-hosp-prop-one-doctor}, \ref{clm::conditions-doctor-prop-one-hosp}, \ref{clm::hosp-prop-interval-constraints}, \ref{clm::doc-prop-interval-constraints}.
The probability that ${\mathcal B}$ occurs is at most $12n^{-(c+1)}$.
}
\begin{proof}
    Summing the probabilities of the bad events in these four claims gives a bound of
    $\big(7n+\frac{16}{g-f}+\frac 3n\kappa\big)n^{-(c+2)}\le 12n^{-(c+1)}$,
    as $g-f\ge 6(c+2) \exp\big(\frac{\alpha k}{2(4a+1)}\big)\cdot \frac{\ln n}{n}> \frac 6n$.
\end{proof}

{\sc Claim \ref{lem::typical-student-eps-nash}}
\emph{Suppose the constraints in \Cref{clm::final-cstrnts-idio} hold.
    Let $d$ be a doctor with public rating at least $a\alpha$, and let $\alpha=\big(\frac{\ln k}{k}\big)^{1/2}$.
        In expectation, ex-ante, $d$ could improve her utility by at most 
        $O(\alpha)=O((\ln k/k)^{1/2})$, if $k\ln k\ge \big(\frac83\big)^2\cdot\frac{4a+1}{\kappa^2}$.
        }
\begin{proof}
We will show the improvement in utility is at most $\frac kn +24n^{-(c+1)}+ 3k^{-4}+2a\alpha =O(\alpha)$ in expectation.

In expectation, the smallest chosen private value for a doctor has value $1- \frac kn$.
%$\taubar\le \frac14\alpha$ is a parameter such that all the doctors' private values are at least $1-\taubar$ with high probability. We bound $\taubar$ in \Cref{clm::value-of-taubar} and in the proof of \Cref{clm::final-cstrnts-idio}.
If $\mathcal B$ does not occur, $d$ could improve her expected private value by at most $\frac kn$.
If $\mathcal B$ occurs, $d$ could improve her expected outcome by at most $\frac 32+\frac kn$ units of utility; 
by \Cref{clm::prob-B}, $\mathcal B$ occurs with probability at most $12n^{-(c+1)}$.
This accounts for the first two terms.

By \Cref{lem::doctors-high-value-match-prob}, $d$ fails to obtain a match providing utility $r(d)+2-\taubar-\psibar_d$ with probability at most $\exp\big(-\frac{k\alpha\psibar_d}{4a+1}\big)$ ($\taubar\le \frac14\alpha$ is a parameter such that all the doctors' private values are at least $1-\taubar$ with high probability. We bound $\taubar$ in \Cref{clm::value-of-taubar} and in the proof of \Cref{clm::final-cstrnts-idio}).
Setting  $\psibar_d=\alpha$, yields a probability of at most $\exp\big(-\frac{\ln k}{4a+1}\big)\le k^{-4}$.  Therefore, as $\alpha \exp\big(-\frac{2(4a+1)}{\ln k}\big)\le 3$, changing her proposals to improve her match probability (not that it is clear how to do this), yields an expected gain of at most $3k^{-4}$, which accounts for the third term in the bound.

When changing her bids, in expectation, there is no improvement in the interview value.

When bidding below her cone the public rating of potential matches only decreases.
When bidding in cone, the public rating can improve by
at most $2a\alpha$, which accounts for the fourth term in the bound. 

Finally, if $d$ switches to one or more proposals above her cone, we show that the expected gain is fairly small.
In expectation, there is no improvement in the interview value. So any gain comes from an improved public rating. A move from the bottom of the cone to the top increases the public rating by up to $2a\alpha$. A further increase of $x$ to the public rating
has a match probability of
at most $\exp\big(-\frac{3\kappa k\alpha x}{8(4a+1)}\big)$, as can be seen by setting $\psibar_h= x$ in \Cref{lem::hosp-high-value-match-prob}. For $x\ge 0$, the expected improvement due to the increase by $x$, $x\cdot \exp\big(-\frac{3\kappa \alpha x}{8(4a+1)}\big)$, is maximized at $x=0$, if $k\alpha \ge 8(4a+1)/3\kappa$,
i.e.\ if $k\ln k\ge \big(\frac83\big)^2\cdot\frac{4a+1}{\kappa^2}$.  This is no better than the gains in-cone.
\end{proof}

\subsection{The Student-School Setting}

\begin{lemma}\label{lem::doctors-high-value-match-prob-noid}
Let $\delta>0$ be an arbitrarily small constant.
    Let $d$ be a doctor with public rating at least $a\alpha$.
    Suppose we run the hospital-proposing double-cut DA for doctor $d$ using the discretized utilities
as specified in \Cref{clm::discrete-dist-idio}.
    Suppose the fraction of unmatched hospitals with public rating in $[r(d)+(a+1)\alpha,1)$ is at most $3\exp\big(-\alpha\kappa k/2(4a+1)\big)$.
    Also, suppose that the event ${\mathcal B}$ does not occur.
    Finally, suppose we run the hospital-proposing double-cut DA for $d$.
     Then the probability that $d$ fails to receive a match for which it has interview value $1-\psibar_d$ or larger is at most $\exp\big(-\frac{4k\psibar_d}{4a+1}\big)$, if both $a\ge 5$ and $\alpha\cdot\exp\big(\alpha k/2(4a+1)\big)\ge 3$.
\end{lemma}

The final lemma is a simplified version of \Cref{lem::hosp-high-value-match-prob}, obtained by setting
$\psibar_h=0$.
\begin{lemma}\label{lem::hosp-high-value-match-prob-noid}
Let $\delta>0$ be an arbitrarily small constant.
    Let $h$ be a hospital with public rating at least $a\alpha$.
    Suppose we run the doctor-proposing double-cut DA for hospital $h$ using the discretized utilities as specified in \Cref{clm::discrete-dist-idio}.
    Suppose the fraction of unmatched doctors with public rating in $[r(d)+(a+1)\alpha,1)$ is at most $3\exp\big(-\alpha k/2(4a+1)\big)$.
    Also, suppose that the event ${\mathcal B}$ does not occur.
    Finally, suppose we run the doctor-proposing double-cut DA for $h$.
     Then the probability that $h$ fails to obtain $\kappa$ matches is at most $\exp\big(-\frac{3\kappa k\alpha}{8(4a+1)}\big)$, if $a\ge 5$, $k\ge 4a+1$, and $\alpha\cdot\exp\big(3\alpha k/2(4a+1)\big)\ge 3$.
\end{lemma}

\subsubsection{Proof of \Cref{lem::doctors-high-value-match-prob-noid}}\label{sec::proof-doc-match-prob-noid}

We fix the utilities for all edges other than those incident on $d$.
\Cref{clm::matching-bounds-extend} continues to hold here.
The analysis from the resident setting simplifies.

\begin{claim}\label{lem::surplus-hosp-prop-one-doc-noid}
    Suppose the following bad events do not occur:
    $\Bdbm([r(d),1))$ and $\Bhbf([r(d)-a\alpha,r(d)+a\alpha)\cup[r(d)+(a+1)\alpha,1))$, and also suppose that both $a\ge 5$ and $\alpha\cdot\exp\big(3\alpha\kappa k/8(4a+1)\big)\ge 3$.
    Then the surplus $s\ge 2\alpha n/\kappa$.
\end{claim}
\begin{proof}
    As $\Bdbm([r(d),1))$ does not occur, the number of doctors with public rating at least $r(d)-\alpha$ is at most $[1-r(d)+\frac 12\alpha ]n$. 
    By the assumption in \Cref{lem::doctors-high-value-match-prob-noid} and by \Cref{clm::matching-bounds-extend}, the number of not-fully matched hospitals in $[r(d)+(a+1)\alpha,1)$ is at most
    $3\frac n{\kappa}\cdot\exp\big(-\alpha\kappa k/2(4a+1)\big)\le  \alpha \cdot \frac n{\kappa}$, if $\exp\big(- \alpha\kappa k/2(4a+1)\big)\le \frac 13\alpha$.
    In addition, some or all of the hospitals in $[r(d)+a\alpha,r(d)+(a+1)\alpha)$ may not be fully matched.
    As $\Bhbf([r(d)-a\alpha,r(d)+a\alpha)\cup[r(d)+(a+1)\alpha,1))$ does not occur,
    the number of hospitals with public rating in this range is at least $[1-r(d)+(a-\frac32)\alpha] n/\kappa$.
    Thus the surplus, the number of not-fully matched available hospitals in $d$'s cone, i.e., hospitals that could propose to $d$, is at least
    \begin{align*}
        \big[1-r(d)+(a-\tfrac32)\alpha\big] \cdot\frac{n}{\kappa} - \big[1-r(d)+\tfrac 12\alpha\big]\cdot\frac{n}{\kappa} -  \alpha \cdot\frac{n}{\kappa}
        \le \big(a - 3\big)\alpha\cdot \frac{n}{\kappa}\ge 2\alpha \cdot \frac{n}{\kappa},
    \end{align*}
    if $a\ge 5$.
\end{proof}

The stochastic dominance argument is not needed here. A simple direct argument suffices.

\begin{claim}\label{noid-doctor-no-match-prob}
    Suppose the surplus $s\ge 2\alpha \frac{n}{\kappa}$, and $\Bhbm(C(d))$ does not occur.
    Then the probability that $d$ receives no proposal with interview value
    at least $1-\psibar_d$ is at most $\exp\big(-4k\psibar_d/(4a+1)\big)$.
\end{claim}
\begin{proof}
  Consider a run of the double-cut DA for $d$. It ends with at least $s$ hospitals in $C(d)$ having either run through all their proposals or finishing with a proposal to $d$. Since the edges from $d$ are uniformly randomly distributed over $C(d)$, the probability that none of the surplus hospitals has an edge to $d$ with interview value
    at least $1-\psibar_d$ is at most $\big(1 - \frac{\psibar_d \cdot s}{(2a+\frac12)\alpha n/\kappa}\big)^{k}$, as $\Bhbm(C(d))$ does not occur. This probability is at most $\exp\big(-4k\psibar_d/(4a+1)\big)$.
\end{proof}

\begin{claim}\label{clm::conditions-hosp-prop-one-doctor-noid}
 The conditions in Claims~\ref{lem::surplus-hosp-prop-one-doc-noid} and~\ref{noid-doctor-no-match-prob} amount to the following constraints: $a\ge 5$,  and $\alpha\cdot\exp\big(\alpha k/2(4a+1)\big)\ge 3$.
 They also exclude events with total probability at most $3n^{-(c+1)}$.
 \end{claim}

\subsubsection{Proof of \Cref{lem::doctors-interval-non-match} for the School Setting}
%\label{sec::doc-int-noid-no-match}
%
The proof is very similar to the proof in the residency setting.
We slightly reduce the set of participating doctors to be those with public rating at least $f$.
The effect is to increase the surplus (but we are going to ignore the slightly improved bound we could obtain).
The main change is in the definition of $\qtilde(v)$. Where before we choose $s-x-y$ hospitals with the smallest residual utilities, now we simply choose any $s-x-y$ hospitals that have not yet attempted to propose to $I$, i.e.\ they still have utility at least $g$. Otherwise the argument is unchanged.

\subsubsection{The $\eps$-Nash Equilibrium for all but the Bottommost Doctors}

We have the same constraints as in the residency setting and the same bad events, so
Claims \ref{clm::final-cstrnts-idio} and 
\ref{clm::prob-B} continue to hold.

\begin{lemma}\label{lem::typical-student-eps-nash-noid}
Suppose the constraints in \Cref{clm::final-cstrnts-idio} hold.
    Let $d$ be a doctor with public rating at least $a\alpha$, and let $\alpha=\frac{2(4a+1)\ln k}{k}\cdot\max\{1,2/\kappa\}$.
        In expectation, ex-ante, $d$ could improve her utility by at most $\taubar +22n^{-(c+1)}+ 3/k^2+4a\cdot \frac{(4a+1)\ln k}{k}\cdot\max\{1,2/\kappa\}=O(\ln k/k)$, if $k\ge a+\frac14$.
\end{lemma}
\begin{proof}
If $\mathcal B$ does not occur, $d$ could improve her private value by at most $\taubar$.
If $\mathcal B$ occurs, $d$ could improve her outcome by at most $2+\taubar$ units of utility; 
by \Cref{clm::prob-B}, $\mathcal B$ occurs with probability at most $11n^{-(c+1)}$.
This accounts for the first two terms.

By \Cref{lem::doctors-high-value-match-prob-noid}, $d$ fails to obtain a match providing utility $r(d)+2-\taubar-\psibar_d$ with probability at most $\exp\big(-\frac{k\psibar_d}{4a+1}\big)$.
Setting  $\psibar_d=\alpha$, yields a probability of at most $\exp\big(-\frac{2(4a+1)\ln k}{4a+1}\big)\le 1/k^2$.  Therefore changing her proposals to improve her match probability (not that it is clear how to do this), yields an expected gain of at most $3/k$, which accounts for the third term in the bound.

When changing her bids, in expectation, there is no improvement in the interview value.

When bidding below her cone the public rating of potential matches only decreases.
When bidding in cone, the public rating can improve by
at most $2a\alpha=2a\cdot \frac{2(4a+1)\ln k}{k}\cdot\max\{1,2/\kappa\}$, which accounts for the fourth term in the bound. 

Finally, if $d$ switches to one or more proposals above her cone, we show that the expected gain is fairly small.
In expectation, there is no improvement in the interview value. So the gain comes from the improved public rating. A move from the bottom of the cone to the top increases the public rating by up to $2a\alpha$. A further increase of $x$ to the public rating
has a match probability of
at most $\exp\big(-\frac{4k x}{4a+1}\big)$, as can be seen by setting $\psibar_h= x$ in \Cref{lem::doctors-high-value-match-prob-noid}. For $x\ge 0$, the expected improvement due to the increase by $x$, $x\cdot \exp\big(-\frac{4kx}{4a+1}\big)$, is maximized at $x=0$, if $k \ge a+\frac14$.  This is no better than the gains in-cone.
\end{proof}

We note from the proof of \Cref{clm::final-cstrnts-idio} that the constraints in \eqref{eqn::alpha-cstrnt} suffice.
Substituting $\alpha =\frac{2(4a+1)\ln k}{k}\cdot\max\{1,2/\kappa\}$ yields the constraints in the next claim.

\begin{claim}\label{clm::final-cstrnts-noid}
    With $\alpha=\frac{2(4a+1)\ln k}{k}$, $a=5$, and $b=3$, constraints \Cref{eqn::alpha-cstrnt} and \Cref{lem::typical-student-eps-nash-noid} are satisfied if 
\begin{align*}
 %       1&\le \Big(\frac{\ln k}{k}\Big)^{1/2}\cdot\exp\Big(\frac{(4k\ln k)^{1/2}}{20}\Big)\\
        8&\le k \le \frac 12(c+2)\ln n \notag\\
         \frac{n}{\ln n}& \ge\frac{(c+2) k^{4}\kappa^2}{65\ln k}.
    \end{align*}
\end{claim}

This concludes the proof of \Cref{thm::eps-nash-no-idio-large-cap}.

Finally, we show that there is exactly one stable matching in this setting.
This is straightfoward, but does not appear to be widely known.

\begin{claim}\label{clm::unique-sm}
    In the student-school setting the stable matching is unique.
\end{claim}
\begin{proof}
    We will show that the school proposing DA yields the same outcome as the student proposing DA.
    First, consider the student-proposing DA. We start by having student 1 propose. Once students $1$ to $i$ are matched, student $i+1$ proposes. Note that every hospital prefers its current matches to student $i+1$. So all tentative matches are actually final matches.
    Now consider the school proposing DA. Note that the schools all order the students by decreasing public rating.
    Recall that we are free to choose which available school makes the next proposal.
    So we first have the schools propose to the rank 1 student. One school matches, namely the school preferred by student 1. This is the school student 1 would match with in the student-proposing DA.
    Then all non-fully matched schools propose to student 2. Once the first $i$ students have received proposals (it some cases there may be 0 proposals to a student), the not-fully matched schools propose the student $i+1$.
    The inductive claim is that the two runs of DA produce the same matches for the first $i$ doctors.
    But then whether the hospitals propose to doctor $i+1$, or doctor $i+1$ proposes to the hospitals, 
    they are going to obtain the same match, for in the first case, doctor $i+1$ selects her favored hospital among the available hospitals, and in the second case, this is her preferred choice among the available hospitals, which demonstrates the inductive claim.
\end{proof}

\subsection{Reducing the Weight of the Interview Values}\label{sec::app-full-int}

We obtain improved bounds in Claims \ref{lem::fail-prob-hospital-prop} and \ref{lem::fail-prob-doctor-prop}, which concern the non-match probability for doctors, and the not-fully matched probability for hospitals, resp.
We also obtain better bounds in Claims \ref{qtilde-bound-at-root-hosp-to-doc-int} and \ref{qtilde-bound-at-root-doc-to-hosp-int}; they concern, respectively, the collective mismatch rate for doctors and hospitals in an interval $I$.
This allows us to reduce the value of $\alpha$
for the $\eps$-Nash analysis in \Cref{lem::typical-student-eps-nash}, leading to smaller values of $\epsilon$. 
We begin by showing the improvement to the bound in \Cref{lem::fail-prob-hospital-prop}.

\begin{claim}\label{lem::fail-prob-hospital-prop-ext}
    Let $d$ be a doctor with public rating at least $a\alpha$.
    Suppose $s\ge \frac{\alpha n}{\kappa}$ and that
    there are at most $\frac12(4a+1)\frac{\alpha n}{\kappa}$ hospitals in $C(d)$. Then the probability that $d$ receives no proposal for which $d$ has interview value at least $1-\psibar_d$ is at most $\exp\big(-k\beta_h\psibar_d/(4a+1)\big)$, where $\beta_h = \min\{1, \frac{\alpha}{\nu_h}\}$.
\end{claim}
\begin{proof}
 The change to the proof of \Cref{lem::fail-prob-hospital-prop} is that now we only require the
 hospitals to have interview value at least $1- \beta_h$. The reason we use this bound is that
 the reduction to the hospital's utility when using this bound is at most $\alpha$, which is what
 is needed to ensure that all the proposals in the double-cut DA have utility above the cut point.
\end{proof}

Similarly, in \Cref{lem::hosp-high-value-match-prob}, the not-fully matched probability improves to
$\exp\big(-\frac{3\kappa k\beta_d\psibar_h}{2(4a+1)}\big)$, where $\beta_d=\min\{1,\frac{\alpha}{\nu_d}\}$. The bounds in Claims \ref{qtilde-bound-at-root-hosp-to-doc-int} and \ref{qtilde-bound-at-root-doc-to-hosp-int} improve to at most $3|I|\exp\big(-\frac{\beta_h k}{2(4a+1)}\big)$ unmatched doctors,
and at most $3|I|\exp\big(-\frac{\beta_d k}{2(4a+1)}\big)$ unmatched hospitals.

Next, we state the resulting improvements to the $\eps$-Nash result. Basically, $\alpha$ can shrink almost linearly in $(\nu_h\nu_d)^{1/2}$ (offset by a logarithmic factor) until it reaches a minimum value $\alpha_{\min}$, and $\eps=O(\alpha)$.

\begin{theorem}\label{clm::improved-eps}
% Let $\alpha_l = \big(\frac{8b(c+2)\kappa}{2a-1}\big)^{1/2}\cdot \big(\frac{\ln n}{n}\big)^{1/2}$, the minimum allowed value for $\alpha$, and let $\alphabar=\big(\frac{4(4a+1)\ln k}{k}\big)^{1/2}$, the value of $\alpha$ when $\nu_d=\nu_h=1$.
Suppose that the hospitals list the doctors they interview in preference order and all but the bottommost $a\alpha$ doctors follow the recommended strategy.
Then, all but the bottommost $a\alpha$ doctors achieve an $\eps$-Nash equilibrium,
with the following bounds, depending on the values of $\nu_d$ and $\nu_h$.
Let $\alpha_{\min}= 3\exp\big(-\frac{k}{4(4a+1)}\big)$.

Case 1: $\alpha_{\min}\le \alpha\le \nu_h\le \nu_d\le 1$. \\
Let $\sigma=\min\big\{1,\frac{3\kappa}{8}\big\}$.
With our choice for $\alpha$ given shortly, this amounts to 
$\nu_h\ge \frac {v_d(4a+1)}{k\sigma} \ln \frac{ek\sigma}{\nu_h\nu_d(4a+1)}$.
Let $x=\frac{\nu_h\nu_d(4a+1)}{k\sigma}$; then we set $\alpha = x^{1/2}\cdot\big[\ln (ex)^{1/2} \big]$.
Then $\eps= O(\alpha)$,
and the mismatch/not-fully matched rates for the doctors and hospitals are
$\exp\big(-\frac{\alpha k}{(4a+1)\nu_h}\big)$
and $\exp\big(-\frac{3\kappa k\alpha}{8(4a+1)\nu_d}\big)$, respectively.

Case 2: $\max\{\alpha_{\min},\nu_h\} < \alpha\le \nu_d\le 1$.\\
Then $\eps= O(\alpha)$, with
$\alpha \ge \frac{4(4a+1)\nu d}{k}\cdot \ln\big(\frac{3ek}{4(4a+1)\nu d}\big)$ and $3k\ge 4(4a+1)\nu_d$;
the mismatch/not-fully matched rates for the doctors and hospitals are
$\exp\big(-\frac{k}{4a+1}\big)$ and $\exp\big(-\frac{3\kappa k\alpha}{8(4a+1)\nu_d}\big)$.

Case 3: $\nu_h\le \nu_d\le 3\alpha$.\\
We set $\alpha = \alpha_{\min}$.
Then $\eps= O(\alpha_{\min})$ and
the mismatch/not-fully matched rates for the doctors and hospitals are
$\exp(-k/4(4a+1))=\frac 13\alpha_{\min}$ and $\exp\big(-\frac{\kappa k}{8(4a+1)}\big)$.
\end{theorem}
\begin{proof}
    Let $\beta_d=\min\big\{1,\frac{\alpha}{\nu_d}\big\}$ 
    and $\beta_h=\min\big\{1,\frac{\alpha}{\nu_h}\big\}$.
    To maintain the surplus bounds from the setting in which doctors select the interviews, the simple interview setting for short, we need to ensure that the fraction of unmatched doctors or hospitals is at most $\alpha$---see Claims \ref{lem::surplus-hosp-prop-one-doc}, \ref{lem::surplus-hosp-prop-one-hosp}, \ref{clm::surplus-interval-hosp-prop}, and \ref{clm::surplus-interval-doc-prop}; i.e.\ we need 
    $3\exp\big(-\frac{\beta_h\cdot k}{4(4a+1)}\big)\le \alpha$
    and $3\exp\big(-\frac{\beta_d\cdot k}{4(4a+1)}\big)\le \alpha$. As $\beta_h\ge \beta_d$,
    $\alpha \cdot \exp\big(\frac{\beta_d\cdot k}{4(4a+1)}\big) \ge 3$ suffices.
    Accordingly, when $\beta_d =\frac{\alpha}{\nu_d}$,  $\alpha \ge \frac{\nu_d\cdot 4(4a+1)}{k}\ln \frac{3k}{\nu_d\cdot4(4a+1)}$ suffices, if $3k\ge e\nu_d\cdot4(4a+1)$,
    and when $\beta_d=1$, $\alpha \ge 3\exp\big(-k/4(4a+1)\big)$ suffices.

    The value of $\eps$ in the $\eps$-Nash analysis (see \Cref{lem::typical-student-eps-nash}) is dominated by two terms: the size $2a\alpha$ of the cone, and the probability that the doctor's match has an interview value less than $1-\psibar_d$; to obtain $\eps=O(\alpha)$, we want the second term to be bounded by $O(\alpha)$. 

    To obtain the doctor non-high-quality match failure probability we set $\psibar_d=\beta_d$, giving a non-high-quality-match failure probability of at most $\exp\big(-\frac{k\beta_h\beta_d}{(4a+1)}\big)$, and we want this to be at most $2\alpha$ (high-quality means that the match has interview value at least $\psibar_d$).
    Likewise, for the hospital not-fully matched high-quality failure probability, we set $\psibar_h=\beta_h$, giving a not-fully matched high-quality failure probability of at most $\exp\big(-\frac{3\kappa k\beta_h\beta_d}{8(4a+1)}\big)$, which we also want to bound by $2\alpha$.
    Then $\alpha\exp\big(\frac{k\beta_h\beta_d}{(4a+1)}\big)\ge 1$  suffices.
    
    To complete the analysis we need to consider several cases.\\
    Case 1. $\alpha\le \beta_h\le \beta_d\le 1$.\\
    To satisfy the second condition
    $\alpha \ge \big(\frac{\nu_h\nu_d(4a+1)}{k\sigma})^{1/2}\cdot \big[\ln\big(\frac{ek\sigma}{\nu_h\nu_d(4a+1)})^{1/2}\big]^{1/2}$ suffices, so long as $k\ge \nu_u\nu_d(4a+1)\sigma$.
    Combining conditions yields the first bound.
\hide{    
    Then for the doctor high-quality mismatch condition, we want
    $\alpha \ge \big(\frac{\nu_d\nu_h(4a+1)\ln 1/\alpha}{k}\big)^{1/2}$;
    $\alpha \ge \big(\frac{\nu_d\nu_h(4a+1)}{k}\cdot \ln \frac{k}{\nu_d\nu_h(4a+1)}\big)^{1/2}$ suffices.
   Similarly, for the hospital mismatch condition, $\alpha \ge \big(\frac{3\kappa\nu_d\nu_h(4a+1)}{8k}\cdot \ln \frac{8k}{3\kappa\nu_d\nu_h(4a+1)}\big)^{1/2}$ suffices.
   }
\\
    Case 2. $\beta_h<\alpha \le \beta_d\le 1$.\\
    The first constraint, $\exp\big(\frac{k\alpha}{\nu_d\cdot 4(4a+1)}\big)\ge 3$ is now the tighter one, so $\alpha \ge \frac{4(4a+1)\nu d}{k}\cdot \ln\big(\frac{3ek}{4(4a+1)\nu d}$ suffices,
    if $3k\ge 4(4a+1)\nu_d$.
    \\
    Case 3: $\nu_h\le \nu_d\le 3\alpha$.\\
    Here it suffices to set $\alpha =3\exp\big(-\frac{k}{4(4a+1)}\big)$.
\end{proof}
\subsection{Further Experiments}

Next, we examine the effect of increasing $k$ from 5 to 12 for the setting in which doctors choose the interviews. For the top 75\% of the hospitals the matching rate becomes close to 100\% and the doctor matching rate also improves.

\begin{figure}[h]
    \begin{minipage}{0.5\textwidth}
        \includegraphics[width=\textwidth]{images/Doctor_Matching_Rate__Grouped___2000__5_5_0.3_DocFilter-False_Seed-42.png}
    \end{minipage}%
    \begin{minipage}{0.5\textwidth}
        \includegraphics[width=\textwidth]{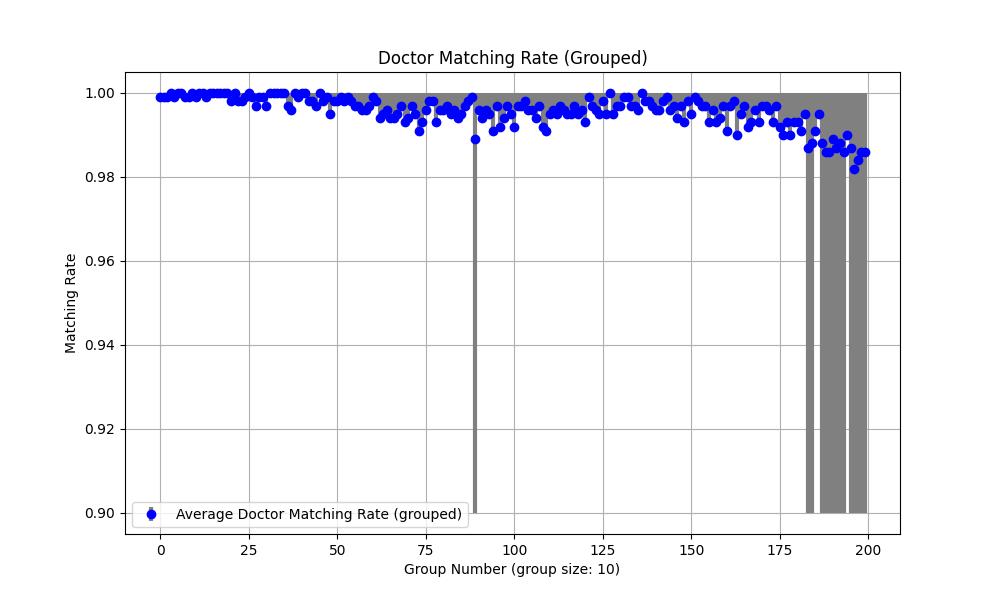}
    \end{minipage}
    \caption{Doctor matching rate comparison with $n=\text{2,000}$, $\kappa=5$, and $k=5$ (on the left) and $12$ (on the right).}
\end{figure}

\begin{figure}[h]
    \begin{minipage}{0.5\textwidth}
        \includegraphics[width=\textwidth]{images/Hospital_Matching_Rate__Grouped___2000__5_5_0.3_DocFilter-False_Seed-42.png}
    \end{minipage}%
    \begin{minipage}{0.5\textwidth}
        \includegraphics[width=\textwidth]{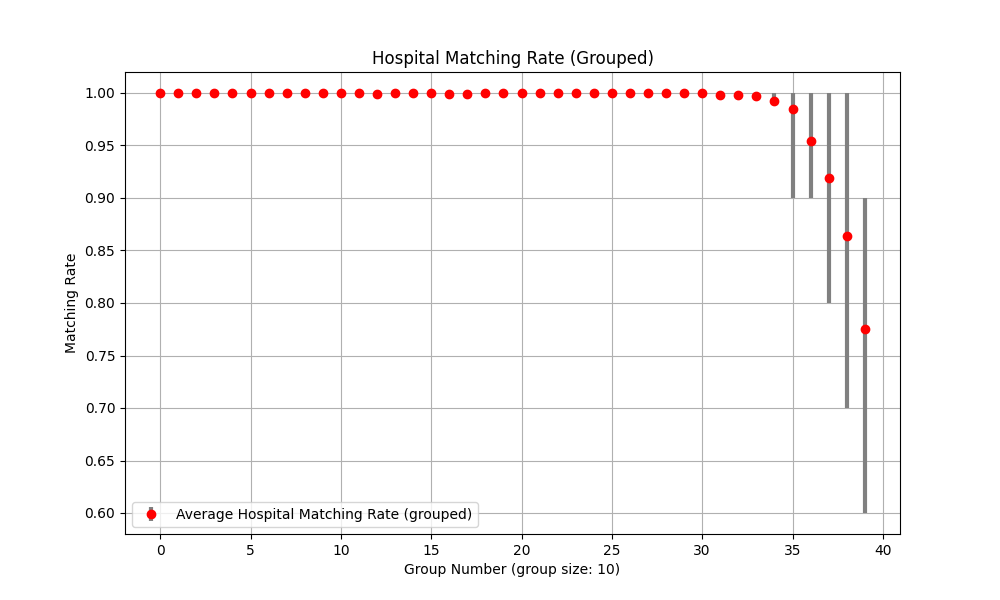}
    \end{minipage}
    \caption{Hospital matching rate comparison with $n=\text{2,000}$, $\kappa=5$, and $k=5$ (on the left) and $12$ (on the right).}
\end{figure}

\begin{figure}[h]
    \begin{minipage}{0.5\textwidth}
        \includegraphics[width=\textwidth]{images/Doctor_Loss__Grouped___2000__5_5_0.3_DocFilter-False_Seed-42.png}
    \end{minipage}%
    \begin{minipage}{0.5\textwidth}
        \includegraphics[width=\textwidth]{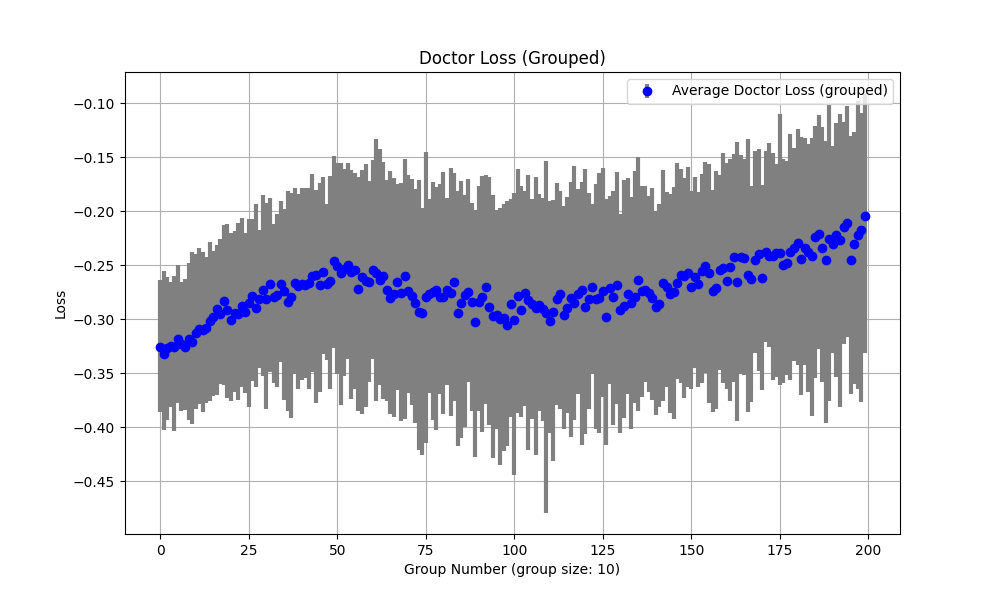}
    \end{minipage}
    \caption{Doctor loss comparison with $n=\text{2,000}$, $\kappa=5$, and $k=5$ (on the left) and $12$ (on the right).}
\end{figure}

\begin{figure}[h]
    \begin{minipage}{0.5\textwidth}
        \includegraphics[width=\textwidth]{images/Hospital_Loss__Grouped___2000__5_5_0.3_DocFilter-False_Seed-42.png}
    \end{minipage}%
    \begin{minipage}{0.5\textwidth}
        \includegraphics[width=\textwidth]{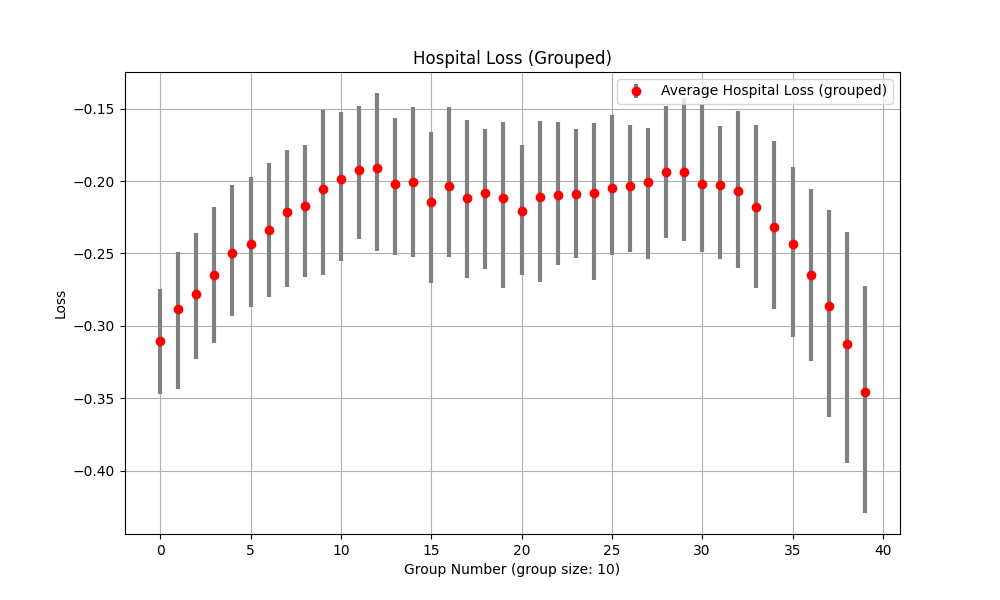}
    \end{minipage}
    \caption{Hospital loss comparison with $n=\text{2,000}$, $\kappa=5$ and $k=5$ (on the left) and $12$ (on the right).}
\end{figure}

\newpage

Next, we look at the effect of reducing $\kappa$ to 1, again in the setting with doctors choosing the interviews. The match rates become less good and the loss increases (note that the scales are not identical on the left hand side and right hand side graphs).

\begin{figure}[h]
    \begin{minipage}{0.5\textwidth}
        \includegraphics[width=\textwidth]{images/Doctor_Matching_Rate__Grouped___2000__5_5_0.3_DocFilter-False_Seed-42.png}
    \end{minipage}%
    \begin{minipage}{0.5\textwidth}
        \includegraphics[width=\textwidth]{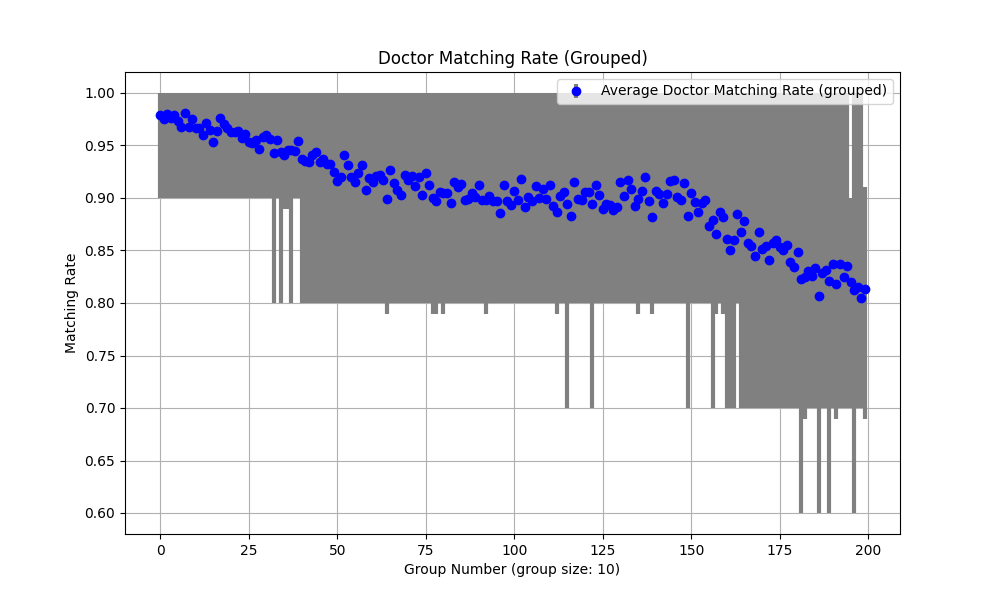}
    \end{minipage}
    \caption{Doctor matching rate comparison with $n=\text{2,000}$, $k=5$ and $\kappa=5$ (on the left) and $1$ (on the right).}
\end{figure}

\begin{figure}[h]
    \begin{minipage}{0.5\textwidth}
        \includegraphics[width=\textwidth]{images/Hospital_Matching_Rate__Grouped___2000__5_5_0.3_DocFilter-False_Seed-42.png}
    \end{minipage}%
    \begin{minipage}{0.5\textwidth}
        \includegraphics[width=\textwidth]{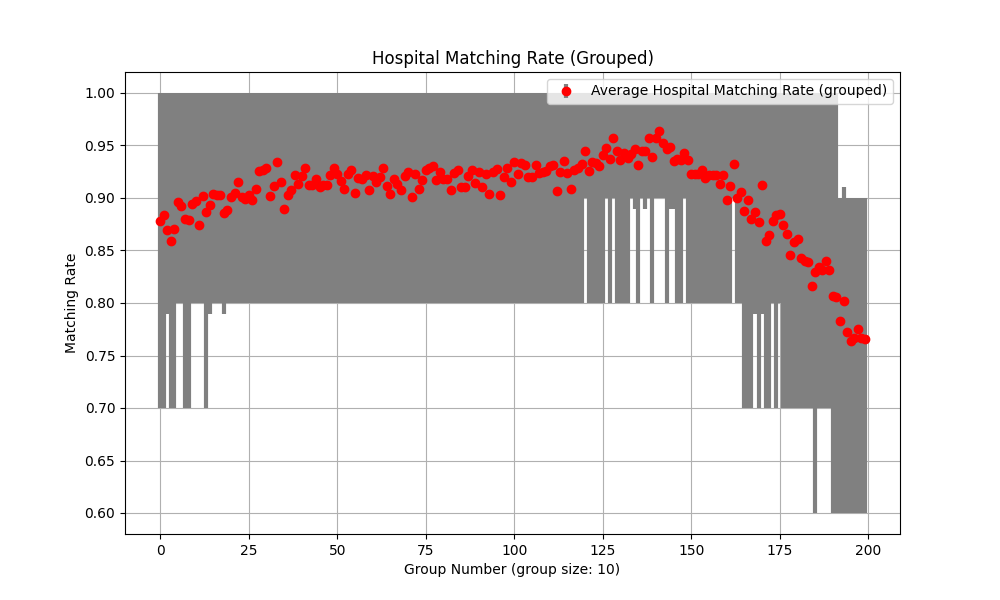}
    \end{minipage}
    \caption{Hospital matching rate comparison with $n=\text{2,000}$, $k=5$ and $\kappa=5$ (on the left) and $1$ (on the right).}
\end{figure}

\begin{figure}[h]
    \begin{minipage}{0.5\textwidth}
        \includegraphics[width=\textwidth]{images/Doctor_Loss__Grouped___2000__5_5_0.3_DocFilter-False_Seed-42.png}
    \end{minipage}%
    \begin{minipage}{0.5\textwidth}
        \includegraphics[width=\textwidth]{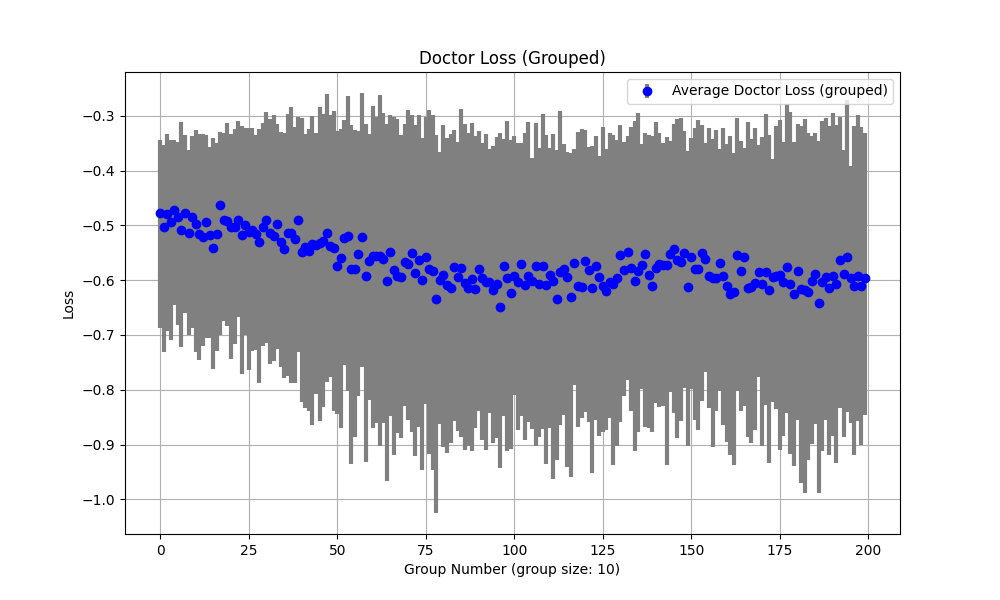}
    \end{minipage}
    \caption{Doctor loss comparison with $n=\text{2,000}$, $k=5$ and $\kappa=5$ (on the left) and $1$ (on the right).}
\end{figure}

\begin{figure}[h]
    \begin{minipage}{0.5\textwidth}
        \includegraphics[width=\textwidth]{images/Hospital_Loss__Grouped___2000__5_5_0.3_DocFilter-False_Seed-42.png}
    \end{minipage}%
    \begin{minipage}{0.5\textwidth}
        \includegraphics[width=\textwidth]{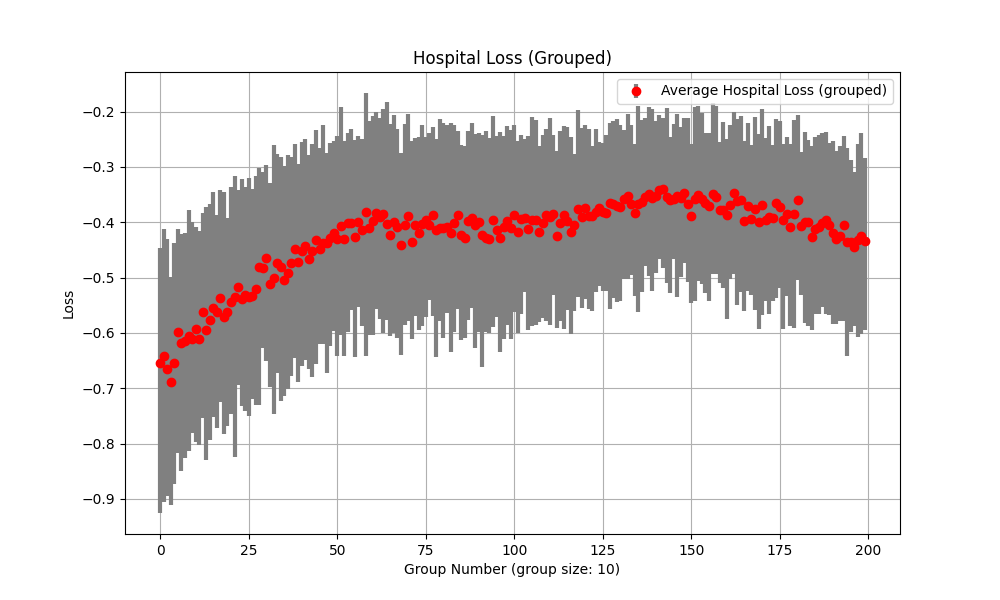}
    \end{minipage}
    \caption{Hospital loss comparison with $n=\text{2,000}$, $k=5$ and $\kappa=5$ (on the left) and $1$ (on the right).}
\end{figure}

\newpage\newpage

Finally, we show a selection of results for the case $n=500$, $\kappa=5$, $k=5$ and 12 and the doctor choosing the interviews.
The match rates overall are somewhat less good than for the case $n=\text{2,000}$ and the losses somewhat higher.

\begin{figure}[h]
    \begin{minipage}{0.5\textwidth}
        \includegraphics[width=\textwidth]{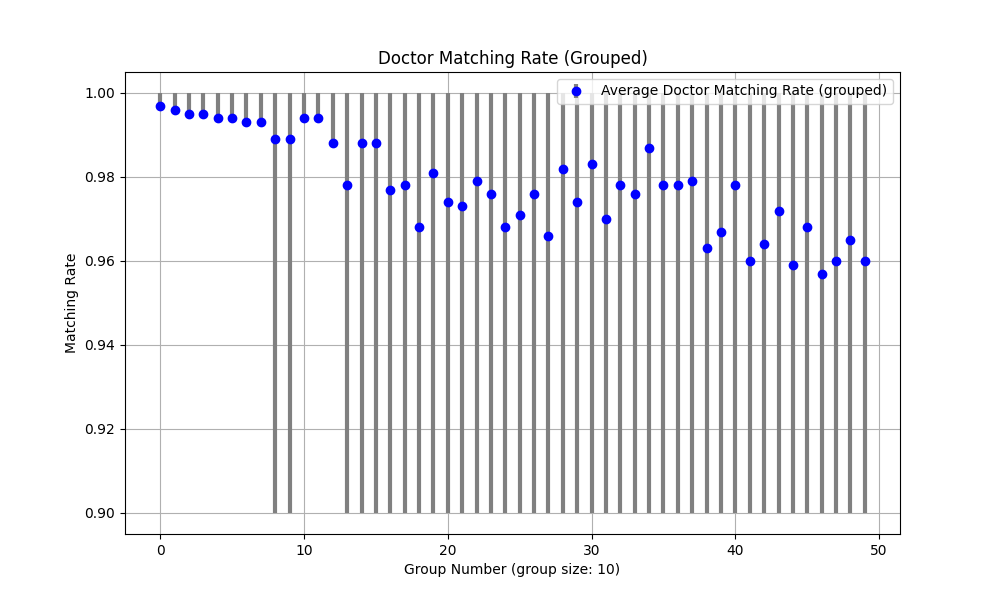}
    \end{minipage}%
    \begin{minipage}{0.5\textwidth}
        \includegraphics[width=\textwidth]{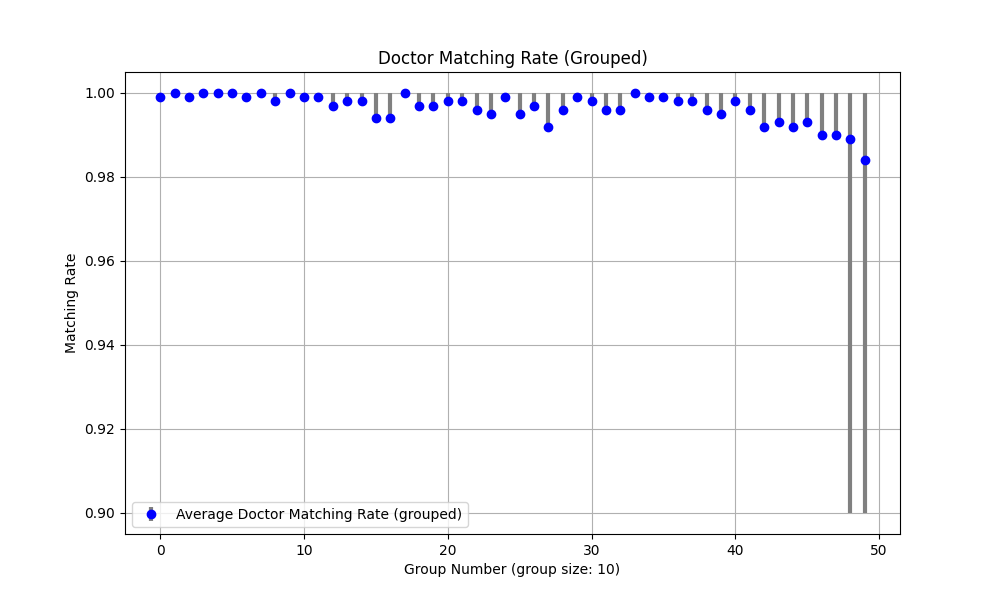}
    \end{minipage}
    \caption{Doctor matching rate, $n=500$, $\kappa=5$, $k=5$ (on the left), 12 (on the right).}
\end{figure}

\begin{figure}[h]
    \begin{minipage}{0.5\textwidth}
        \includegraphics[width=\textwidth]{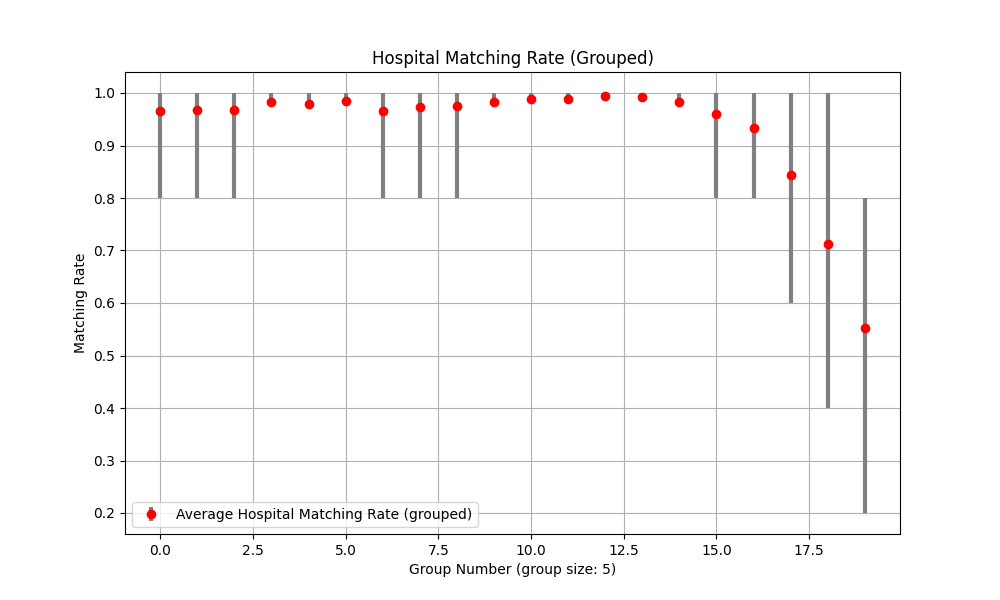}
    \end{minipage}%
    \begin{minipage}{0.5\textwidth}
        \includegraphics[width=\textwidth]{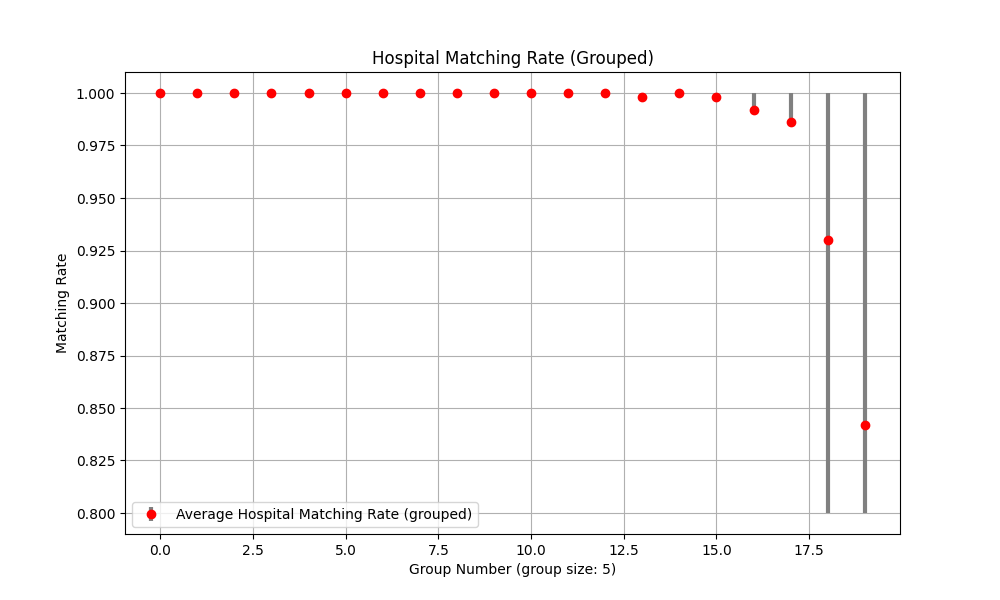}
    \end{minipage}
    \caption{Hospital matching rate $n=500$, $\kappa=5$, $k=5$ (on the left), 12 (on the right).}
\end{figure}

\begin{figure}[h]
    \begin{minipage}{0.5\textwidth}
        \includegraphics[width=\textwidth]{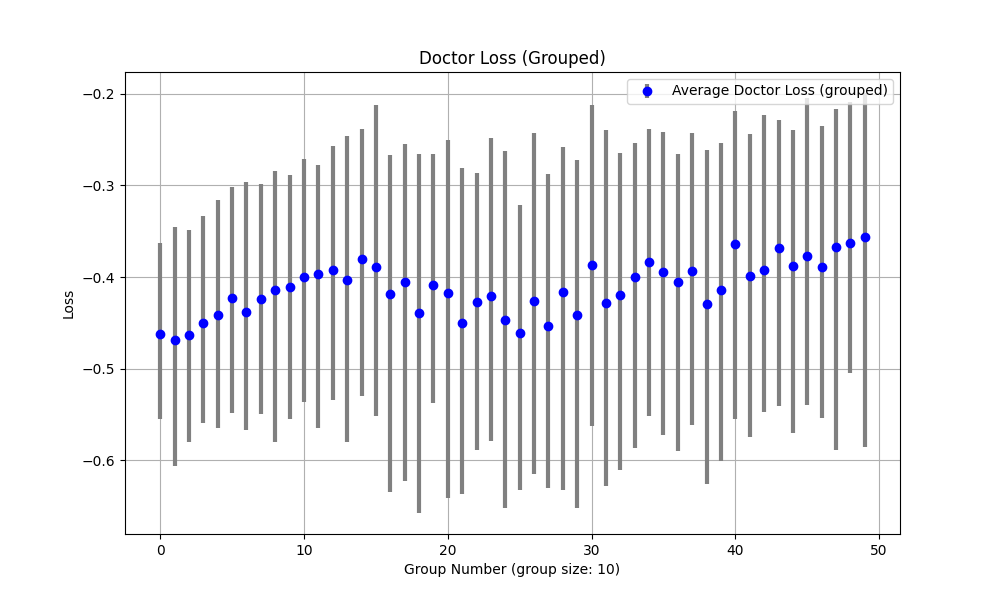}
    \end{minipage}%
    \begin{minipage}{0.5\textwidth}
        \includegraphics[width=\textwidth]{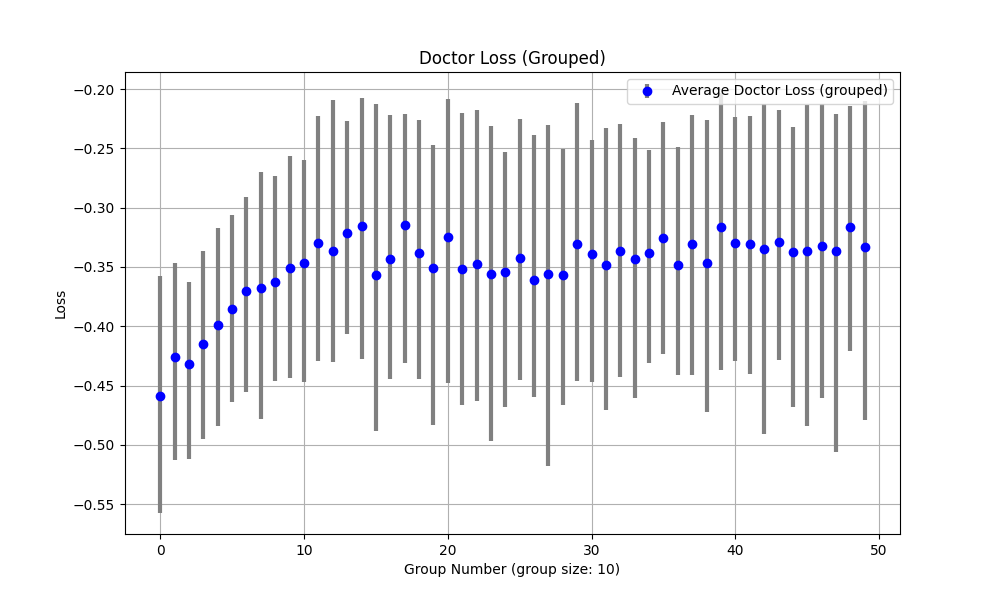}
    \end{minipage}
    \caption{Doctor loss $n=500$, $\kappa=5$, $k=5$ (on the left), 12 (on the right).}
\end{figure}

\begin{figure}[h]
    \begin{minipage}{0.5\textwidth}
        \includegraphics[width=\textwidth]{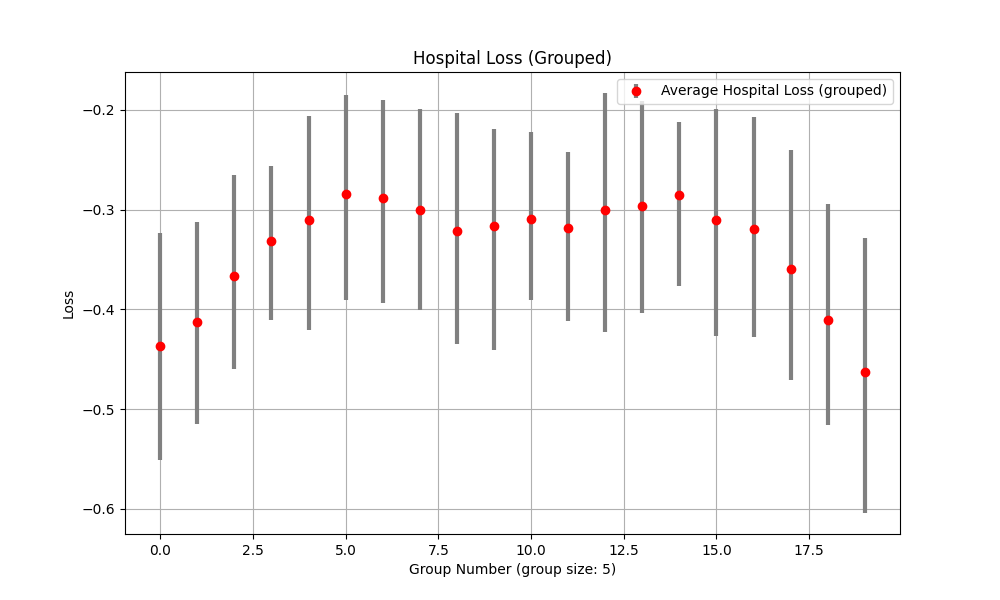}
    \end{minipage}%
    \begin{minipage}{0.5\textwidth}
        \includegraphics[width=\textwidth]{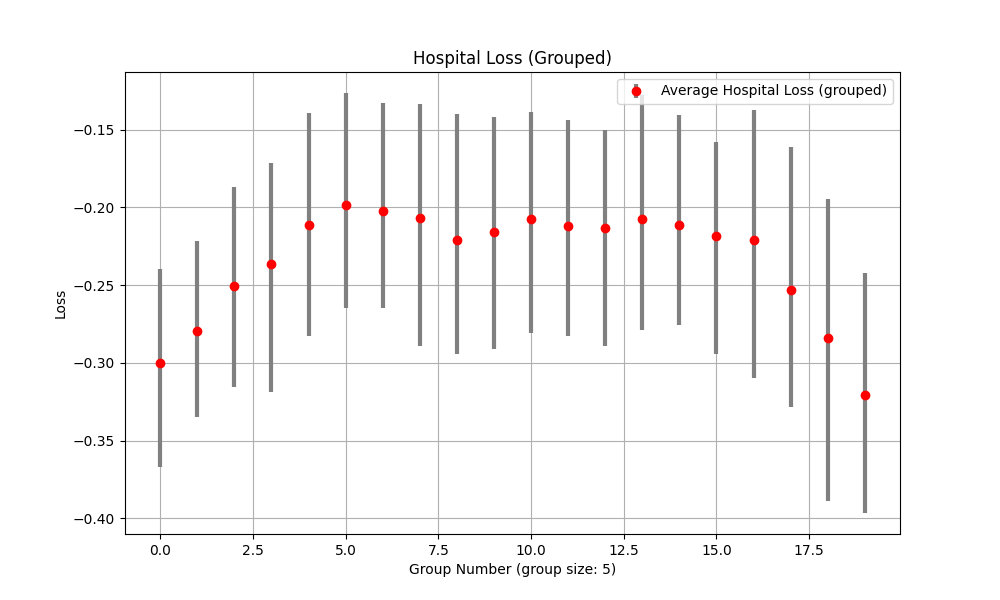}
    \end{minipage}
    \caption{Hospital loss $n=500$, $\kappa=5$, $k=5$ (on the left), 12 (on the right).}
\end{figure}

A final comment is that one should not expect a very close fit with the NRMP data, because there are many programs, small and large, in the NRMP data, and a significant fraction of doctors apply to more than one program. Nonetheless, at a high level, these results seem reasonably broadly consistent with the NRMP outcomes.

%\newpage\newpage
\clearpage

\bibliographystyle{plain}
\bibliography{sample-bibliography}

\end{document}